\documentclass[11pt,letterpaper]{article}

\usepackage{amsmath,amsthm,amsfonts,amssymb}
\usepackage{thm-restate}
\usepackage{fullpage}
\usepackage[utf8]{inputenc}
\usepackage[dvipsnames]{xcolor}
\usepackage{xspace,enumerate}
\usepackage[shortlabels]{enumitem}
\usepackage[hypertexnames=false,colorlinks=true,urlcolor=Blue,citecolor=OliveGreen,linkcolor=BrickRed]{hyperref}
\usepackage[capitalise]{cleveref}
\usepackage[OT4]{fontenc}
\usepackage{ifpdf}
\usepackage{todonotes}
\usepackage{enumerate}
\usepackage{thmtools}
\usepackage{algorithm}
\usepackage[noend]{algpseudocode}
\usepackage{algorithmicx}
\usepackage[mathlines]{lineno}

\usepackage[margin=1in]{geometry}
\def\poly{\operatorname{poly}}
\def\polylog{\operatorname{polylog}}

\newcommand{\email}[1]{\href{mailto:#1}{#1}}

\title{Max $s,t$-Flow Oracles and Negative Cycle Detection\\in Planar Digraphs}
\date{\vspace{-5ex}}
\author{Adam Karczmarz\thanks{University of Warsaw and IDEAS NCBR, Poland. \email{a.karczmarz@mimuw.edu.pl}. Partially supported by the ERC CoG grant TUgbOAT no 772346 and the National Science Centre (NCN) grant no. 2022/47/D/ST6/02184.}}

\newcommand{\tin}{\textrm{in}}
\newcommand{\tout}{\textrm{out}}

\newcommand{\Ot}{\ensuremath{\widetilde{O}}}
\newcommand{\eps}{\ensuremath{\epsilon}}
\newcommand{\dist}{\delta}
\newcommand{\len}{\ell}
\newcommand{\wei}{w}

\theoremstyle{plain}
\newtheorem{theorem}{Theorem}[section]
\newtheorem{lemma}[theorem]{Lemma}
\newtheorem{corollary}[theorem]{Corollary}

\newtheorem{fact}[theorem]{Fact}
\newtheorem{observation}[theorem]{Observation}
\newtheorem{definition}[theorem]{Definition}

\newtheorem{remark}[theorem]{Remark}

\newcommand{\rev}[1]{\ensuremath{#1}^{\mathrm{R}}}

\newcommand{\ds}{\ensuremath{\mathcal{D}}}

\newcommand{\exc}{\mathrm{exc}}
\newcommand{\totexc}{\Psi}

  \newcommand{\source}{\ensuremath{s}}
  \newcommand{\sink}{\ensuremath{t}}

\newcommand{\TG}{\mathcal{T}}

\newcommand{\DDG}{\mathrm{DDG}}

\newcommand{\rdiv}{\mathcal{R}}

\newcommand{\bnd}{\ensuremath{\partial}}

\begin{document}

\maketitle
  \thispagestyle{empty}

  \begin{abstract}
    We study the maximum $s,t$-flow oracle problem on planar \emph{directed} graphs
    where the goal is to design a data structure answering max $s,t$-flow value (or equivalently, min $s,t$-cut value) queries
    for arbitrary source-target pairs $(s,t)$.
    For the case of polynomially bounded integer edge capacities, we describe an \emph{exact} max $s,t$-flow
    oracle with truly subquadratic space and preprocessing, and sublinear query time.
      Moreover, if $(1-\eps)$-approximate answers are acceptable, we obtain a static oracle
    with near-linear preprocessing and $\Ot(n^{3/4})$ query time
    and a dynamic oracle supporting edge capacity updates and queries in $\Ot(n^{6/7})$ worst-case time.

    To the best of our knowledge, for \emph{directed} planar graphs, no (approximate) max $s,t$-flow oracles have been described even in the unweighted case,
    and only trivial tradeoffs involving either no preprocessing or precomputing all the $n^2$ possible answers have been known.

    One key technical tool we develop on the way
    is a sublinear (in the number of edges) algorithm for finding a negative cycle in so-called dense distance graphs.
    By plugging it in earlier frameworks, we obtain improved bounds for other fundamental problems on planar digraphs. In particular, we show:
    \begin{enumerate}[(1)]
    \item a deterministic $O(n\log(nC))$ time algorithm for negatively-weighted SSSP in planar digraphs with integer edge weights at least $-C$.
    This improves upon the previously known
        bounds in the important case of weights polynomial in $n$.
    \item an improved $O(n\log{n})$ bound on finding a perfect matching in a bipartite planar graph.
    \end{enumerate}
\end{abstract}

\clearpage
\setcounter{page}{1}

\section{Introduction}
Single-source shortest paths (with negative weights allowed) and maximum flow are among the most
fundamental computational problems on directed graphs.
The best known \emph{strongly polynomial} $\Ot(nm)$ time bounds for these problems~\cite{bellman1958routing, ford1956network, GalilN80, SleatorT83} have stood for decades.
The scaling framework \cite{EdmondsK72, Gabow85} made it apparent that this quadratic barrier can be overcome
in a relaxed \emph{weakly polynomial} model, where the magnitude $\log{C}$ of the largest numeric value $C$ (such 
as the maximum absolute weight/capacity of an edge) is also used to measure the size of an instance.
That is, a graph algorithm is weakly polynomial-time if it runs in time polynomial in $n$, $m$, and $\log{C}$.
Weakly polynomial-time algorithms typically assume integer data to guarantee that exact solutions are computed.
In particular, \cite{Gabow85} showed an $O(mn^{3/4}\log{C})$ algorithm for single-source shortest paths with integer weights
(and more general problems such as minimum-cost bipartite matching). That method has been subsequently refined
to run in time $\Ot(m\sqrt{n}\log{C})$~\cite{GabowT89, Goldberg95}. Later, a similar subquadratic bound
has been obtained also for the maximum flow problem by Goldberg and Rao~\cite{GoldbergR98}.
These algorithms, developed in the previous century, remain the best deterministic combinatorial algorithms
for these problems to date.

Recent years have brought impressive progress on single-source shortest paths and maximum flow
in the weakly polynomial regime with integer data. The usage of the interior-point method (IPM) for solving linear programs~\cite{Karmarkar84}
along with dynamic algebraic data structures proved very powerful for solving fundamental polynomial-time graph problems and
ultimately led to very efficient ${\Ot((m+n^{3/2})\log{C})}$\footnote{As usual, throughout we use the $\Ot(\cdot)$ notation to suppress $\polylog(n)$ factors.} \cite{BrandLLSS0W21} and  $O(m^{1+o(1)}\log^2{C})$~time~\cite{ChenKLPGS22,det-flow}
algorithms for
the minimum-cost flow problem, which generalizes both the SSSP and maximum flow problems.
Even more recently,~\cite{BernsteinNW22} developed a randomized algorithm for the negatively-weighted SSSP problem exclusively
that is purely combinatorial in nature and runs (with the tweaks of~\cite{sssp-logs} applied) in $O(m\log^2{n}\log{(nC)}\log\log{n})$ expected time.

The focus of this paper is on algorithms for planar digraphs. Planarity-exploiting near-optimal \emph{strongly polynomial} algorithms have been known
for both single-source shortest paths and maximum flow. 
Fakcharoenphol and Rao~\cite{FR06} gave the first near-optimal $O(n\log^3{n})$-time algorithm for negative cycle detection
and SSSP in planar digraphs.
This has been subsequently improved by~\cite{KleinMW10, MozesW10} and the current best known bound is
$O(n\log^2{n}/\log\log{n})$.
Maximum $s,t$-flow on planar digraphs can be computed in $O(n\log{n})$ time~\cite{BorradaileK09, Erickson10}.
The multiple-source multiple-sink case of max-flow is not easily reducible to the single-source single-sink
case while preserving planarity, and the best known upper bound for that case is $O(n\log^{3}{n})$~\cite{BorradaileKMNW17}.
All these state-of-the-art algorithms for planar graphs are combinatorial and deterministic.

Since planarity enables efficient strongly polynomial algorithms for shortest paths and maximum flow,
the study of \emph{weakly polynomial-time} algorithms
for these problems in the planar case has been rather limited. To the best of our knowledge, the combinatorial scaling framework has only
been applied to unit-capacity minimum-cost flow problems on planar graphs~\cite{AsathullaKLR20, KarczmarzS19, LahnR19}, but the achieved
bounds were still a polynomial factor away from linear. Recently, an algebraic IPM-based near-linear randomized
algorithm for planar min-cost flow has been proposed~\cite{DongGGLPSY22}.

In this paper we apply the combinatorial scaling framework to the maximum $s,t$-flow oracle problem 
and the negative-weights single-source shortest paths problem 
on planar digraphs.

\subsection{Our contribution}

\subsubsection{Max $s,t$-flow oracles}

We consider the following \emph{max $s,t$-flow oracle} problem. 
Given a planar digraph ${G=(V,E)}$, preprocess it
into a data structure supporting max $s,t$-flow value (or, equivalently, min $s,t$-cut capacity)
queries for arbitrary pairs $s,t\in V\times V$.
One extreme trivial solution for this problem is to
skip preprocessing and answer queries using the state-of-the-art $O(n\log{n})$-time
max $s,t$-flow algorithms for planar digraphs~\cite{BorradaileK09, Erickson10}.
Another extreme approach is to store all the possible $O(n^2)$ answers
and read the stored result in $O(1)$ time upon query.
\cite{LackiNSW12} showed that computing all-pairs max $s,t$-flow values
in planar graphs takes near-optimal $\Ot(n^2)$ time.
They also showed that in planar digraphs there might be $\Theta(n^2)$ distinct
max $s,t$-flow values (as opposed to \emph{general} undirected graphs, where there may be at most $n-1$ max-flow values) and asked whether constant-time queries are possible after subquadratic
preprocessing.
Nussbaum~\cite[Section~9.1.3]{nussbaum2014network} relaxed this open question
and asked whether \emph{sublinear} query time is possible after \emph{subquadratic} preprocessing.
Of course, ideally, we would like to support queries in $\Ot(1)$ time after $\Ot(n)$-time preprocessing.
This is in fact possible in \emph{undirected} planar graphs~\cite{BorradaileSW15}, where min $s,t$-cuts
have much more convenient structure compared to directed planar graphs.
However, for the directed case, no non-trivial preprocessing/query tradeoffs
have been described to date even for the case when the input graph is unweighted and/or
we only seek constant-factor approximate query answers.

First of all,
we give an affirmative answer to the open problem posed in~\cite{nussbaum2014network}
in the important case of integer weights which also covers unweighted graphs.
\begin{restatable}{theorem}{texactflow}\label{t:exact-flow}
  Let $G=(V,E)$ be a planar digraph with integer edge capacities in $[1,C]$ and let $r\in [1,n]$.
  In $\Ot(nr)$ time one can construct an $\Ot(nr)$-space data structure that,
  for any query pair $(s,t)\in V\times V$, computes the max $s,t$-flow (or min $s,t$-cut) value
  in $G$ in
  $\Ot(n^{3/2}/r^{3/4}\cdot \log^2{C})$ time.
\end{restatable}
Note that the data structure of Theorem~\ref{t:exact-flow} achieves both truly subquadratic
preprocessing/space and sublinear query time whenever $r=n^\delta$ satisfies $\delta\in (2/3,1)$.
It can also be used for computing max $s,t$-flow values for some $k$ specified source-sink pairs
faster than using the aforementioned trivial tradeoffs following from \cite{BorradaileK09, Erickson10, LackiNSW12}
unless $k=\Ot(n^{2/3})$ or $k=\widetilde{\Omega}(n^{5/4})$.
By suitably choosing the parameter $r=n^{2/7}k^{4/7}$, we obtain the following.
\begin{corollary}\label{t:multiple-pairs}
  Let $G=(V,E)$ be a planar digraph with integer edge capacities in $[1,C]$ and let\linebreak $k\in [n^{2/3},n^{5/4}]$
  be an integer.
  One can compute exact max $s_i,t_i$-flow values for $k$ source/sink pairs $(s_1,t_1),\ldots,(s_k,t_k)\in V\times V$
  in $G$ in $\Ot(n^{9/7}k^{4/7}\log^2{C})$ time.
\end{corollary}
It is worth noting that~\cite{LackiNSW12} in fact showed a near-linear time algorithm computing
max $s,t$-flow values for $n$ pairs $(s,t)$, where the source $s$ is \emph{fixed} (that is, common to all pairs) and $t\in V$
is arbitrary. In comparison,
the algorithm of Corollary~\ref{t:multiple-pairs} can handle $n$ \emph{arbitrary} source-sink pairs in $\Ot(n^{13/7})$ time in the
case of $\poly(n)$-bounded integer edge capacities.
Achieving subquadratic time for general integer-weighted sparse graphs is unlikely (conditionally on SETH) even
in the single-source all-sinks case~\cite{AbboudWY18, KrauthgamerT18}.

On the way to obtaining the exact oracle of Theorem~\ref{t:exact-flow}, we show a significantly
more efficient \emph{feasible $\lambda$-$s,t$-flow oracle}. The oracle is given a parameter $\lambda\in [1,nC]$
upon construction and is only required to support querying whether there exists
a flow between a specified source/sink pair of value precisely (or, equivalently, at least) $\lambda$.
Note that for $\lambda$ fixed to $1$, the problem is equivalent to the \emph{reachability oracle} problem
on planar digraphs, for which an optimal solution has been shown in~\cite{HolmRT15}.
In unweighted digraphs, if $\lambda$ is fixed to be an integer more than $1$, the problem is called the $\lambda$-reachability oracle problem.
A near-optimal $2$-reachability oracle for planar digraphs is known~\cite{ItalianoKP21}.
To the best of our knowledge, no non-trivial $\lambda$-reachability oracle for planar digraphs
and arbitrary $\lambda\geq 3$ has been described so far.
We show an oracle with \emph{near-linear} preprocessing (and space) and \emph{sublinear} query time
for an arbitrary parameter $\lambda$.

\begin{restatable}{theorem}{tfeasibleflow}\label{t:feasible-flow}
  Let $G$ be a planar digraph with integer edge capacities in $[1,C]$ and let $\lambda\in [1,nC]$
  be an integer. In $\Ot(n)$ time one can construct a data structure that
  can test, for any given query source/sink pair $(s,t)$ whether there exists an $s,t$-flow of
  value $\lambda$ in $G$ in $\Ot(n^{3/4}\log{C})$ time.
\end{restatable}

Using standard methods, the feasible flow oracle can be extended in two natural ways.
First, by using $O(\log{(nC)}/\eps)$ copies of the data structure,
it can be turned into an $(1-\eps)$-approximate max $s,t$-flow oracle
with near-linear preprocessing and sublinear query time. Formally, we have:

\begin{restatable}{corollary}{capproxoracle}
  Let $G$ be a planar digraph with integer capacities in $[1,C]$.
  In $\Ot(n\log{(C)}/\eps)$ time one can construct an $\Ot(n\log{(C)}/\eps)$-space data structure that,
  for any query source/sink pair $(s,t)$, estimates the max $s,t$-flow value
  in $G$ up to a factor $1-\eps$ in $\Ot(n^{3/4}\cdot \log{C}\cdot \log{(\log(C)/\eps)})$ time.
\end{restatable}

By following the same approach that has been previously used to convert the
static distance oracles
into dynamic distance oracles
for planar graphs~\cite{FR06, Klein05, KaplanMNS17}, we obtain a dynamic
approximate max $s,t$-flow oracle supporting edge capacity updates.

\begin{theorem}\label{t:dynamic-flow-approx}
  Let $G$ be a planar digraph with integer edge capacities in $[1,C]$.
  There exists a dynamic $(1-\eps)$-approximate max $s,t$-flow oracle
  with $\Ot(n^{6/7}\cdot \log{C}\cdot \log{(\log{(C)}/\eps)})$ query time that
  supports edge capacity updates in $\Ot(n^{6/7}\log{(C)}/\eps)$ worst-case time.
\end{theorem}
To the best of our knowledge, no other dynamic max $s,t$-flow oracles with sublinear update and query bounds
for planar \emph{directed} graphs have been described so far.
\emph{Exact} dynamic max $s,t$-flow oracles with sublinear update and query time
have been previously developed only for planar \emph{undirected} graphs~\cite{ItalianoNSW11}.
It is worth noting that \cite{Karczmarz21} gave a dynamic algorithm maintaining $(1-\eps)$-approximate
max $s,t$-flow value in a planar digraph for a single fixed pair $s,t$ in $\Ot(n^{2/3}\log{(C)}/\eps)$ amortized
time per update. Whereas the algorithm of~\cite{Karczmarz21} also exploits the connection between
the feasible flow problem and negative cycle detection, it crucially relies
on the fact that a \emph{single} value is maintained and its update bound is inherently amortized.

Let us remark that all the proposed max $s,t$-flow oracles can be easily extended
to report a certifying $s,t$-cut-set $F\subseteq E$ after answering the query (i.e., an $s,t$-cut-set of capacity less than $\lambda$
in the case of the feasible flow oracle, an $(1+\eps)$-approximate min $s,t$-cut-set in the approximate
oracle, and exact min $s,t$-cut-set in the exact oracle) in near-optimal $\Ot(|F|)$ extra time.

\newcommand{\ddg}{\ensuremath{\mathrm{DDG}}}
\paragraph{Reduction to negative cycle detection in DDGs.} The key tool that we develop to enable sublinear-time queries in max-flow oracles
is a new algorithm for finding a negative cycle in so-called \emph{dense distance graphs} (DDGs)~\cite{FR06} 
that constitute~a~fundamental concept in the area of planar graph algorithms. 

Roughly speaking, a dense distance~graph $\ddg(\rdiv)$ associated with an $r$-division $\rdiv$ with few holes~\cite{KleinMS13} of a plane digraph $G$ is a certain graph with $O(n/\sqrt{r})$ vertices and $O(n)$ edges allowing efficient shortest paths computations.
\cite{FR06} famously developed an implementation of Dijsktra's algorithm on a dense distance graph running in $\Ot(n/\sqrt{r})$ time, and an implementation of the Bellman-Ford algorithm (accepting also negative edges) running
in $\Ot(n^2/r)$ time. While both are faster than their respective general-graph implementations for any $r=\poly{n}$, the latter never runs in time sublinear in the graph size since $r\leq n$.

Based on scaling techniques~\cite{GabowT89, GoldbergHKT17, KarczmarzS19} we give a deterministic
algorithm for finding a negative cycle in $\ddg(\rdiv)$ in $\Ot(n^{3/2}/r^{3/4})$ time
in the case of polynomial integer weights (see Theorem~\ref{t:negcyc-rdiv}). Crucially, this bound is truly sublinear in $n$ for a sufficiently large~choice~of~$r$.

Our max $s,t$-flow oracles use negative cycle detection on dense distance graphs in a fairly black-box manner. As a result, obtaining a faster algorithm for that subproblem (ideally, matching the $\Ot(n/\sqrt{r})$ Dijkstra bound of~\cite{FR06}) -- formulated with no connection to flows -- will lead
to more efficient max $s,t$-flow oracles for planar digraphs. This is an especially promising further research direction
given the recent randomized near-linear combinatorial SSSP algorithm~\cite{BernsteinNW22,sssp-logs}.\footnote{A very recent manuscript~\cite{abs-2303-00811} generalizes~\cite{BernsteinNW22} into a black-box reduction of negatively-weighted SSSP on a digraph $G$ to non-negatively weighted SSSP instances on subgraphs of $G$ (not necessarily vertex-induced, and with edge weights potentially changed) of total size $mn^{o(1)}$. It is not clear how to apply that framework here, though. First, even describing the subgraphs might require $\Theta(m)$ space. Moreover, a dense distance graph with some edges removed does not, in general, preserve the properties leveraged by the efficient Dijkstra implementation of~\cite{FR06}.}

\subsubsection{Negatively-weighted SSSP}
Our new algorithm finding a negative cycle in a dense distance graph, plugged in the earlier framework~\cite{KleinMW10}, leads to 
the fastest known algorithm for negative cycle detection and
single-source shortest paths in integer-weighted planar digraphs. Formally, we prove:

\begin{restatable}{theorem}{tshpath}\label{t:shpath}
  Let $G$ be a planar digraph with integer edge weights in $[-C,\infty)$.
  One can either detect a negative cycle in $G$ or compute
  a feasible price function of~$G$ in $O(n\log{(nC)})$ time.
\end{restatable}
\noindent It is well-known that computing a feasible price function of a graph reduces the SSSP problem
to a non-negatively weighted instance, which can be in turn solved in $O(m+n\log{n})$ time
using Dijkstra's algorithm. For planar digraphs, an even faster linear-time algorithm is known~\cite{HenzingerKRS97}.
  
For the important case $C=\poly{n}$, our SSSP algorithm runs in $O(n\log{n})$ time. This improves upon the state-of-the-art
strongly polynomial
$O(n\log^2{n}/\log\log{n})$ bound of~\cite{MozesW10}.
Like all the previous shortest paths algorithms for planar graphs -- and in contrast to the recent
breakthrough combinatorial result for general graphs~\cite{BernsteinNW22} -- our algorithm is also deterministic.

As shown by~\cite{MillerN95}, routing a \emph{feasible flow} respecting given vertex demands in a plane digraph reduces to solving
a (negatively-weighted) SSSP problem in the dual graph. We thus obtain:

\begin{corollary}
  Let $G$ be a planar digraph with integer edge capacities no larger than $C$. Let $d:V(G)\to\mathbb{Z}$
  be a demand vector. Then, in $O(n\log{(nC)})$ time one can compute a flow function $f:E(G)\to \mathbb{Z}_{\geq 0}$ on $G$
  such that the excess of a vertex $v\in V(G)$ equals $d(v)$ (if such $f$ exists).
\end{corollary}
An important special case of the above feasible flow problem is that of finding a perfect matching in
a bipartite graph (use unit-capacity edges, and demands $\pm 1$ depending on the side of a vertex).
\begin{corollary}
  A perfect matching in a bipartite planar graph $G$ can be found in $O(n\log{n})$ time.
\end{corollary}
The previous best time bound for the planar bipartite perfect matching problem
was that of negative cycle detection, i.e., $O\left(n\frac{\log^2{n}}{\log\log{n}}\right)$~\cite{MozesW10}.
Our obtained bound also looks much cleaner.

\subsubsection{A more general framework for solving structured transportation-like problems.}
Even though all our developments are motivated by shortest paths and maximum flow applications in planar graphs,
on the way to proving Theorem~\ref{t:shpath}, we develop a fairly general framework for
solving the minimum-cost circulation problem (up to additive error $\delta$)
on networks $G$ with a relatively small (say $O(n)$) sum of vertex capacities 
and whose most of the defined edges are uncapacitated and structured in a way.
One can easily observe that this setting captures the transportation problem,
the minimum-cost perfect matching problem, and -- crucially for our applications -- the negative cycle detection
problem (as observed by~\cite{Gabow85}).
The goal is to be able to solve the problem in time that is \emph{sublinear}
in the number of uncapacitated edges defined implicitly.

This general idea has been previously used by Sharathkumar and Agarwal~\cite{SharathkumarA12}
to obtain an $\Ot(n^{3/2}\cdot \Phi(n)\cdot \log{(C/\delta)})$ algorithm (based on~\cite{GabowT89})
for $\delta$-additive approximate transportation problem
with sum of demands $O(n)$ on \emph{geometric networks} with diameter $C$,
where $\Phi(n)$ is the update/query time of a \emph{dynamic weighted nearest neighbor
data structure} on the input point set wrt. the used metric.
For example, if the points lie on the plane, have integer coordinates, and the metric
used is $L_1$ or $L_\infty$,~\cite{SharathkumarA12} give
an exact $\Ot(n^{3/2}\log{C})$-time algorithm for such a problem,
even though the network in consideration has $\Theta(n^2)$ defined edges.

We generalize this approach to arbitrary digraphs, where the edge costs
need not form a metric. 
Our framework is an extension of the minimum-cost circulation algorithms in~\cite{GoldbergHKT17, KarczmarzS19}.
We analyze the performance of the framework
wrt. the efficiency of dynamic \emph{closest pair} and \emph{near neighbor} graph data structures (defined precisely in Section~\ref{s:abstractions})
on the subgraph $G_\infty\subseteq G$ containing the uncapacitated edges.
Intuitively, networks $G_\infty$ with good graph closest pair data structures are those
that allow very efficient (e.g., sublinear in $|E(G_\infty)|$) implementations of Dijkstra's algorithm,
whereas fast near neighbor data structures allow very efficient computation
of blocking flows.
These data-structural abstractions are flexible enough to compose
them efficiently by taking graph unions or performing the vertex-splitting
transformation (Section~\ref{s:circ-overview}) which is crucial for applying the framework to the negative-cycle detection problem.
We note that the idea of analyzing the performance of graph algorithms
in terms of abstract dynamic closest pair data structures is not new and has been described before e.g. in~\cite{ChanE01}.
See Section~\ref{s:circ-overview} and Theorem~\ref{t:circ} for details.

We believe that our generalized min-cost circulation framework is of independent interest, and can be also
applied beyond planar graphs, e.g. in some more complex geometric settings.
To support this claim, we now give an example of a problem where our framework is applicable.
Suppose the network $G=(V,E)$ is a \emph{weighted unit disk graph} on $n$ points in the plane,
where every vertex
except the source $s$ and sink~$t$ has \emph{unit vertex capacity} and $||s-t||>1$.
For each $u,v\in V$, there is an implicit uncapacitated edge $uv$
with cost $||u-v||$ if $||u-v||\leq 1$. The vertices can also have costs assigned.
The goal is to compute the minimum cost $s\to t$ flow of a specified value $k\leq n$.
With our framework, and using known dynamic weighted nearest neighbor structures~\cite{AgarwalES99}\footnote{Which also imply dynamic closest pair data structures in a black-box way~\cite{Chan20}.}, we can compute
an $\delta$-additive approximation of the min-cost $s,t$-flow of value $k$ in $G$ in truly subquadratic
time (also dependent on $\log(1/\delta)$), even though $G$ may have $\Theta(n^2)$ edges defined implicitly.
\subsection{Further related work}
For general undirected graphs, the max $s,t$-flow oracle problem can be solved by computing
the so-called \emph{cut-equivalent tree} (or, the \emph{Gomory-Hu tree}) that encodes all the pairwise max-flow information~\cite{gomory1961multi}.
Recently, there has been a lot of effort to compute cut-equivalent trees efficiently (e.g.,~\cite{AbboudKT21, AbboudKT22, LiPS21, Zhang22}) which
culminated in an almost-optimal algorithm for unweighted graphs and an $\Ot(n^2)$-time algorithm on weighted graphs~\cite{AbboudK0PST22, ChenKLPGS22}.

For general directed graphs, there is a $n^{3-o(1)}$ conditional lower bound~\cite{KrauthgamerT18} for the problem
of computing all-pairs max-flows in weighted \emph{sparse} graphs, which almost matches
the obvious upper bound that follows from the recent near-optimal max-flow algorithms~\cite{ChenKLPGS22,det-flow}.

\section{Preliminaries}\label{s:prelims}
In this paper we deal with \emph{directed} graphs.
We write $V(G)$ and $E(G)$ to denote the sets of vertices and edges of $G$, respectively. We omit $G$ when the graph in consideration is clear from the context.
We write $e=uv\in E(G)$ when referring to edges of $G$.
We call $u$ the tail of $e$, and $v$ the head of $e$.

The graphs we deal with are \emph{weighted}, i.e., the edges have either weights or costs assigned (but not both).
Following the literature, when dealing with minimum cost flow problems, we stick to the
term \emph{cost}, and when talking about shortest paths or negative cycles, we use the term \emph{weight}.
We use $\wei_G(uv)$ to denote the weight of an edge $uv$, and $c_G(uv)$ to denote its cost.
The domain of possible edge weights/costs may vary, but will be always either stated explicitly or clear from the context.
The subscript is dropped when it is clear which graphs's weights/costs we are referring to.

The \emph{cost} or \emph{length} of the path $\len(P)$ is defined as $\len(P)=\sum_{i=1}^k c(e_i)$ or $\len(P)=\sum_{i=1}^k\wei(e_i)$,
depending on whether edges have costs or weights assigned.

The \emph{distance} $\dist_G(u,v)$ between the vertices $u,v\in V(G)$ is the length
of the \emph{shortest}, i.e., minimum-cost/minimum-length $u\to v$ path in $G$ or $\infty$, if no $u\to v$ path
exists in $G$.
Note that the distance is well-defined only if $G$ contains no negative cycles.

  We call any function $p:V\to \mathbb{R}$ a \emph{price function} on $G$.
  The \emph{reduced cost} (\emph{reduced weight}) of an edge $uv=e\in E$ wrt. $p$ is defined
  as $c_p(e):=c(e)-p(u)+p(v)$ ($\wei_p(e)=\wei(e)-p(u)+p(v)$, resp.).
  We call $p$ a \emph{feasible price function}\footnote{In the shortest paths literature~(e.g.,~\cite{Goldberg95}), the price of the edge's tail is usually added to $w(e)$, and the price of the head is subtracted. But in the min-cost flow literature~(e.g.,~\cite{GoldbergHKT17}), it is the other way around. We stick to the latter definition even when talking about shortest paths.} of $G$ if each $e\in E$
  has non-negative reduced cost/weight wrt.~$p$.

  It is known that $G$ has no negative-cost/weight cycles (negative cycles, in short) iff some feasible
  price function $p$ for $G$ exists.
  If $G$ has no negative cycles, distances in $G$ are well-defined.

  \begin{fact}\label{f:distanceto}
    Suppose $G$ has no negative cycles. Suppose $s$ can reach all vertices of $G$.
    Then the \emph{minus distance} function
    $p(v):=-\dist_{G}(s,v)$ is a feasible price function of~$G$.
  \end{fact}
\section{Data structural abstractions and negative cycle detection}\label{s:abstractions}
In this section we define the interface of abstract graph data structures
whose operations constitute the bottleneck of our min-cost circulation algorithm
that is in turn used to efficiently find negative cycles in dense distance graphs built upon $r$-divisions of planar graphs (Theorem~\ref{t:negcyc-rdiv}).
The missing proofs from this section can be found in Section~\ref{s:omitted}.

\subsection{Graph near neighbor data structure}\label{s:graph-near}
Let $G=(V,E)$ be a directed graph with edge costs given by $\wei:E\to \mathbb{R}$.
After some preprocessing of $G$, 
a \emph{graph near neighbor data structure} on $G$ has the following interface.
\begin{enumerate}[itemsep=-2pt]
  \item \emph{Initialization} with a subset $T\subseteq V$ and two mappings $p:T\to \mathbb{R}$ (vertex prices), and $\tau:V\to \mathbb{R}$ (vertex thresholds).
  \item \emph{Near neighbor query}: given $v\in V$, compute an incident edge $vt=e\in E$ with $t\in T$, such that
    \begin{equation*}
      \wei(e)+p(t)<\tau(v).
    \end{equation*}
    If no such $e$ exists, $\perp$ is returned.
  \item \emph{Deactivation} of $t\in T$ removes $t$ from the subset $T$.
\end{enumerate}
The \emph{total update time} of the graph near neighbor data structure is defined
as the total time needed for initialization and processing an arbitrary sequence of
deactivations resulting in $T=\emptyset$.
The \emph{query time} is the time needed for processing a near neighbor query.

\newcommand{\nn}{{\mathrm{nn}}}
\newcommand{\cp}{\mathrm{cp}}
\newcommand{\fin}{{\mathrm{fin}}}

For a graph $H$, we denote by
$T_\nn(H)=\Omega(|V(H)|)$ its total update time and by $Q_\nn(H)=\Omega(1)$ its query time.
We assume that these bounds hold only after suitably preprocessing $H$ (possibly given implicitly)
in time $P_\nn(H)$. The reason why preprocessing is not included in the total update time
is that the preprocessing can be shared by many instantiations of the data structure.

\begin{restatable}{observation}{sspdummyclo}
  For a digraph $G=(V,E)$ with $n$ vertices and $m$ edges given explicitly,
  there exists a near neighbor data structure
  with $T_\nn(G)=O(m+n)$ and $Q_\nn(G)=O(1)$.
\end{restatable}
The following lemma allows composing graph near neighbor data structures efficiently.

\begin{restatable}{lemma}{nnsum}\label{l:nn-sum}
  Let $G_1=(V_1,E_1),\ldots,G_k=(V_k,E_k)$ be weighted digraphs with preprocessed near neighbor data structures.
  There exists a near neighbor data structure for $\bigcup_{i=1}^k G_i$ that requires no further preprocessing and
  satisfies 
  \begin{align*}
    T_\nn\left(\bigcup_{i=1}^k G_i\right)&=\sum_{i=1}^k \left(T_\nn(G_i)+|V_i|\cdot Q_\nn(G_i)+O(|V_i|)\right),\\
    Q_\nn\left(\bigcup_{i=1}^k G_i\right)&=\max_{i=1}^k\{Q_\nn(G_i)\}+O(1).
  \end{align*}
\end{restatable}
\subsection{Graph closest pair data structure}\label{s:graph_closest}
Let $G=(V,E)$ be as in Section~\ref{s:graph-near}. Let $S,T\subseteq V$
and let $\alpha:S\to \mathbb{R}$ and $\beta:T\to \mathbb{R}$ be vertex weights.
A \emph{graph closest pair data structure} on $G$ explicitly maintains an edge $st=e^*\in E\cap (S\times T)$ (if exists) such
that
\begin{equation*}
  d(e^*):=\wei(e^*)+\alpha(s)+\beta(t)
\end{equation*}
is minimized.
After preprocessing $G$ (possibly given implicitly), it has the following interface.
\begin{enumerate}[itemsep=-2pt]
  \item \emph{Initialization} with specified sets $S,T\subseteq V$, and mappings $\alpha,\beta$.
  \item \emph{Activation} of $s\in V\setminus S$ so that $s$ gets inserted into $S$ (with a given value $\alpha(s)$).
  \item \emph{Extraction} of $t\in T$, so that $t$ gets removed from $T$.
\end{enumerate}
The \emph{total update time} of the graph closest pair data structure is defined
as the total time needed for initialization and processing an arbitrary sequence of
(at most $n-|S_0|$, where $S_0$ is the initial set $S$) activations and (at most $|T_0|$, where $T_0$ is the initial set $T$)
extractions. 

We denote by $T_\cp(H)=\Omega(|V(H)|)$ the total update time of a graph closest pair data structure on an
appropriately preprocessed graph $H$.

Below we state some useful observations about the above closest pair abstraction.
\begin{fact}\label{l:ssp-dummy-clo}
  For a single-edge digraph $G$, there exists a graph closest pair data structure
  with total update time $T_\cp(G)=O(1)$.
\end{fact}

\begin{restatable}{lemma}{dijkstraadd}\label{l:dijkstra-add}
  {\normalfont\cite{ChanE01}} Given a feasible price function $p$ and a preprocessed graph closest pair data structure on $G$,
  the 
  distances $\dist_{G}(s,\cdot)$ from any source $s$ can be computed in $\Ot\left(T_\cp(G)\right)$ time.
\end{restatable}
\begin{proof}
  We initialize a graph closest pair data structure $\ds$ on $G$
  with $S=\{s\}$, $\alpha(s):=-p(s)$, $T=V\setminus S$, $\beta(t):=p(t)$.

  We simulate Dijkstra's algorithm on $G$ using $\ds$. Recall that Dijkstra's algorithm
  maintains a growing set $S$ (initially $S=\{s\}$) of vertices $v$ for which distances $\dist_G(s,v)$
  are known and repeatedly picks vertices $z$ minimizing
  \begin{equation}\label{eq:dijkstra-min}
    g(z):=\min_{vz=e\in E\cap (S\times (V\setminus S))}\{\dist_G(s,v)+(\wei_G(e)-p(v)+p(z))\},
  \end{equation}
  and subsequently establishes that $\dist_G(s,z)$ equals $g(z)$ and
  extends $S$ with $z$.
  We can indeed achieve the same using the data structure $\ds$ as follows. Whenever Dikstra's
  algorithm establishes some distance $\dist_G(s,a)$,
  we activate $a$ in $\ds$ with $\alpha(a):=\dist_G(s,a)-p(a)$.
  For finding next $z\in T$ minimizing $g(z)$, observe that the edge $e^*=vz$
  maintained by $\ds$ is the one minimizing $g(z)$ over the set $T$.
  To maintain the invariant that $T=V\setminus S$ in $\ds$, we simply extract
  $z$ from $\ds$ after activating $z$ in $S$.
  
  Clearly, such an implementation of Dijkstra's algorithm runs
  in $O(n)$ time plus the total update time of $\ds$, which is $\Omega(n)$.
\end{proof}

\begin{restatable}{lemma}{sspsum}\label{l:ssp-sum}
  Let $G_1=(V_1,E_1),\ldots,G_k=(V_k,E_k)$ be weighted digraphs with preprocessed closest
  pair
  data structures.
  There exists a graph closest pair data structure for $\bigcup_{i=1}^k G_i$ with no further
  preprocessing required 
  and total update time
  \begin{equation*}
    T_\cp\left(\bigcup_{i=1}^k G_i\right)=\sum_{i=1}^k T_\cp(G_i)+\sum_{i=1}^k O(|V_i|\log{k}).
  \end{equation*}
\end{restatable}
\subsection{Negative cycle detection}
Having defined the data structural abstractions, we are now ready to state
the following result giving an upper bound on the running time of a negative
cycle detection algorithm in terms of the performance characteristics
of the assumed graph data structures.

\begin{restatable}{theorem}{tnegcyc}\label{t:negcyc}
  Let $G$ be a graph with integer edge weights in $[-C,\infty)$.
  Suppose closest pair and near neighbor data structures have been preprocessed for $G$.
  Then, one can test whether~$G$
  contains a negative cycle (and possibly find one) in
  \begin{equation*}
    \Ot\left(\sqrt{n}\cdot (T_\cp(G)+T_\nn(G)+n\cdot Q_\nn(G))\log{C}\right)
  \end{equation*}
  time. If no negative cycle is detected, a feasible price function $p$ of $G$ is returned.
\end{restatable}
We give an overview of the techniques behind the algorithm of Theorem~\ref{t:negcyc} and compare it with
the previous scaling approaches in
Section~\ref{s:circ-overview}. A complete proof can be found in Section~\ref{s:circ0}.

\section{Technical overview}\label{s:overview}
In this section we give an overview of the techniques used in (1) our results for
planar graphs regarding negatively-weighted SSSP and max $s,t$-flow oracles 
and (2) our min-cost circulation framework for general graphs with efficient
graph near neighbor and closest pair data structures.

\subsection{Negative cycle detection in DDGs and planar digraphs}\label{s:overview-negcyc}

All the previous near-optimal algorithms for negative cycle detection in planar digraphs~\cite{FR06, KleinMW10, MozesW10}
rely on an efficient
implementation of the Bellman-Ford algorithm on so-called \emph{dense distance graphs}
built upon \emph{$r$-divisions with few holes}.

Let us now introduce these and some other standard planar graph tools in more detail.
  Once again, an \emph{$r$-division}~\cite{DBLP:journals/siamcomp/Frederickson87} $\rdiv$ of a planar graph, for $r \in [1,n]$,
is a decomposition of a planar graph~$G$ into $O(n/r)$ \emph{pieces} $P\subseteq G$ of size $O(r)$ such that each piece $P$ shares $O(\sqrt{r})$ vertices with other
pieces. The shared vertices of a piece $P$, denoted $\bnd{P}$, are called its \emph{boundary vertices}.
We denote by $\bnd{\rdiv}$ the set $\bigcup_{P\in\rdiv}\bnd{P}$.
If additionally all pieces of $\rdiv$ are connected, and the boundary vertices of each piece $P\in\rdiv$ are distributed
among $O(1)$ faces of $P$ (containing boundary vertices exclusively, also called the holes\footnote{This definition is slightly more general than usually. Namely, the definition of an $r$-division does not assume a fixed embedding of the entire $G$; it only assumes some fixed embeddings of individual pieces.}
of $P$), we call $\rdiv$ an \emph{$r$-division with few holes}.
\begin{theorem}[\cite{KleinMS13}]\label{t:rdiv}
  Let $G$ be a simple triangulated connected plane graph with $n$ vertices.
  For any $r\in [1,n]$, an $r$-division with few holes of $G$
  can be computed in $O(n)$ time.
\end{theorem}
A dense distance graph $\DDG(\rdiv)$ built upon $\rdiv$ is obtained by unioning, for each piece $P\in\rdiv$,
a complete weighted graph $\DDG(P)$ on $\bnd{P}$ encoding distances between vertices $\bnd{P}$ in $P$.
The graph $\DDG(\rdiv)$ has $O(n/\sqrt{r})$ vertices and $O(n)$ edges. Dense distance graphs are typically constructed using the following \emph{multiple-source shortest paths} data structure.

\begin{theorem}[MSSP~\cite{CabelloCE13, Klein05}]\label{t:mssp}
  Let $G=(V,E)$ be a plane digraph with a distinguished face $f$.
  Suppose a feasible price function on $G$ is given.
  Then, in $O(n\log{n})$ time one can construct a data
  structure that can compute $\dist_G(s,t)$ for any query vertices
  $s\in V(f)$, $t\in V$ in $O(\log{n})$ time.
\end{theorem}

\newcommand{\dc}{\ensuremath{\text{DC}}}

A naive implementation of the Bellman-Ford algorithm on $\DDG(\rdiv)$ runs in
$O(n^2/\sqrt{r})$ time. \cite{FR06, KleinMW10, MozesW10} showed how to implement the Bellman-Ford algorithm
on $\DDG(\rdiv)$ in $\Ot(|V(\DDG(\rdiv))|^2)=\Ot(n^2/r)$ time.
Using this implementation for $r=\Theta(n)$ or even\linebreak $r=n/\polylog{n}$, in combination with recursion,
one obtains near-linear (strongly polynomial) running time, which~\cite{MozesW10} manage to optimize to $O(n\log^2{n}/\log\log{n})$.
To obtain Theorem~\ref{t:shpath} and all the other results for planar graphs in this paper, we prove the following.
\begin{theorem}\label{t:negcyc-rdiv}
  Let $\rdiv$ be an $r$-division with few holes of a planar digraph $G$ whose
  edge weights are integers in $[-C,\infty)$.
  Given $\DDG(\rdiv)$, one can either detect a negative cycle in $\DDG(\rdiv)$ or compute
  its feasible price function in $\Ot\left(|V(\DDG(\rdiv))|^{3/2}\log{C}\right)=\Ot(n^{3/2}/r^{3/4}\cdot \log{C})$ time.
\end{theorem}

\newcommand{\mmat}{\mathcal{M}}
Let us first explain how Theorem~\ref{t:negcyc} implies Theorem~\ref{t:negcyc-rdiv}.
Let $\rdiv$ be an $r$-division with few holes of $G$.
It is well-known~\cite{FR06, MozesW10, GawrychowskiK18} that for $P\in\rdiv$, $\DDG(P)$,
seen as a matrix with rows and columns $\bnd{P}$, can be expressed
as an element-wise minimum of a number of \emph{full Monge matrices}\footnote{An $n\times m$ matrix~$\mmat$ is called \emph{Monge}, if for any
rows $a\preceq b$ and columns $c\preceq d$ we have $\mmat_{a,c}+\mmat_{b,d}\leq \mmat_{a,d}+\mmat_{b,c}$.}
$\mmat_1,\ldots,\mmat_k$
with a total of $\Ot(|\bnd{P}|)$ rows and columns (counting with multiplicities).
For
$m_1\times m_2$ Monge  matrices,
\cite[Lemma 1]{MozesNW18} showed an efficient data structure that allows to deactivate columns
and supports queries for a minimum
element in a row (limited to the remaining active columns only).
The data structure has total update time $\Ot(m_1+m_2)$ and supports queries in $\Ot(1)$ time.
Such a data structure for $\mmat_i$ can be used in a trivial way
  to implement a \emph{near neighbor} data structure on a graph $G_{\mmat_i}$
  whose edges correspond to the elements of $\mmat_i$,
  so that
  $T_\nn(G_{\mmat_i})=\Ot(|V(G_{\mmat_i})|)=\Ot(|\bnd{P}|)$ and $Q_\nn(G_{\mmat_i})=\Ot(1)$.
  That data structure, combined with the \emph{Monge heap} of~\cite{FR06},
  can be used to obtain a closest pair data structure on $G_{\mmat_i}$ with
  $T_\cp(G_{\mmat_i})=\Ot(|\bnd{P}|)$.
But the graphs $\bigcup_{i=1}^k G_{\mmat_k}$ and $\DDG(P)$ have the
same distances between the vertices $\bnd{P}$,
and thus by Lemmas~\ref{l:ssp-sum}~and~\ref{l:nn-sum},
we also have $T_\cp(\DDG(P)),T_\nn(\DDG(P))\in \Ot(|\bnd{P}|)=\Ot(\sqrt{r})$,
and $Q_\nn(\DDG(P))=\Ot(1)$.
Applying Lemmas~\ref{l:ssp-sum}~and~\ref{l:nn-sum} once again, we obtain
$T_\cp(\DDG(\rdiv)),T_\nn(\DDG(\rdiv))\in \Ot(n/\sqrt{r})$ and $Q_\nn(\DDG(\rdiv))=\Ot(1)$.
By plugging these bounds in Theorem~\ref{t:negcyc}, we obtain Theorem~\ref{t:negcyc-rdiv}.

Given an efficient negative cycle detection algorithm on $r$-divisions, for negative
cycle detection in planar digraphs, we apply
the same recursive strategy as previous works~\cite{FR06, KleinMW10, MozesW10} did.
However, we reduce the graph size in the recursive calls much more aggressively.
First, an $r$-division with few holes is computed in linear time for $r=n^{8/9}$.
Then, the algorithm is run recursively on the individual pieces of $\rdiv$.
If no piece $P$ has a negative cycle fully contained in $P$, we only need
to look for negative cycles going through $\bnd{\rdiv}=\bigcup_{P\in\rdiv}\bnd{P}$.
To this end, $\DDG(\rdiv)$ is first built using the obtained per-piece feasible price functions
and the MSSP algorithm~\cite{Klein05}.
This takes $O(n\log{n})$ time.
Then, the algorithm of Theorem~\ref{t:negcyc-rdiv} is run on $\DDG(\rdiv)$,
which takes $\Ot(n^{3/2}/r^{3/4}\log{C})=\Ot(n^{5/6}\log{C})$ time.
A feasible price function of $G$ can be computed out of the price function
on $\DDG(\rdiv)$ and the individual per-piece price functions in $O(n\log{n})$ time.

The time cost $T(n)$ of this algorithm satisfies
$T(n)=O(n\log{n})+\Ot(n^{5/6}\log{C})+ \sum_{P\in \rdiv} T(|P|)$.
Intuitively, the bound $T(n)=O(n\log{n}+n\log{C})$ holds because (a) at each level of recursion, except possibly the leafmost levels,
the total sum of $\Ot(n^{5/6}\log{C})$ terms is $o(n\log {C})$,
and (b) by summing the $O(n\log{n})$ terms in subsequent levels of the recursion tree, we have:
\begin{equation*}
  n\log{n}+n\log{n^{8/9}}+n\log{n^{(8/9)^2}}+\ldots\leq n\log{n}\cdot \sum_{i=0}^{\infty}(8/9)^i=O(n\log{n}).
\end{equation*}

\subsection{Max $s,t$-flow oracles}

Our developments for max-flow oracles build upon a well-known reduction of the decision~variant of the max $s,t$-flow problem on a planar digraph
to the negative cycle detection problem.
In the recent literature~\cite{Erickson10, nussbaum2014network}, this reduction is generally attributed to an unpublished manuscript by Venkatesan (and also appears in~\cite{JohnsonV83}).
The reduction has been generalized to arbitrary vertex demands in~\cite{MillerN95}.

Assume that each edge of $G$ comes with a reverse edge of capacity~$0$.
This does not influence the amounts of flow that one can send in this graph.
Let $\lambda\geq 0$.
Pick an arbitrary $s\to t$ path $Q_{s,t}$ in $G$. 
Let $G_\lambda$ be obtained by decreasing the capacity
of each edge $e$ of $Q_{s,t}$ by $\lambda$, and increasing the capacity
of the reverse of $e$ by~$\lambda$.
\begin{restatable}{lemma}{venkatesan}\label{l:venkatesan}
  There exits an $s,t$-flow of value $\lambda$ in $G$ iff the dual
  $G_\lambda^*$ of $G_\lambda$ has no negative cycles.
\end{restatable}
For some intuition behind this reduction, see Section~\ref{s:flow-reduction}. A full proof can be found, e.g., in~\cite{Erickson10}.
Note that the reduction allows binary searching for the maximum $s,t$-flow value.
For example, using Theorem~\ref{t:shpath},
for integral capacities in $[1,C]$, one can compute a maximum $s,t$-flow
value in $O(n\log^2(nC))$ time.
However, this bound is worse and less general 
than the best known strongly polynomial $O(n\log{n})$ bound~\cite{BorradaileK09, Erickson10}
achieved without applying this reduction directly.
Nevertheless, the reduction is quite powerful, e.g., in the parallel setting~\cite{MillerN95,KarczmarzS21}.
\paragraph{Feasible flow oracle.} Let us first consider the $\lambda$-feasible flow oracle problem, where $\lambda$ is fixed
and we only need to support queries about whether for a given pair $(s,t)$
the max $s,t$-flow value is at least $\lambda$.
Clearly, we cannot simply find, upon query, a path $Q_{s,t}\in G$,
and explicitly adjust its capacities since such a path can have length $\Theta(n)$.
In order to make use of Lemma~\ref{l:venkatesan} for many distinct pairs $(s,t)$, we use
it combined with a \emph{recursive decomposition} $\TG(G^*)$ of the dual graph $G^*$
using cycle separators~\cite{Miller84}. Such a decomposition can be computed in linear time~\cite{KleinMS13}.
$\TG(G^*)$ is a tree of connected subgraphs of $G^*$ rooted at $G^*$, has $O(\log{n})$ depth,
and the total size of pieces of $\TG(G^*)$ is $\Ot(n)$.
Moreover, every piece $H\in \TG(G^*)$ has at most $\Ot(\sqrt{n})$ boundary vertices $\bnd{H}$
(i.e., vertices shared with pieces of $\TG(G^*)$ that are not ancestors nor descendants of $H$),
and $\sum_{H\in \TG(G^*)}|\bnd{H}|^2=\Ot(n\log{n})$.

Given the decomposition, we construct a
set $\mathcal{P}$ of paths in $G$ 
such that for every $(s,t)$, there
exists an $s\to t$ path in $G$ that can be decomposed
into $O(\log{n})$ paths from~$\mathcal{P}$.
Specifically, the set $\mathcal{P}$ consists of small individual per-piece sets
$\mathcal{P}_H$, $H\in \TG(G^*)$, such that $|\mathcal{P}_H|=\Ot(1)$
and for every path $Q\in \mathcal{P}_H$, its dual edges $Q^*\subseteq G^*$ satisfy $Q^*\subseteq H$.
Moreover, for every pair $(s,t)$, $G^*$ can be efficiently decomposed
into a collection $\mathcal{H}_{s,t}$ of $O(\log{n})$ pieces $H_1,\ldots,H_\ell\in \TG(G^*)$
such that:
\begin{enumerate}[itemsep=-2pt]
  \item $H_1\cup\ldots\cup H_\ell=G$, and no two of these pieces are in an ancestor-descendant relationship,
  \item there exist paths $Q_1\in \mathcal{P}_{H_1},\ldots,Q_\ell\in \mathcal{P}_{H_\ell}$ such that $Q_1\cup\ldots \cup Q_\ell$ forms an $s\to t$ path in $G$, which
    we fix as the chosen path $Q_{s,t}$ (as required by the reduction) for that pair $(s,t)$.
\end{enumerate}

Next, roughly speaking, for every $H\in \TG(G^*)$ and $Q\in \mathcal{P}_H$, we preprocess a
dense distance graph $\DDG(G^*_\lambda[H],Q)$ (defined, again, as a distance clique on $\bnd{H}$) of the induced subgraph $G^*_\lambda[H]$ under the
assumption that $G^*_\lambda[H] \cap Q^*_{s,t}=Q^*$.
Building one such $\DDG(G^*_\lambda[H],Q)$ (after either finding a feasible price function of $G^*_\lambda[H]$ or detecting a negative cycle)
costs $\Ot(|H|+|\bnd{H}|^2)$ time using the strongly polynomial negative cycle detection algorithm~\cite{MozesW10}
and the MSSP algorithm~\cite{Klein05}.
Since $|\mathcal{P}_H|=\Ot(1)$, the total time spent on preprocessing
is $\Ot\left(\sum_{H\in\TG(G^*)} |H|+|\bnd{H}|^2\right)=\Ot(n)$.

Given the preprocessing, the intuition behind the query algorithm is as follows.
First, we compute the collection $\mathcal{H}_{s,t}=\{H_1,\ldots,H_\ell\}$ as described above. 
For each $H_i\in \mathcal{H}_{s,t}$, let $Q_i\in \mathcal{P}_{H_i}$ be the required
subpath of $Q_{s,t}$.
Recall that to answer the query, it is enough to check if $G^*_\lambda$ contains
a negative cycle.
We first check whether any of the graphs $G^*_\lambda[H_i]$
(with $G^*_\lambda[H_i] \cap Q^*_{s,t}=Q^*_i$) contains a negative cycle, which
is an information that we have precomputed.
If not, then any negative cycle in $G^*_\lambda$ has to  pass through a vertex from
$\bigcup_{i=1}^\ell \bnd{H_i}$.
To check if such a negative cycle exists, we run (a variant of) 
the algorithm of Theorem~\ref{t:negcyc-rdiv} on the graph $\bigcup_{i=1}^\ell \DDG(G^*_\lambda[H_i],Q_i)$.
Since this graph has only $\Ot(n^{1/2})$ vertices, 
negative cycle detection takes $\Ot(n^{3/4}\log{C})$ time as desired.

\paragraph{Dynamic feasible/approximate max $s,t$-flow oracle.} Turning the static feasible flow and approximate max $s,t$-flow oracles into corresponding dynamic
oracles requires using an \emph{$r$-division-aware} version of the decomposition
$\mathcal{H}_{s,t}$, in which the pieces $\mathcal{H}_{s,t}$ are only allowed
to be weak descendants of some fixed $r$-division $\rdiv$ drawn from the nodes of $\TG(G^*)$ (such $\rdiv$ always exists~\cite{KleinMS13}).
In such a case, $\mathcal{H}_{s,t}$ may have size $\widetilde{\Theta}(n/r)$, and the corresponding
graph $\bigcup_{i=1}^\ell \DDG(G^*_\lambda[H_i],Q_i)$ may have $\widetilde{\Theta}(n/\sqrt{r})$ vertices,
so the query time gets increased to $\Ot(n^{3/2}/r^{3/4}\log{C})$.
However, if some edge capacity is updated, we need to recompute the data
structures for only one piece of $\rdiv$ (and its descendants in $\TG(G^*)$), which takes $\Ot(r)$ time.
Setting $r=n^{6/7}$ balances the worst-case update and query times.

\paragraph{Exact max $s,t$-flow oracle.}
For answering exact max $s,t$-flow queries, we would like to find the maximum
$\lambda\in [0,nC]$ such that $G^*_\lambda$ does not contain a negative cycle
via binary search.
The $r$-division-aware decomposition $\mathcal{H}_{s,t}$, nor the implied
path $Q_{s,t}$, do not depend on the parameter~$\lambda$. Therefore, we would be able to perform binary search
in $\Ot(n^{3/2}/r^{3/4}\log^2{C})$ time if
only we had the dense distance graphs $\DDG(G^*_\lambda[H_i],Q_i)$
for $i=1,\ldots,\ell$ precomputed for all parameters $\lambda$
generated on-line by binary search.
Unfortunately, we cannot afford to simply precompute these graphs
for all the $nC$ possible parameters $\lambda$.
This would be too costly even for unweighted graphs.

Recall from the reduction of Theorem~\ref{t:negcyc-rdiv} to Theorem~\ref{t:negcyc}, that for detecting a negative cycle, we actually do not need the
graphs $\DDG(G^*_\lambda[H_i],Q_i)$ themselves, but rather the associated graph closest pair and near
neighbor data structures.
We prove that after preprocessing
a piece \linebreak $H\in \TG(G^*)$ in $\Ot(|H^2|)$ time, one can construct
the required data structures for $\DDG(G^*_\lambda[H],Q)$ (where $Q\in P_H$) for any \emph{given} 
$\lambda$, so that
$T_\cp(\DDG(G^*_\lambda[H],Q)),T_\nn(\DDG(G^*_\lambda[H],Q))\in\Ot(|\bnd{H}|)$,
and $Q_\nn(\DDG(G^*_\lambda[H],Q))=\Ot(1)$.
As we need this preprocessing only for the pieces $H\in \TG(G^*)$
that are weak descendants of $\rdiv$,
this costs $\sum_{H\in\TG(G^*)}\Ot(|H|^2)\leq \Ot(r)\cdot \sum_{H}\Ot(|H|)=\Ot(nr)$ time.

Recall that $\bnd{H}$ lies on $O(1)$ holes $h^1,\ldots,h^k$ of $H$.
We first decompose $\DDG(G^*_\lambda[H],Q)$ into
subgraphs that capture only shortest paths between $O(1)$ fixed pairs of holes $h^i,h^j$.
We prove that for given $\lambda$, $h^i$, and $h^j$, the shortest paths between the
vertices $V(h^i)$ of the hole $h^i$ and vertices $V(h^j)$ of the hole $h^j$
all cross the path~$Q$ (seen as a curve in $G^*_\lambda[H]$) \emph{almost the same number of times}.
More specifically, the individual net numbers of crossings (or simply the \emph{crossing numbers}) of these shortest paths wrt.~$Q$
differ by at most $O(1)$.
Tracking crossing numbers is useful, since the length of a path $P$ with crossing number $\pi$ in 
$G^*_\lambda[H]$ can be seen to be the length of $P$ in $G^*_0[H]$
shifted by $\lambda\cdot \pi$.
This allows us to limit our attention to constructing the required data structures
for auxiliary graphs that encode minimum lengths of such paths 
between $V(h^i)$ and $V(h^j)$ in $G^*_0[H]$ that additionally have their crossing numbers wrt. $Q$ equal precisely $\pi$,
where $\pi$ is a parameter chosen in $O(1)$ possible ways (which can be determined efficiently).

To solve the modified problem, we consider a graph $H'$ obtained from $H$ by first cutting $H$ along $Q$ and then
gluing $2n+1$ copies of the cut graph, numbered $-n,\ldots,n=|V(H)|$. The graph $H'$ is conceptually similar
to the infinite ``universal cover'' graph from~\cite[Section 2.4]{Erickson10}.
The crucial property of $H'$ is that the distance between the $0$-th copy of $u\in V(h^i)$
and the $\pi$-th copy of $v\in V(h^j)$ in $H'$ equals
the minimum length of a \emph{simple} $u\to v$ path in $H$ with crossing number wrt. $Q$ precisely~$\pi$.
This way, we reduce our modified problem to constructing the required data structures for a graph
encoding distances between some two holes of the much larger graph $H'$.
  
Using an MSSP data structure~\cite{Klein05} built on $H'$, we can access distances
from the $0$-th copy $h_0^i$ of $h^i$ to any specified copy $h_l^j$ of $h^j$ in $\Ot(1)$ time.
However, in order to construct efficient closest pair/near neighbor data structures
for a graph encoding distances from $h_0^i$ to $h_\pi^j$, we need to be able
to organize these distances into Monge matrices.
While this is easy if $h_0^i=h_\pi^j$, it is more problematic
if these holes are distinct. \cite{MozesW10}~deal with this general problem using
additional \emph{near-linear} preprocessing \emph{per each pair of holes} of interest,
that allows to decompose the respective distance matrix into element-wise minimum of two \emph{full} (i.e., rectangular) Monge matrices.
Unfortunately, in our case we need to handle $\Theta(|H|)$ pairs of holes $(h_0^i,h_l^j)$ (for all $j=-n,\ldots,n$),
so using their approach would lead to $O(|H|^3)$-time preprocessing per piece, which in
our case would eventually prevent us from achieving subquadratic preprocessing
and sublinear query time.

We develop a more involved method which allows to achieve the goal
for single \emph{source hole} and all \emph{target holes} at once using near-linear preprocessing.
However, the produced decomposition involves element-wise minimum of two so-called \emph{partial} Monge matrices
which are slightly more difficult to handle.
Fortunately,
any $m\times m$ partial Monge matrix can be decomposed into full Monge matrices 
with $\Ot(m)$ total rows and columns~\cite{GawrychowskiMW20}. Consequently, efficient closest pair
and near neighbor data structures can be constructed as was the case for piecewise
dense distance graphs in Theorem~\ref{t:negcyc-rdiv}.

\subsection{Min-cost circulations and negative cycle detection in general graphs}\label{s:circ-overview}

We use the term \emph{vertex splitting} to refer to the following graph transformation:
for each $v\in V$, create two vertices $v_\tin,v_\tout$, and for each non-loop edge $uv=e\in E$ (i.e., $u\neq v$),
replace it with an edge $u_\tout v_\tin$ of the same weight.
As shown by the following lemmas (proved in Section~\ref{s:omitted}), graph near-neighbor and closest pair data structures
are well-behaved under vertex splitting.

\begin{restatable}{observation}{nnsplit}\label{l:nn-split}
  Suppose a graph near neighbor data structure has been preprocessed for $G$.
  If the graph $G'=(V',E')$ is obtained from $G=(V,E)$ by vertex splitting, then with no additional preprocessing we have
  $T_\nn(G')=O(T_\nn(G))$ and $Q_\nn(G')=O(Q_\nn(G))$.
\end{restatable}

\begin{restatable}{lemma}{sspsplit}\label{l:ssp-split}
  Suppose a graph closest pair data structure has been preprocessed for $G$.
  If the graph $G'=(V',E')$ is obtained from $G=(V,E)$ by vertex splitting, then with no additional preprocessing we have
  and $T_\cp(G')=\Ot(T_\cp(G))$.
\end{restatable}

Gabow~\cite{Gabow85} gave an elegant vertex splitting-based reduction of the integral negative cycle detection problem
to the minimum-cost perfect matching problem with acceptable additive error $O(1/n)$.
In the proof of Theorem~\ref{t:negcyc}, we use a similar reduction\footnote{We use this approach, instead of e.g., adapting the more direct Goldberg's scaling algorithm~\cite{Goldberg95} for negative cycle detection because it is easily described in terms of repeated non-negative shortest paths computations~\cite{KarczmarzS19}, for which very efficient algorithms on dense distance graphs (see Section~\ref{s:overview-negcyc}) are known~\cite{FR06}.}, but instead of considering matchings, we find it
more convenient to reduce to the minimum-cost circulation problem in a network satisfying the following:
for each vertex $v$, the minimum $\lambda_v$ of the total capacity
of $v$'s incoming edges and the total capacity of $v$'s outgoing edges
satisfies $\lambda_v= 1$. This resembles so-called ``type~2'' networks from~\cite{EvenT75, GoldbergHKT17},
for which classical combinatorial unit-capacity flow algorithms~\cite{EvenT75, GoldbergHKT17}
run in $O(\sqrt{n}m)$ time as opposed to $O(m^{3/2})$ time.
In the reduction that we apply, the transformed graph for which we want to solve the min-cost
circulation instance consists of all the edges of $G$ (uncapacitated) and~$O(n)$ auxiliary unit-capacity edges.

The reduction allows us to focus on the more general min-cost circulation problem
with the additional parameter $\Lambda:=\sum_{v\in V}\lambda_v$.
For arbitrary values of $\Lambda\geq n$, in Section~\ref{s:circ} (Theorem~\ref{t:circ}) we prove that a $\delta$-additive approximation
of the min-cost circulation can be computed in time

\begin{equation}\label{eq:ov1}
  \Ot\left(\sqrt{\Lambda}\log{(\Lambda C/\delta)}\cdot \left(T_\cp(G_\infty)+T_\nn(G_\infty)+m_\fin+\Lambda\cdot Q_\nn(G_\infty)\right) \right).
\end{equation}
Here, $G_\infty\subseteq G$ is the subgraph of $G$ containing all the uncapacitated edges,
and $m_\fin$ is the number of finite-capacity edges in $G$.
Recall that in the instance arising from the reduction, we have $G_\infty=G$, $\Lambda=n$ and $m_\fin=O(n)$,
and we need additive error $\delta=O(1/n)$. This is how Theorem~\ref{t:negcyc} follows from the bound~\eqref{eq:ov1} proved in Theorem~\ref{t:circ}.

The minimum-cost circulation algorithm of Section~\ref{s:circ} is a modification
of a simple successive shortest augmenting path algorithm for min-cost circulation in unit-capacity networks,
as described in~\cite{KarczmarzS19}. 
The outer loop of the algorithm implements the successive approximation framework
of~\cite{GoldbergT90}: each iteration (called \emph{refinement}) is supposed to convert a $2\eps$-approximate
min-cost circulation into an $\eps$-approximate circulation.
Since a trivial zero circulation is always a \linebreak $\Theta(\Lambda\cdot C)$-approximate min-cost circulation,
using $O(\log{(\Lambda C/\delta)})$ refinement steps, it can be turned
into a circulation whose cost differs from the optimum by at most $\delta$.

In~\cite{KarczmarzS19}, the refinement step first 
converts the input circulation into a flow $f$ that is trivially
minimum-cost, but violates conservation, i.e. $f$ is not a circulation.
Next, it suitably \emph{rounds} the edge costs $c$ 
in the residual network $G_f$ up to nearest multiple of $\eps$, this way obtaining a modified \emph{discretized} cost function $c'$. 
Then, it gradually converts $f$ into a circulation by sending flow from excess to deficit vertices
in $G_f$
while maintaining
that $f$ is minimum-cost wrt. costs~$c'$. 
More specifically, it simply runs $O(\sqrt{\Lambda})$ \emph{augmentation} steps\footnote{In~\cite{KarczmarzS19}, only the $O(\sqrt{m})$ bound
on the number of augmentation steps is proven for unit-capacity networks. For such networks, only the bound $\Lambda=O(m)$ can be proven
with no additional assumptions.
We extend their analysis (based on~\cite{GoldbergHKT17}) to networks with integer or infinite capacities and arbitrary values $\Lambda$.} 
sending
flow along a maximal set of edge-disjoint shortest paths (wrt. $c'$) in $G_f$.
Each such step in turn can be split into substeps (1): a single-source distances computation
and (2): finding a maximal set of edge-disjoint
paths consisting of edges with reduced cost~$0$ using a DFS-like procedure (both can be implemented in $\Ot(m)$ time).

Compared to~\cite{KarczmarzS19}, our refinement procedure avoids rounding for the following reason:
we want to keep the subgraph $G_\infty$ with unchanged costs in the residual network at all times, since
we want to benefit from closest pair and near neighbor data structures for $G_\infty$.
These data structures could have, in principle, much worse performance if applied
to e.g., a subgraph of $G_\infty$ or $G_\infty$ with perturbed costs.
On the other hand, without rounding, the bound $O(\sqrt{\Lambda})$ on the number of augmentation steps
fails to hold due to lack of discretization of possible lengths of shortest augmenting paths.
Nevertheless, we show that such a bound still holds if we (i) only increase the cost of $O(\Lambda)$ edges in $G_f\setminus G$
by $\eps$, and (ii) after computing distances in substep (1), we augment the flow through a maximal
set of \emph{nearly-tight} edges in substep (2), with reduced costs less than $\eps/2$.
In substep (1), the distances can be computed using Lemma~\ref{l:dijkstra-add} in
$\Ot(T_\cp(G_f))$ time. But since graph closest pair data structures
can be efficiently composed under unions (Lemma~\ref{l:ssp-sum}) we~have:
\begin{equation*}
  T_\cp(G_f)=T_\cp(G_\infty)+T_\cp(G_f\setminus G_\infty)+\Ot(n)=T_\cp(G_\infty)+\Ot(m_\fin+\Lambda).
\end{equation*}
Similarly, to implement substep~(2), we show a simple procedure whose running time is
dominated by the total update time $T_\nn(G_f)$ of a graph near neighbor data structure on $G_f$, and
performing $O(\Lambda)$ queries on such a data structure. Again, by Lemma~\ref{l:nn-sum},
we get:
\begin{equation*}
  T_\nn(G_f)+Q_\nn(G_f)\cdot O(\Lambda)=T_\nn(G_\infty)+Q_\nn(G_\infty)\cdot O(\Lambda)+\Ot(m_\fin+\Lambda).
\end{equation*}

\section{Negative-weight shortest paths in planar graphs}\label{s:negcyc}

In this section we describe a recursive algorithm that,
given a simple connected plane digraph $G$ with integral edge
weights at least $-C$,
either detects a negative weight cycle in $G$ or computes a feasible
price function of $G$. Assume that $C\geq 2$, so that $\log{C}\geq 1$.

If $G$ has $n\leq n_0$ vertices, where $n_0$ is a constant to be set later, we solve the problem in $O(n^2)=O(n_0^2)=O(1)$ time using any strongly polynomial algorithm, e.g., Bellman-Ford. So, in the following, assume that $n$ is at least a sufficiently large constant.

Suppose $n>n_0$. We start by computing an embedding of $G$. We then triangulate $G$ 
by adding bidirectional edges of weight $nC$ inside faces whose bounding cycles have more than $3$ edges.
Note that this cannot introduce any new negative cycles to $G$
and the lower bound $-C$ on the smallest weight still holds.
Next, for $r<n$ to be set later, we build an $r$-division with few holes $\rdiv$ of $G$.
This can be done in linear time by Theorem~\ref{t:rdiv}.
In the following, for $P\in \rdiv$ we will use the notation $|P|$ to refer to $|V(P)|$.
Note that we have $\sum_{P\in \rdiv}|P|=n+O(n/\sqrt{r})$
by the definition of boundary vertices.
The individual pieces $P\in \rdiv$ may have non-simple (that is, whose bounding cycles are not simple cycles) holes.
For simplicity, in the following part of this section we assume this is not the case and discuss how to deal
with non-simple holes in Section~\ref{s:non-simple}.

We then solve the problem recursively for each piece $P\in \rdiv$.
If any of the recursive calls finds a negative cycle, $G$ has a negative cycle
and we can stop.
Otherwise, let $p_P$ be the feasible price function of the piece $P$.
Observe that if $G$ has a negative cycle $O$ and none of the individual pieces $P\in\rdiv$ has one, then by splitting $O$ into a sequence
of maximal subpaths fully contained in some single piece, such a sequence will have length at
least $2$, and thus each of the subpaths will connect two vertices of $\bnd{R}$.
Consequently, we can focus our attention on cycles of that kind only.
The following lemma is the key to achieving sublinear time for detecting such negative cycles.
\begin{lemma}\label{l:ddg-ds}
  Suppose a feasible price function $p_P$ of $P\in\rdiv$ is given.
  Then, for $\DDG(P)$, there exist:
  \begin{itemize}
    \item a closest pair data structure with total update time $\Ot(|\bnd{P}|)$.
    \item a near neighbor data structure with total update time $\Ot(|\bnd{P}|)$ and query time $\Ot(1)$.
  \end{itemize}
  Both data structures require $O((|P|+|\bnd{P}|^2)\log{|P|})$ preprocessing time.
\end{lemma}
\begin{proof}
To obtain efficient closest pair and near neighbor data structures for $\DDG(P)$ we need it decomposed into
Monge matrices.
The following lemma states the properties of such a decomposition proved in~\cite{FR06, MozesW10}.

\begin{lemma}[\cite{FR06, MozesW10}]\label{l:decomp}
  Given a feasible price function $p_P$ of $P$,
  the graph $\ddg(P)$ (seen as a distance matrix) can be computed and decomposed into a set $D_P$ of Monge matrices whose sum of numbers of rows and columns (that are subsets of $\bnd{P})$ is $O(|\bnd{P}|)$.
  For all $u,v\in \bnd{P}$ we have:
  \begin{itemize}
  \item for each $\mmat\in D_P$ such that
  $\mmat_{u,v}$ is defined, $\mmat_{u,v}\geq \dist_{P}(u,v)$.
  \item there exists $\mmat\in D_P$ such that
  $\mmat_{u,v}$ is defined and $\mmat_{u,v}=\dist_{P}(u,v)$.
  \end{itemize}
  The decomposition can be computed in $O((|P|+|\bnd{P}|^2)\log{|P|})$ time.
\end{lemma}

  Computing $\DDG(P)$ and the decomposition $D_P$ 
  of Lemma~\ref{l:decomp} constitutes the only preprocessing for both desired data structures.
  As a result, preprocessing takes $O((|P|+|\bnd{P}|^2)\log{|P|})$ time.

A Monge matrix $\mmat\in D_P$ can be also interpreted as
a directed graph $G_{\mmat}$ with vertices corresponding
to the union of rows and columns of $\mmat$, and a directed edge
from $r$ to $c$ of weight $\mmat_{r,c}$ if the entry $\mmat_{r,c}$ is defined.

  \cite[Lemma 1]{MozesNW18} showed a data structure for $m_1\times m_2$ Monge matrices
that allows to deactivate columns and supports queries for a minimum
element in a row (limited to active columns only) and a contiguous range of columns, given only black-box oracle access to the entries
of the input matrix.
  The data structure has total update time $\Ot(m_1+m_2)$ and supports queries in $\Ot(1)$ time.
  Observe that such a data structure for the matrix $\mmat'$, obtained from $\mmat$ by
  shifting all entries in each column $c$ by some offset $p(c)$, can be used to implement
  a near neighbor data structure on $G_{\mmat}$
  with $T_\nn(G_\mmat)=\Ot(|V(G_{\mmat})|)$ total update time
  and $Q_\nn(G_\mmat)=\Ot(1)$ query time.
  To see this, note first that $\mmat'$ is also a Monge matrix and can be accessed in $O(1)$ time.
  The deactivations of vertices in the near neighbor data structure correspond to deactivations
  of columns of $\mmat'$. A near neighbor query can be handled using a single query
  for a minimum element in a row of $\mmat'$.

Let us now discuss a closest pair data structure of $G_\mmat$, where $\mmat\in D_P$.
We essentially use a variant of the Monge heap of~\cite{FR06}, which we now sketch for completeness.
Again, let $\mmat''$ be a Monge matrix obtained from $\mmat$ by shifting every element in a row $r$
by $\alpha(r)$, and every element in a column $c$ by $\beta(c)$.
To get a desired closest pair data structure, it is enough to have a data structure
  for $\mmat''$ that:
  \begin{enumerate}[(1)]
    \item starts with a subset of \emph{active} rows $S$ and columns~$T$ for which the offsets $\alpha$ and $\beta$ are known,
    \item maintains some minimum entry in $\mmat''$ (limited to the active rows and columns),
    \item supports activations of rows $s\notin S$ revealing
  $\alpha(s)$ and deactivations of columns $t\in T$.
  \end{enumerate}

  It is well-known (see, e.g.,~\cite{FR06}) that for a Monge matrix $\mmat''$, there exist such a sequence $\mathcal{S}=(r_i,c_i,d_i)_{i=1}^k$,
  that:
  \begin{enumerate}[(a)]
    \item for each $i$, $c_i\preceq d_i$, the row $r_i$ is active and contains some column minimum of each of the active columns
  between $c_i$ and $d_i$ (inclusive),
    \item the intervals $[c_i,d_i]$ are disjoint and cover all active columns of $\mmat''$,
  and
    \item the sequences $(r_i)_{i=1}^k$, $(c_i)_{i=1}^k$, $(d_i)_{i=1}^k$ are monotonous wrt. natural
  orders $\prec$ on rows and columns of $\mmat''$.
  \end{enumerate}
  The data structure maintains such a sequence $\mathcal{S}$ subject to row activations in $\mmat''$
  in a balanced binary search tree.
  We can even allow that $c_i,d_i$ come from an initial set $T_0$ of columns so that a column
  deactivation does not break the invariants posed on $\mathcal{S}$.
  Additionally, for each $i$ we explicitly maintain a currently active column $\gamma_i$ such that
  $c_i\preceq \gamma_i\preceq d_i$ and $\mmat''_{r_i,\gamma_i}$ is the minimum entry in the subrow of $r_i$
  spanned by (active) columns in $[c_i,d_i]$.
  The values $\mmat''_{r_i,\gamma_i}$ are all stored in a sorted multiset, so that their minimum $\psi$
  is maintained efficiently. This way, whenever $\gamma_i$ gets updated or deleted, $\psi$ is updated
  in $\Ot(1)$ time. Observe that $\psi$ equals, in fact, the sought minimum value
  in $\mmat''$ at all times.

  Finally, we also build and maintain the data structure of \cite[Lemma 1]{MozesNW18}
  for the matrix $\mmat'$ (as in the near neighbor data structure, with no row offsets $\alpha$ and all rows active, but including the column offsets $\beta$) to allow subrow minimum queries and column deactivations in $\mmat'$.

  Let us now describe how the data structure operates. The initial $\mathcal{S}$
  can be constructed in $O(|V(G_\mmat)|)$ time by finding the initial column
minima of $\mmat''$ using the algorithm of~\cite{AggarwalKMSW87}.

When a column $t\in T$ is deactivated, $\mathcal{S}$ needs no updates.
However, the value $\gamma_j$ for the unique~$j$ such that $c_i\leq t\leq d_j$ might change.
Such a $j$ can be located in $\Ot(1)$ time since $\mathcal{S}$ is stored in a BST.
The new value $\gamma_j$ can be found using a single subrow minimum query to the stored data structure of 
\cite[Lemma 1]{MozesNW18} after processing the deactivation of~$t$.
Note that a subrow minimum value in $\mmat''$ can be obtained from a subrow minimum value in $\mmat'$
by shifting the result by $\alpha(r_j)$.

Now suppose a row $s$ is activated. As in~\cite{FR06}, the sequence $\mathcal{S}$
is updated as follows. First, one needs to identify such columns $c^*,d^*$, $c^*\preceq d^*$, that
$s$ contains the minima in columns in the range $[c^*,d^*]$. To find such a $c^*$ in $\Ot(1)$ time,
one can perform binary search using the previous $\mathcal{S}$: we need to find a minimum $c^*$
such that if $(r_j,c_j,d_j)$ satisfies $c_j\preceq c^*\preceq d_j$, then $\mmat_{r_j,c^*}>\mmat_{s,c^*}$.
As in~\cite{FR06}, the Monge property can be used to prove that binary search indeed works in this case.
Similarly, $d^*$ can be found in $\Ot(1)$ time.
With $c^*,d^*$ computed, if they exist, one needs to remove some number of elements 
$(r_i,c_i,d_i)$ with $[c_i,d_i]\subseteq [c^*,d^*]$ from $\mathcal{S}$, insert $(t,c^*,d^*)$
at some index $l$ to~$\mathcal{S}$, and possibly update the values $d_{l-1}$ and $c_{l+1}$
of the two neighboring elements of $(t,c^*,d^*)$ in~$\mathcal{S}$.
For each updated element of $\mathcal{S}$, we recompute $\gamma_i$ as before (or remove it).
While a row activation can remove many elements from the sequence, it can cause at most one insertion into the sequence, and thus the amortized update time is $\Ot(1)$.
As a result, for the closest pair data structure we get $T_{\cp}(G_\mmat)=\Ot(|V(G_\mmat)|)$.

Consider a graph $H_P=\bigcup_{\mmat\in D_P}G_\mmat$.
 By Lemma~\ref{l:ssp-sum}, we obtain a closest pair
  data structure for $H_P$ with
  \begin{equation*}
    T_\cp(H_P)=\Ot\left(\sum_{\mmat\in D_P}T_\cp(G_\mmat)\right)=\Ot\left(\sum_{\mmat\in D_P}|V(G_\mmat)|\right)=\Ot(|\bnd{P}|).
  \end{equation*}
  Similarly, by Lemma~\ref{l:nn-sum} we obtain $T_\nn(H_P)=\Ot(|\bnd{P}|)$ and
  $Q_\nn(H_P)=\Ot(1)$.

  By Lemma~\ref{l:decomp}, we have that
  $\DDG(P)\subseteq H_P$,
  and for each $u,v\in V(H_P)$, the minimum weight of an edge $uv=e$
  in $H_P$ satisfies $w_{H_P}(e)=w_{\DDG(P)}(uv)$.
  Since parallel edges with non-minimal weight are effectively ignored
  by the closest pair and near neighbor data structures,
  the respective data structures for $H_P$ can be used as corresponding
  data structures for $\DDG(P)$.
  Hence, we obtain $T_\cp(\DDG(P))=\Ot(|\bnd{P}|)$, $T_\nn(\DDG(P))=\Ot(|\bnd{P}|)$
  and $Q_\nn(\DDG(P))=\Ot(1)$, as desired.
 \end{proof}

By the definition of boundary vertices, we have
that $\DDG(G):=\bigcup_{P\in \rdiv}\ddg(P)$ preserves distances between the
vertices of $\bnd{\rdiv}$ in $G$.
As a result, to detect a negative cycle that goes through at least two vertices of $\bnd{\rdiv}$ in $G$,
we can equivalently test whether the graph $\DDG(G)$ has a negative cycle.
Note that we have
$|V(\DDG(G))|\leq \sum_{P\in\rdiv}|\bnd{P}|=O(n/\sqrt{r})$.

Observe that by Lemma~\ref{l:ddg-ds}, the total time needed for preprocessing
closest pair and near neighbor data structures through all $\DDG(P)$ is:
\begin{equation*}
  \sum_{P\in\rdiv} O\left(|P|+|\bnd{P}|^2\log{|P|}\right)=O(n/r)\cdot O\left(\left(r+(\sqrt{r})^2\right)\log{r}\right)=O(n\log{n}).
\end{equation*}
By applying Lemmas~\ref{l:ssp-sum}~and~\ref{l:nn-sum} to the data structures of Lemma~\ref{l:ddg-ds},
we obtain:
\begin{corollary}
  Given preprocessed closest pair and near neighbor data structures for all $P\in\rdiv$, there exist
  closest pair and near neighbor data structures for $\DDG(G)$ such that:
  \begin{align*}
    T_\cp(\DDG(G))&=\Ot\left(\sum_{P\in\rdiv}|\bnd{P}|\right)=
  \Ot(n/\sqrt{r}),\\
    T_\nn(\DDG(G))&=\Ot\left(\sum_{P\in\rdiv}|\bnd{P}|\right)=
  \Ot(n/\sqrt{r}),\\
    Q_\nn(\DDG(G))&=\Ot(1).
  \end{align*}
\end{corollary}

As the edges of $\DDG(G)$ represent distances in individual pieces of $\rdiv$,
edge weights in $\DDG(G)$ are integral and their absolute values can
be bounded by $O(rC)=O(nC)$.
By applying Theorem~\ref{t:negcyc} to $\mathcal{D}:=\DDG(G)$, we obtain that one can test
for a negative cycle in $\DDG(G)$ (or find a feasible price function $q$ of $\DDG(G)$)
in
\begin{equation*}
  \Ot\left(\sqrt{|V(\mathcal{D})|}\cdot (T_\cp(\mathcal{D})+T_\nn(\mathcal{D})+|V(\mathcal{D})|\cdot Q_\nn(\mathcal{D}))\cdot \log{C}\right)=\Ot\left(\frac{n^{3/2}}{r^{3/4}}\log{C}\right)
\end{equation*}
time.
If no negative cycle in $\DDG(G)$ is found, we still need to find a feasible price
function of $G$.
\begin{lemma}
  Given feasible price functions $p_P$ of all pieces $P\in\rdiv$, and a
  feasible price function $q$ of $\DDG(G)$, one can compute a feasible price
  function $p$ of $G$ in $O(n\log{n})+\Ot(n/\sqrt{r})$ time.
\end{lemma}
\begin{proof}
  Let $G'$ ($\DDG'(G)$) be obtained from $G$ ($\DDG(G)$, resp.) by adding a super-source~$s$
  and connecting it to all vertices $v$ of $G$ ($\DDG(G)$, resp.) using an edge of weight $0$
if $v\in \bnd{\rdiv}$ and of weight $nC$ if $v\notin \bnd{\rdiv}$ (i.e., the auxiliary edges of weight $nC$
appear only in $G'$).
We will compute single-source distances in $G'$ from $s$, as $p(v):=-\dist_{G'}(s,v)$
gives a feasible price function of $G'$, and thus also of $G$.

Observe that by Fact~\ref{l:ssp-dummy-clo}, Lemma~\ref{l:ssp-sum} and Lemma~\ref{l:dijkstra-add} combined,
  using the price function $q$ of $\DDG(G)$ (extended to $s$ so that $-q(s)$ is large enough and $q$ becomes a feasible price function of $\DDG(G)'$)
  we can compute single-source distances from $s$ on $\DDG'(G)$ in $\Ot(T_\cp(\DDG'(G)))=\Ot(n/\sqrt{r})$ time.

Now, for each piece $P\in \rdiv$, consider a graph $P'$ obtained from $P$
by adding the super-source~$s$ with edges of weight $nC$ to all $v\in V(P)\setminus \bnd{P}$,
  and with edges of weight $\dist_{\DDG'(G)}(s,v)$ to all $v\in \bnd{P}$.
Let us extend the feasible price function $p_P$ of $P$ to a feasible price function of $P'$
similarly as we did for~$q$.
Using $p_P$, using standard Dijkstra's algorithm,
we compute single-source distances $\dist_{P'}(s,v)$ to all $v\in V(P)$ and set
$p(v)=-\dist_{P'}(s,v)$.
This takes $O(|P|\log{|P|})$ time. Through all pieces, we spend $O(n\log{n})$ time on this.

To prove the above algorithm computing a feasible price function $p$ of $G'$ correct,
it is enough to argue that for a piece $P\in \rdiv$, we have $\dist_{P'}(s,v)=\dist_{G'}(s,v)$
for all $v\in V(P)$.
It is clear that $\dist_{P'}(s,v)\geq \dist_{G'}(s,v)$ since $P'$ is a subgraph of $G'$
with some edges shortcutting paths in $G'$ added.
To prove $\dist_{P'}(s,v)\leq \dist_{G'}(s,v)$, consider a shortest path $R=s\to v$ in $G'$.
If $R$ contains no vertices of $\bnd{P}$, then we have $R\subseteq P'$
so the inequality holds. Otherwise, let $z$ be the last vertex of $\bnd{P}$
appearing on $R$. Let us write $R=R_1R_2$, where $R_1=s\to z$ and $R_2=z\to v$.
Note that we have $R_2\subseteq P'$ so it is enough
to argue that $\dist_{P'}(s,z)$ is no more than the length of $R_1$.
If the first edge of $R_1$ is $sb$, where $b\in \bnd{\rdiv}$,
  then $\dist_{\DDG'(G)}(s,z)\leq \dist_{\DDG'(G)}(b,z)=\dist_{G'}(b,z)=\len(R_1)$.
  But there is an edge $sz$ of weight $\dist_{\DDG'(G)}(s,z)$ in $P'$ which proves
that $\dist_{P'}(s,z)\leq \len(R_1)$.
So suppose the first edge of $R_1$ is $sy$, where $y\notin\bnd{\rdiv}$.
But then $\len(R_1)=nC+\dist_{G}(y,z)\geq nC-(n-1)C= C>0$.
  However, we have $\dist_{P'}(s,z)\leq \dist_{\DDG'(G)}(s,z)\leq 0$ since there is
  a direct $0$-weight edge $sz$ in $\DDG'(H)$. So $\dist_{P'}(s,z)\leq \len(R_1)$ in this case as well.
\end{proof}

Let us analyze the running time $T(n)$ of the given algorithm for our choice of $r=n^{8/9}$.
We have
\begin{equation}\label{eq:ssp-time}
  T(n)=O(n\log{n})+\Ot(n^{5/6}\log{C})+ \sum_{P\in \rdiv} T(|P|).
\end{equation}

We now proceed with a formal proof of the bound $T(n)=O(n\log{(nC)})$.
Let $z\geq 1$ be a constant simultaneously larger than that hidden in the $O(n\log{n})$ term in~\eqref{eq:ssp-time},
the one in the $O(r)$ piece size bound, 
the one in the $O(n/r)$ bound on $|\rdiv|$,
and the one in the $O(\sqrt{r})$ bound on $|\bnd{P}|$.

Note that $\rdiv$ has at least $n/(2zr)$ pieces of size at least $r/(2z)$,
as otherwise the total number of vertices in pieces would be less then $(n/(2zr))\cdot zr+(zn/r)\cdot (r/2z)=n$, a contradiction.

We now prove that $T(n)\leq cn\log{(nC)}-dn^{0.9}\log{C}$ for some constants $c\geq 2$, $d>0$
by induction on~$n$.
Let $n_0$ be a constant such that for all $n\geq n_0$ we have
(1) $16z^2<n^{0.01}$,
(2) $z^2(n/\sqrt{r})(\log{n}+1)=z^2n^{5/9}(\log{n}+1)\leq n^{0.9}$,
(3) the term $\Ot(n^{5/6}\log{C})$ is no more than $n^{0.9}\log{C}$ in the right-hand side of~\eqref{eq:ssp-time}.
Then for any $P\in\rdiv$ we also have $|P|\leq zr\leq zn^{8/9}\leq n^{8/9+0.01}\leq n^{0.9}$.

Suppose first that $n>n_0$ and the bound holds for all $n'<n$. Then we have:
\begin{align*}
  T(n)&\leq O(n\log{n})+\Ot(n^{5/6}\log{C})+\sum_{P\in \rdiv}\left(c\cdot |P|(\log{|P|}+\log{C})-d\cdot|P|^{0.9}\log{C}\right)\\
  &\leq zn\log{n}+n^{0.9}\log{C}+\sum_{P\in \rdiv}c\cdot |P|\left(\log n^{9/10}+\log{C}\right)-d\cdot \frac{n}{2zr}\left(\frac{r}{2z}\right)^{0.9}\log{C}\\
  &\leq zn\log{n}+n^{0.9}\log{C}+c\left(\frac{9}{10}\log{n}+\log{C}\right)\cdot \sum_{P\in\rdiv}|P|-\frac{d}{4z^2}\frac{n}{r^{0.1}}\log{C}\\
  &\leq zn\log{n}+n^{0.9}\log{C}+c\left(\frac{9}{10}\log{n}+\log{C}\right)\cdot \left(n+z^2n/\sqrt{r}\right)-\frac{d}{4z^2}\frac{n}{r^{0.1}}\log{C}\\
  &\leq \left(z+\frac{9}{10}c\right)n\log{n}+cn\log{C}+n^{0.9}\log{C}+c\left(\frac{9}{10}\log{n}+\log{C}\right)\cdot \frac{z^2n}{\sqrt{r}}-\frac{d}{4z^2}\cdot n^{0.91}\log{C}\\
  &\leq \left(z+\frac{9}{10}c\right)n\log{n}+cn\log{C}+\left(c+1-\frac{d}{4z^2}n^{0.01}\right)n^{0.9}\log{C}.\\
\end{align*}
To finish the proof we first set the constant $c=\Omega(n_0)$ so that $z+\frac{9}{10}c\leq c$, and $cn/2$ is for $n\leq n_0$ at least as large as
the $O(n^2)$ bound on the running time of Bellman-Ford algorithm.
We would also like to have $(c+1-\frac{d}{4z^2}n^{0.01})\leq -d$, which would imply the
desired bound $cn\log{(nC)}-dn^{0.9}\log{C}$ for $n>n_0$.
For that to hold, it is enough to put $d=\frac{c+1}{(n_0^{0.01}/4z^2)-1}$. This way, by $n_0^{0.01}/4z^2\geq 4$,
 we also have $d\leq \frac{c+1}{3}\leq c/2$.
We can use this to prove the base of the induction.
Indeed, we have $cn\log{(nC)}-dn^{0.9}\log{C}\geq (c-d)n\log{C}\geq cn/2$, which, by the definition of $c$,
is enough to cover the cost of running Bellman-Ford for $n\leq n_0$.
We have thus proved the following.

\tshpath*
\section{Max $s,t$-flow oracles for planar digraphs}\label{s:flow-oracle}

\subsection{Reduction to a parametric negative cycle detection problem}\label{s:flow-reduction}

We now refer to a well-known reduction of the decision variant of the max $s,t$-flow problem on a planar digraph
to the negative cycle detection problem.
In the recent literature~\cite{Erickson10, nussbaum2014network}, this reduction is generally attributed to an unpublished manuscript by Venkatesan (see also~\cite{JohnsonV83}).

Let us denote by $u(e)\geq 0$ the capacity of an edge $e$.
Let us assume that each edge $e\in G$ has also its reverse $\rev{e}$ in $G$,
with the same embedding as~$e$, of capacity $u(\rev{e})=0$.
This does not change the amounts of flow that can be sent between the vertices
but enables the correspondence between cuts and dual cycles, and also
guarantees that each connected subgraph of $G$ (and also the dual graph) is strongly connected.
For a subset $F\subseteq E$, let $u(F)=\sum_{e\in F}u(e).$

Let $\lambda\geq 0$.
Let us embed some $s\to t$ path $Q_{s,t}$ in $G$ -- this path can be chosen in an arbitrary
way, for example it can be embedded ``almost parallel'' to some of the existing $s\to t$ paths $P$ in~$G$.
The ``capacity'' of each edge of $Q_{s,t}$ is set to $-\lambda$, whereas its reverse capacity is set to $\lambda$.
The path $Q_{s,t}$ can be likewise embedded \emph{on} $P$
in $G$ without introducing new edges or faces -- it is enough to increase the capacity
of each edge $e\in P$ by $-\lambda$, and the capacity of $\rev{e}$ by $\lambda$.

Let $G_\lambda$ be the graph obtained this way, i.e., we have some $s\to t$ path $Q_{s,t}\subseteq G$ fixed,
and $G_\lambda$ differs from $G$ only by edge weights on the path $Q_{s,t}$ and its reverse.

\venkatesan*
Whereas in the primal graph edges have capacities (as typical for max-flow problems), in the dual graph, we call them weights
(as typical in the negative cycle detection problem).

The intuition behind Lemma~\ref{l:venkatesan} is as follows (see e.g.,~\cite{Erickson10} for a formal proof).
By the max-flow-min-cut theorem, there exists a flow of value $\lambda$ in $G$
iff the minimum $s,t$-cut in $G$ has capacity~$\lambda$ or more.
It is well known that simple undirected $s,t$-cuts in $G$ correspond to simple
cycles in the dual $G^*$ that contain one of the faces $s^*,t^*$ inside
and the other outside.
However, not all directed cycles in $G^*$ with this property correspond to 
simple directed $s,t$-cuts in $G$.
Directed cycles in the dual have the ``left'' and the ``right'' side (defined
naturally once we fix what is the left side of a directed edge in~$G^*$).
Only simple cycles with $s^*$ on the left and $t^*$ on the right correspond
to simple directed $s,t$-cuts in~$G$.
For any directed \emph{simple} cycle $C^*$ in the dual, 
consider counting how many times $Q_{s,t}$ crosses it in the following way:
if $Q_{s,t}$ crosses $C^*$  left to right (i.e., an edge $e\in Q_{s,t}$ satisfies $e^*\in E(C^*)$), increment the counter
by $1$, and if $Q_{s,t}$ crosses $C^*$ right to left (i.e., an edge $e\in Q_{s,t}$ satisfies $(\rev{e})^*\in E(C^*)$) decrement the counter by $1$.
One can prove (see~\cite{Erickson10}):
\begin{fact}\label{f:crossing}
  The final value of the counter, called the \emph{crossing number} of $C^*$ wrt $Q_{s,t}$,
is in the set $\{-1,0,1\}$.
\end{fact}
The above actually holds for any Jordan curve $C^*$,
not only simple cycles in the dual.
Formally, for any subsets $A\subseteq E(G)$, $B\subseteq E(G^*)$, we will
use:
\begin{align*}
  \pi_B(A)&=|\{e\in A:e^*\in B\}|-|\{e\in A:(e^R)^*\in B\}|,\\
  \pi_A(B)&=|\{e^*\in B:e\in A\}|-|\{e^*\in B:e^R\in A\}|.
\end{align*}

With this definition, a directed dual $C^*$ corresponds to a simple directed $s,t$-cut if and only if $\pi_{Q_{s,t}}(C^*)=\pi_{C^*}(Q_{s,t})=1$.
Note that for any cycle $C^*$ in $G^*_\lambda$, its weight $w_{G^*_\lambda}(C^*)$ equals $u(C)-\lambda\cdot\pi_{Q_{s,t}}(C^*)$ (where $C=\{e\in E(G):e^*\in C^*\}$).
Thus, we have $w_{G^*_\lambda}(C^*)\geq u(C)\geq 0$ if $\pi_{Q_{s,t}}(C^*)\leq 0$
and $w_{G^*_\lambda}(C^*)=u(C)-\lambda$ if $\pi_{Q_{s,t}}(C^*)=1$.
Since simple $s,t$-cuts in~$G_\lambda$ correspond precisely to simple cycles $C^*$ with $\pi_{Q_{s,t}}(C^*)=1$ in $G^*_\lambda$,
one can prove that there exists a negative cycle in $G^*_\lambda$ if and only if there is an $s,t$-cut with
capacity less than $\lambda$ in $G$.

Note that the reduction allows binary searching for the maximum $s,t$-flow via
using negative cycle detection-based tests.
For example, using Theorem~\ref{t:shpath},
for integral capacities in $[0,C]$ one can compute a maximum $s,t$-flow
value in $O(n\log^2(nC))$ time.
However, this bound is worse and less general 
than the best known strongly polynomial $O(n\log{n})$ bound~\cite{BorradaileK09, Erickson10}.

\subsection{Combining the reduction with a decomposition}

Even though the reduction from~Section~\ref{s:flow-reduction} is not the most efficient method to compute
max $s,t$-flow in a planar graph, it turns out to be robust enough
to allow combining it with a recursive decomposition to answer queries
about max $s,t$-flow between arbitrary $(s,t)$ pairs.
For any $(s,t)$ pair, we will be able
to reduce the problem to $O(\log{n})$ negative cycle detection instances on graphs with $\Ot(\sqrt{n})$ vertices.
In each of these graphs, we will search for a negative cycle in sublinear time (in the number of edges)
even though the graph will have $\Theta(n)$ edges defined.

We will use the following recursive decomposition of a simple plane graph.
Miller~\cite{Miller84} showed how to compute, in a triangulated plane graph with $n$ vertices, a simple cycle of size $2\sqrt{2}\sqrt{n}$ that separates the graph into two subgraphs, each with at most $2n/3$ vertices. Simple cycle separators  can be used to recursively separate a planar graph until pieces have constant size.
\cite{KleinMS13} show how to obtain a complete recursive decomposition tree $\TG(G)$ of a triangulated graph~$G$
using cycle separators in $O(n)$ time. 
$\TG(G)$ is a binary tree whose nodes correspond to subgraphs of $G$ (pieces), with the root being all of~$G$ and the leaves being pieces of constant size.
We identify each piece $P$ with the node representing it in $\TG(G)$. We can thus abuse notation and write $P\in \TG(G)$.
The \emph{boundary vertices} $\bnd{P}$ of a \emph{non-leaf piece} $P$ are vertices that $P$ shares
with some other piece $Q\in \TG(G)$ that is neither $P$'s ancestor nor its descendant.
We assume $P$ to inherit the embedding from $G$. The faces of $P$
that are faces of $G$ are called \emph{natural}, whereas the faces of $P$ that
are not natural are the \emph{holes} of~$P$.
The construction of~\cite{KleinMS13} additionally guarantees
that for each piece $H\in\TG(G)$:
\begin{enumerate}[(a)]
  \item $H$ is connected and contains at least one natural face,
  \item if $H$ is a non-leaf piece, then each natural face $f$ of $H$ is a face of
precisely one child of~$H$,
  \item $H$ has $O(1)$ holes containing precisely the vertices~$\bnd{H}$.
\end{enumerate}

\newcommand{\hex}{\ensuremath{h^{\mathrm{ex}}}}
\newcommand{\bex}{\ensuremath{\partial^{\mathrm{ex}}}}

Note that since the natural faces of a piece are partitioned among its children, and there are $O(1)$ holes per piece,
the subtree of $\TG(G)$ rooted at $H$ has $O(|V(H)|)$ nodes.
Throughout, to avoid confusion, we use \emph{nodes} when referring to $\TG(G)$ and \emph{vertices} when referring to $G$
or its subgraphs.
It is well-known~\cite{BorradaileSW15, vorexact, KleinMS13} that by suitably choosing
cycle separators that alternate between balancing vertices, boundary vertices, and holes,
one can also guarantee that (1) $\sum_{H\in\TG(G)}|H|=O(n\log{n})$, (2) $\sum_{H\in\TG(G)}|\bnd{H}|^2=O(n\log{n})$, and (3) $|\bnd{H}|=O(\sqrt{n})$.

To proceed, we first need to guarantee that $G$ has degree~$3$. This is easily
achieved by standard embedding-respecting transformations. First, we repeatedly
introduce parallel $0$-capacity edges incident to vertices of degree less than $3$,
until there are none.
Afterwards, we replace each vertex by a cycle of sufficiently large-capacity edges. The graph $G$
grows by a constant factor only, and clearly for any pair of vertices,
the max $s,t$-flow value in $G$ before the transformation equals
the max $s',t'$-flow values between some copies $s'$ of $s$ and $t'$ of $t$
after the transformation. Therefore, in the following we assume that each vertex in $G$ has degree $3$.

Observe that by the degree-$3$ assumption, the dual graph $G^*$ is triangulated.
We can thus build a recursive decomposition $\TG(G^*)$ of $G^*$ as described above.

For simplicity, let also additionally assume that the obtained decomposition
satisfies the following: for each node $H\in \TG(G^*)$
the holes of $H$ are \emph{simple and pairwise vertex-disjoint}.
We discuss how to drop this assumption in Section~\ref{s:non-simple}.

Recall that the faces of $G^*$ correspond to the vertices of $G$.
As a result, for a node $H\in\TG(G^*)$, every natural face of $H$
corresponds to some vertex of $G$. Let us denote by $W(H)$ the vertices
of~$G$ whose duals in $G^*$ are natural faces of $H$.
We use $W^*(H)$ when referring to natural faces of $H$.

Since a piece $H\in \TG(G^*)$ has simple and vertex-disjoint holes, for any of $H$'s two natural faces $f=u^*,g=v^*$,
corresponding to vertices $u,v\in V$ respectively,
there is a chain of distinct natural faces $f=f_1,\ldots,f_{k+1}=g$ in~$H$
such that $f_i$ and $f_{i+1}$ share an edge $e_i^*$ in~$H$.
To see this, note that for each hole $h$, any two of its neighboring (natural) faces
can be connected with a face chain consisting of the neihboring (natural) faces of $h$ only.
Therefore, any path between two natural faces of $H$
can be transformed to avoid the holes.
A chain $f_1,\ldots,f_{k+1}$ in question corresponds to a \emph{simple} path between $u\to v$ in $G$.
If $H$ is a \emph{leaf piece}, let us fix an arbitrary one such chain and denote by $Y_H[f,g]$ the corresponding $u\to v$ path
$e_1\ldots e_k$ (where $e_i^*$ is the dual of $e_i$) in $G$. 

\paragraph{Split faces and edges.}
For each non-leaf node $H\in \TG(G^*)$, let us arbitrarily pick two natural \emph{split faces} $f_{H,1},f_{H,2}$ of the children nodes $H_1,H_2$ of $H$ satisfying the following:
$f_{H,i}$ is a face of $H_i$ but
not of $f_{H,{3-i}}$ and $f_{H,1},f_{H,2}$ are adjacent faces of $H$.
To see that such a pair exists, recall
that both $H_1$ and $H_2$ contain at least one natural face and the natural faces of $H$ are partitioned among $H_1$ and $H_2$. Consider a simple chain of faces of $H$ connecting a natural
face of $H_1$ to a natural face of $H_2$ (that we have already argued to exist).
Such a chain has to contain a natural face of $H_1$ followed by a natural face of $H_2$.
We also say that $f_{H,i}$ is a \emph{sibling split face} of $f_{H,3-i}$.

Denote by $e_H^*$ the common edge of $f_{H,1}$ and $f_{H,2}$, called a \emph{dual split edge}.
The \emph{primal split edge} $e_H$ is such that $e_H^*$ is the dual of $e_H$.

\paragraph{Choosing and decomposing the primal $s,t$-path.} Having defined split faces, we will now fix, for any pair $(s,t)\in V\times V$,
a particularly convenient simple $s\to t$ path $Q_{s,t}\subseteq G$ which is required by the reduction of Section~\ref{s:flow-reduction}.

\newcommand{\pth}{\Phi}

We will set $Q_{s,t}$ to be the path $\pth(G^*,s^*,t^*)$ constructed using a recursive function $\pth(H,x^*,y^*)$
that, given a piece $H\in\TG(G^*)$ and its two natural faces $x^*, y^*\in W^*(H)$, produces a simple path $x\to y$
in~$G$ whose edges' duals come from $E(H)$.
Let $H_1,H_2$ denote the children of $H$, if $H$ is non-leaf.
We set:
\begin{equation}\label{eq:path}
  \pth(H,x^*,y^*)=\begin{cases} Y_H[x^*,y^*] &\text{ if }H\text{ is a leaf piece of }\TG(G^*),\\
    \pth(H_i,x^*,y^*) &\text{ if $x,y\in W(H_i)$},\\
    \pth(H_i,x^*,f_{H,i})\cdot e_H\cdot \pth(H_{3-i},f_{H,{3-i}},y^*) &\text{ $x\in W(H_i)$ and $y\in W(H_{3-i})$}.
  \end{cases}
\end{equation}

Since the sets $W(H_1),W(H_2)$ are disjoint if $H_1,H_2$ are children of the same parent, we have:
\begin{fact}
  For any $H\in\TG(G^*)$, and $x,y\in W(H)$, $\pth(H,x^*,y^*)$ is a simple $x\to y$ path in $G$.
\end{fact}

Recall that for the purpose of using the reduction of Section~\ref{s:flow-reduction}, ideally, we would like to be able to
increase the capacities of the edges on $Q_{s,t}$ by $-\lambda$, and of the reverse edges on $Q_{s,t}$ by $\lambda$.
Unfortunately, the path $Q_{s,t}$ can be in general of length $\Theta(n)$
which seems prohibitive if we want to explicitly adjust the weights upon source/sink query.

We will now prove that the possible paths $Q_{s,t}$ constructed via $\pth(G^*,s^*,t^*)$ have many subpaths in common.
For any $H\in\TG(G^*)$, define the \emph{split set} $Z(H)$ to contain
all the natural faces $f$ of $H$ such that~$f$ is a split face of
one of $H$'s ancestors in $\TG(G^*)$.
We would like each element of $Z(H)$ to also record
the ancestor $A$ where the split face originates from,
so technically $Z(H)$ might be a multiset.
Additionally, it will be convenient to assume that the faces
of $Z(H)$ are not considered ordinary natural faces of $H$ since they carry more information (about the ancestor).
As a result, we will assume that $W^*(H)\cap Z(H)=\emptyset$.
However, if $x^*\in W^*(H)$, we might still say that $x^*$ is an \emph{ancestor split face}
if $Z(H)$ contains at least one element based on the face $x^*$.
Similarly, for $z^*\in Z(H)$, $z^*$ can be projected on $W^*(H)$ (converted to an element of $W^*(H)$)
by dropping the ancestor information in a trivial way.

Observe that for any $H$ we have $|Z(H)|=O(\log{n})$ since $H$ has $O(\log{n})$ ancestors
and each of them contributes at most a single split face.
As a result, $W^*(H)$ contains $O(\log{n})$ ancestor split faces.
Moreover, note that for any $f_{A,i}\in Z(H)$, where $A$
is some ancestor of $H$, $f_{A,3-i}$ is not a natural face of $H$
and the dual split edge $e_A^*$ lies on the boundary of some hole of $H$.
As a result, the primal split edge $e_A$ crosses the boundary of that hole of $H$.

The following lemma gives a decomposition of any of $\Omega(n^2)$ possible $Q_{s,t}$ paths into 
a small number of subpaths coming from a set of only $\Ot(n)$ distinct paths in $G$.

\begin{lemma}\label{l:path-decomp}
  Any path $\pth(G^*,s^*,t^*)$ can be expressed as a concatenation
  of $O(\log{n})$ paths, each of which is either
  \begin{enumerate}
    \item a split edge,
    \item or is of the form $\pth(H,x^*,y^*)$, where $x^*,y^*\in W^*(H)$ are ancestor split faces of $H$,
    \item or is of the form $\pth(L,x^*,y^*)$, where $L$ is a leaf of $\TG(G^*)$.
  \end{enumerate}
    Such a decomposition can be computed in $O(\log{n})$ time.
\end{lemma}
\begin{proof}
  We prove that a desired decomposition exists for all paths $\pth(H, x^*, y^*)$, where $H\in \TG(G^*)$ and $x^*,y^*\in W^*(H)$, by induction
  on the height $d$ of the subtree of $\TG(G^*)$ rooted at $H$.
  The sought decomposition will have at most $4d+1$ elements if none of $x^*,y^*$ are ancestor split faces of $H$,
  at most $2d+1$ elements if at least one of $x^*,y^*$ is an ancestor split face of $H$,
  and precisely $1$ element if both $x^*,y^*$ are ancestor split faces of $H$.
  The lemma will follow as the height of $\TG(G^*)$ is $O(\log{n})$.
  
  If $H$ has no children, i.e., $d=0$, then $H$ is a leaf and thus $\pth(H,x^*,y^*)$ is a valid decomposition
  of size $1$, so the desired bound holds in all cases.
  So in the following assume $d\geq 1$.

  If both $x^*$ and $y^*$ are ancestor split faces, we don't need to decompose $\pth(H, x^*, y^*)$, 
  so a decomposition of size $1\leq 2d+1$ exists.
  Otherwise, we apply Equation~\eqref{eq:path} and recursively decompose
  the right-hand side of~\eqref{eq:path}.

  If neither $x^*$ nor $y^*$ is an ancestor split face, then there is either one recursive call
  of the same kind that yields a decomposition of size $4(d-1)+1\leq 4d+1$,
  or two recursive calls with at least one ancestor split face
  that, combined, yield a decomposition of size $2\cdot (2(d-1)+1)+1=4d-1\leq 4d+1$.
  
  Finally, let us assume that exactly one of $x^*,y^*$ is an ancestor split face.
  Suppose that $H$ has children $H_1,H_2$. If both $x^*,y^*\in W^*(H_i)$ for some $i$,
  then $\pth(H,x^*,y^*)=\pth(H_i,x^*,y^*)$.
  Moreover, if $x^*$ or $y^*$ is an ancestor split face of $H$, then it is an ancestor split face of $H_i$ as well.
  As a result, among $x^*,y^*$ there is at least one ancestor split face of $H_i$, so we can apply the
  inductive assumption
  and the decomposition of $\pth(H,x^*,y^*)$ has length at most $2d-1\leq 2d+1$.

   Otherwise, $\pth(H,x^*,y^*)=\pth(H_i,x^*,f_{H,i})\cdot e_H\cdot \pth(H_{3-i},f_{H,{3-i}},y^*)$.
  Suppose wlog. that $x^*$ is an ancestor split face of $H$.
  Then, both $x^*,f_{H,i}$ are ancestor split faces of $H_i$, so the decomposition of $\pth(H_i,x^*,f_{H,i})$
  has size $1$. By the fact that $f_{H,{3-i}}$ is an ancestor split face of $H_{3-i}$ and the inductive assumption, the decomposition of 
  $\pth(H_{3-i},f_{H,{3-i}},y^*)$ has size at most $2d-1$. 
  So indeed there exists a decomposition of 
  $\pth(H,x^*,y^*)$ of size at most $1+1+(2(d-1)+1)=2d+1$.

  Clearly, the above inductive proof can be turned into an $O(\log{n})$ time algorithm expressing
  $\pth(G^*,s^*,t^*)$ as desired.
\end{proof}

\subsection{A feasible flow oracle}\label{s:feasible-flow-oracle}

Let us first consider the problem of designing a data structure that supports queries $(s,t)$ whether
the max $s,t$-flow is at least $\lambda$, where $\lambda$ is a fixed parameter
given upon initialization.
\subsubsection{Preprocessing}
For each $H\in \TG(G^*)$, we will explicitly compute the $O(|Z(H)|^2)=\Ot(1)$ paths of the
form $\pth(H,x^*,y^*)$, where $x^*,y^*$ are ancestor split faces of $H$.
Note that computing $\pth(H,x^*,y^*)$ from the definition takes $O(|V(H)|)$ time.
As a result, computing all these required paths explicitly
takes $\sum_{H\in \TG(G)}|Z(H)|^2\cdot |H|=\Ot(1)\cdot \sum_{H\in \TG(G)}|H|=\Ot(n)$ time.

To proceed, we first need to extend the definition of $\pth(H,x^*,y^*)$ from $x^*,y^*\in W^*(H)$
to all \linebreak $x^*,y^*\in W^*(H)\cup Z(H)$.
If $x^*\in Z(H)$ ($y^*\in Z(H)$), then let $(x')^*$ ($(y')^*$, resp.) be the sibling split face of $x^*$ ($y^*$, resp).
We extend the definition naturally as follows:
\begin{equation*}
  \pth(H,x^*,y^*)=\begin{cases}
    x'x\cdot \pth(H,x^*,y^*) \cdot yy' & \text{ if $x^*\in Z(H)$, $y^*\in Z(H)$, }\\
    \pth(H,x^*,y^*) \cdot yy' &\text{ if $x^*\in W^*(H)$, $y^*\in Z(H)$, }\\
    x'x\cdot \pth(H,x^*,y^*) &\text{ if $x^*\in Z(H)$, $y^*\in W^*(H)$. }
  \end{cases}
\end{equation*}
In the right-hand side above, where needed, the faces $x^*,y^*$ are projected naturally onto $W^*(H)$
as explained before.

For any $H\in\TG(G^*)$ and $x^*,y^*\in Z(H)$, where $x^*\neq y^*$, we apply the following additional preprocessing.
Observe that $\pth(H,x^*,y^*)$ is a simple $x'\to y'$ path in $G$
whose all vertices but its endpoints correspond to natural faces of $H$.
Moreover, the dual edges $(\pth(H,x^*,y^*))^*$ of $\pth(H,x^*,y^*)$
satisfy $(\pth(H,x^*,y^*))^*\subseteq E(H)$.
From the point of view of $H$, $\pth(H,x^*,y^*)$
is a simple directed curve in the plane
\begin{enumerate}[label=(\roman*)]
  \item originating in a hole $h_x$ of $H$ that
contains the primal vertex $x'$ (as embedded in~$G$),
  \item ending in the hole $h_y$
that contains the primal vertex $y'$,
  \item going
only through natural faces of~$H$ between leaving $h_x$
and entering $h_y$.
\end{enumerate}

Let $G^*_\lambda[H](x^*,y^*)$ denote the edge-induced subgraph $G^*[H]$
such that the weight of edges $(\pth(H,x^*,y^*))^*$ have been increased
by $-\lambda$, and their reverse edges by $\lambda$.

First, we test whether $G^*_\lambda[H](x^*,y^*)$ contains a negative cycle and 
potentially also compute a feasible price function $p$ of that graph.
For this, we can
either use Theorem~\ref{t:negcyc} or the strongly polynomial
algorithm of~\cite{MozesW10} run on that graph
whose size is $O(|H|)$.
Moreover, with the help of the feasible price function $p$ and the MSSP data structure,
we compute the dense distance graph
$\DDG(G^*_\lambda[H](x^*,y^*))$ (with vertex set $\bnd{H}$, defined analogously as in Section~\ref{s:negcyc})
in $\Ot(|H|+|\bnd{H}|^2)$ time.
Since Lemma~\ref{l:ddg-ds} only required the boundary vertices to lie on $O(1)$ faces
of the piece, using Lemma~\ref{l:ddg-ds}, we preprocess the closest pair and near
neighbor data structures for $\DDG(G^*_\lambda[H](x^*,y^*))$ in
$\Ot(|H|+|\bnd{H}|^2)$ time and space.

Finally, for each piece $H\in\TG(G^*)$, we also precompute and store
the dense distance graph $\DDG(H)$ along with the closest pair
and near neighbor data structures using MSSP and Lemma~\ref{l:ddg-ds}.
This requires $\Ot(|H|+|\bnd{H}|^2)$ time and space as well.

\begin{lemma}\label{l:flow-preproc}
  Preprocessing takes $\Ot(n)$ time.
\end{lemma}
\begin{proof}
  The cost of preprocessing (through all pieces and pairs of split faces) can be bounded by:
  \begin{equation*}
    \sum_{H\in \TG(G)}\Ot\left(|Z(H)|^2\cdot \left(|H|+|\bnd{H}|^2\right)\right)=\Ot(1)\cdot \sum_{H\in \TG(G)}\left(|H|+|\bnd{H}|^2\right)=\Ot(n).\qedhere
  \end{equation*}
\end{proof}
\subsubsection{Query procedure}
We are now ready to give an algorithm detecting a negative
cycle in $G^*_\lambda$ given the source/sink pair~$s,t$. By the reduction of Section~\ref{s:flow-reduction}, this is equivalent
to deciding whether the max $s,t$-flow in $G$ is less than $\lambda$.
Recall that we have fixed the path $Q_{s,t}$ in $G$ (as required by the reduction) to be specifically the path $\pth(G^*,s^*,t^*)$.

\newcommand{\detect}{\mathtt{detect}}

Similarly as before, 
for $H\in\TG(G)$, $x^*,y^*\in W^*(H)\cup Z(H)$, let $G^*_\lambda[H](x^*,y^*)$ denote the edge-induced subgraph $G^*[H]$
such that the weights of edges $(\pth(H,x^*,y^*))^*$ have been increased
by $-\lambda$, and the weights of their reverse edges by $\lambda$.

We now define a recursive procedure $\detect(\lambda, H, x^*,y^*)$ that 
tests whether $G^*_\lambda[H](x^*,y^*)$ has a negative cycle.
Equipped with this procedure, by calling $\detect(\lambda,G^*,s^*,t^*)$, we will
achieve our goal of testing whether $G^*_\lambda$ contains a negative cycle.

The procedure $\detect(\lambda,H,x^*,y^*)$ will decompose the problem in exactly the same way
as definition~\eqref{eq:path} deconstructs the path $\pth(H,x^*,y^*)$ (with $x^*,y^*$ projected to $W^*(H)$)
and stop recursing when we have $x^*,y^*\in Z(H)$ or $H$ is a leaf.
By Lemma~\ref{l:path-decomp}, the recursion tree will have $O(\log{n})$ leaves and $O(\log{n})$ depth,
and thus the number of recursive calls will be $O(\log{n})$.

If $\detect(\lambda,H,x^*,y^*)$ does not detect a negative cycle, it produces (pointers to) preprocessed closest
pair and near neighbor data structures of a graph $D(H,x^*,y^*)$
with $\bnd{H}\subseteq V(D(H,x^*,y^*))$ such that 
for any $u,v\in \bnd{H}$ we have
\begin{equation*}
  \dist_{D(H,x^*,y^*)}(u,v)=\dist_{G^*_\lambda[H](x^*,y^*)}(u,v),
\end{equation*}
and for all $u,v\in V(D(H,x^*,y^*))$ we have
\begin{equation*}
  \dist_{D(H,x^*,y^*)}(u,v)\geq \dist_{G^*_\lambda[H](x^*,y^*)}(u,v).
\end{equation*}
In other words, $D(H,x^*,y^*)$ preserves the boundary-to-boundary distances of $G^*_\lambda[H](x^*,y^*)$,
and does not underestimate any other distances in $G^*_\lambda[H](x^*,y^*)$.
The graphs $D(H,x^*,y^*)$ are not constructed explicitly in principle;
we define them for the purpose of analysis and we operate on the associated data structures on these graph instead.

We now proceed with describing how the procedure $\detect(\lambda,H,x^*,y^*)$ works.
\begin{enumerate}
  \item  If $H$ is a leaf, we run any negative cycle detection
algorithm on $G^*_\lambda[H](x^*,y^*)$. Since $H$ has $O(1)$ size, this takes $O(1)$ time.
    If no negative cycle is found, we can put\linebreak ${D(H,x^*,y^*):=G^*_\lambda[H](x^*,y^*)}$ 
    and preprocess the required data structures for $D(H,x^*,y^*)$ in $O(1)$ time.

\item If $x^*,y^*\in Z(H)$, we simply look up the precomputed
  information whether $G^*_\lambda[H](x^*,y^*)$ has a negative cycle. If so,
we have found a negative cycle in $G^*_\lambda$ and the procedure terminates globally.
Otherwise, we return the appropriate preprocessed data structures
    for $D(H,x^*,y^*)$ set to $\DDG(G^*_\lambda[H](x^*,y^*))$.

\item Otherwise, if $x^*,y^*\in W^*(H_i)\cup Z(H_i)$ for some $i$, we recurse
using $\detect(\lambda,H_i,x^*,y^*)$.
    Note that in this case we have
    \begin{equation}\label{eq:case3}
      G^*_\lambda[H](x^*,y^*)=G^*_\lambda[H_i](x^*,y^*)\cup H_{3-i}.
    \end{equation}
If the recursive call does not find a negative cycle, let us set $D_i$
to be the graph $D(H_i,x^*,y^*)$ associated with the returned data structures.
Moreover, let us set $D_{3-i}:=\DDG(H_{3-i})$.

\item Finally, if $x^*\in W^*(H_i)\cup Z(H)$ and $y^*\in W^*(H_{3-i})\cup Z(H_{3-i})$,
then we perform recursive calls
$\detect(H_i,\lambda,x^*,f_{H,i})$ and $\detect(H_{3-i},\lambda,f_{H,3-i},y^*)$,
where $f_{H,i}\in Z(H_i)$ and \linebreak $f_{H,3-i}\in Z(H_{3-i})$.
Note that in this case we have
    \begin{equation}\label{eq:case4}
      G^*_\lambda[H](x^*,y^*)=G^*_\lambda[H_i](x^*,f_{H,i})\cup G^*_\lambda[H_{3-i}](f_{H,3-i},y^*).
    \end{equation}

If none of the recursive calls detects a negative cycle, let us set
    $D_i:=D(H_i,x^*,f_{H,i})$ and $D_{3-i}:=D(H_{3-i},f_{H,3-i},y^*)$.
\end{enumerate}

In cases $3$ and $4$, we proceed with additional computation: even though the negative
cycle does not exist in the parts of $G^*_\lambda[H](x^*,y^*)$ processed
in the recursive calls, it might still exist in $G^*_\lambda[H](x^*,y^*)$.
But then, observe that such a cycle $C$ has to have 
parts in both graphs on the respective right-hand sides of~\eqref{eq:case3} (in case 3)
and~\eqref{eq:case4} (in case 4).
Since both graphs intersect only in the vertices
of $\bnd{H_1}\cup \bnd{H_2}$, it is enough to search for such a negative cycle
in $G^*_\lambda[H](x^*,y^*)$ that goes through at least one vertices of $\bnd{H_1}\cup \bnd{H_2}$.
Since the graphs $D_1,D_2$ preserve distances between $\bnd{H_1}$ and $\bnd{H_2}$, resp., in
these graphs, similarly as in Section~\ref{s:negcyc}, we can look
for a negative cycle in $D(H,x^*,y^*):=D_1\cup D_2$ instead.
Note that $D(H,x^*,y^*)$ defined this way meets the requirement
that it preserves distances between $\bnd{H}\subseteq \bnd{H_1}\cup \bnd{H_2}$
in $G^*_\lambda[H](x^*,y^*)$ and
does not underestimate other distances $G^*_\lambda[H](x^*,y^*)$.
The correctness of this approach follows.

Recall that from the recursive calls we get the closest pair and near neighbor
data structures for $D_1,D_2$, so the respective data structures
for $D(H,x^*,y^*)$ can be obtained using Lemmas~\ref{l:ssp-sum}~and~\ref{l:nn-sum}.
Let us put $D^*:=D(H,x^*,y^*)$. We then have:
\begin{align*}
  T_\cp(D^*)&=T_\cp(D_1)+T_\cp(D_2)+\Ot\left(|\bnd{H_1}|+|\bnd{H_2}|\right),\\
    T_\nn(D^*)&=T_\nn(D_1)+T_\nn(D_2)+Q_\nn(D_1)\cdot |\bnd{H_1}|+Q_\nn(D_2)\cdot |\bnd{H_2}|+\Ot\left(|\bnd{H_1}|+|\bnd{H_2}|\right),\\
    Q_\nn(D^*)&=\max(Q_\nn(D_1),Q_\nn(D_2))+O(1).
\end{align*}
Consequently, running the algorithm of Theorem~\ref{t:negcyc} on the graph $D(H,x^*,y^*)$
takes
\begin{equation}\label{eq:piece-negcyc}
  \Ot\left(\sqrt{|V(D^*)|}\cdot (T_\cp(D^*)+T_\nn(D^*)+|V(D^*)|\cdot Q_\nn(D^*))\log{C}\right)
\end{equation}
time. The following lemma analyzes the running time of the call $\detect(G^*,s^*,t^*)$.
\begin{lemma}\label{l:detect-time}
  $\detect(\lambda,G^*,s^*,t^*)$ runs in $\Ot(n^{3/4}\log{C})$ time.
\end{lemma}
\begin{proof}
  Consider the leaves of the recursion tree. Each leaf corresponds precisely
  to a single element $\pth(H,x^*,y^*)$ of the decomposition of $\pth(G^*,s^*,t^*)$ produced by Lemma~\ref{l:path-decomp}
  that is not a split edge. Consequently, there are $O(\log{n})$ leaves in the recursion tree.
  As the height of $\TG(G)$ is $O(\log{n})$, there are $\Ot(1)$ recursive calls in total.

  The leaf calls (i.e., cases 1~and~2 of $\detect$) are processed in $O(1)$ time since we have the necessary information
  precomputed. 
  Note that the associated precomputed graph $D(H,x^*,y^*$) (equal to either $\DDG(H,x^*,y^*)$ or $G^*_\lambda[H](x^*,y^*)$)
  has $O(|\bnd{H}|)$ vertices.

  Each leaf call can be seen to
  contribute the graph $D(H,x^*,y^*)$ to the graph $D^*$ of each of its $O(\log{n})$ ancestor calls.
  Each non-leaf call with a single child call (case 3), in turn, additionally contributes a
  single dense distance graph $\DDG(H_{3-i})$ to the graph $D^*$ of each of its $O(\log{n})$ ancestor calls.
  In other words, the graph $D(H,x^*,y^*)$ can be seen to be the union of the
  respective graphs $D^*_L=D(L,\cdot,\cdot)$ from $O(\log{n})$ descendant leaf calls $L$,
  and dense distance graphs from $\Ot(1)$ non-leaf descendant calls.
  By the properties of $\TG(G^*)$, each of these precomputed graphs $X$ whose union forms $D(H,x^*,y^*)$
  satisfies $T_\cp(X)=\Ot(\sqrt{n})$, $T_\nn(X)=\Ot(\sqrt{n})$ and $Q_\nn(X)=\Ot(1)$.
  Therefore, we conclude that $D(H,x^*,y^*)$ has $\Ot(\sqrt{n})$ vertices.
  
  Let us now analyze the performance of closest pair and near neighbor data structures
  on the used graphs $D(H,x^*,y^*)$.
  By Lemma~\ref{l:nn-sum}, the quantity $Q_\nn(D(H,x^*,y^*))$ can be seen to be the maximum of $Q_\nn(X)$
  over all precomputed graphs $X$ contributing to $D(H,x^*,y^*)$, plus $\Ot(1)$.
  As a result, we have $Q_\nn(D(H,x^*,y^*))=\Ot(1)$.
  $T_{\cp}(D(H,x^*,y^*))$ can be bounded by the sum of $T_{\cp}(X)$ over such graphs $X$
  plus a term near-linear in the total size of boundaries of all graphs in the
  $\Ot(1)$ descendant calls of $\detect(\lambda,H,x^*,y^*)$.
  As a result, we also have $T_{\cp}(D(H,x^*,y^*))=\Ot(\sqrt{n})$.
  By $Q_\nn(D(H,x^*,y^*))=\Ot(1)$ and a similar argument, we obtain
  $T_{\nn}(D(H,x^*,y^*))=\Ot(\sqrt{n})$ as well.
  Consequently, by plugging in these bounds into~\eqref{eq:piece-negcyc}, detecting a negative
  cycle in $D(H^*,x^*,y^*)$ takes $\Ot(n^{3/4}\log{C})$ time.
  Since the total number of recursive calls is $\Ot(1)$, the lemma follows.
\end{proof}

By combining Lemmas~\ref{l:flow-preproc} and~\ref{l:detect-time}, we obtain the following.

\tfeasibleflow*

\begin{remark}\label{r:report-cut}
  The data structure of Theorem~\ref{t:feasible-flow} can be also extended to report an $s,t$-cut-set
$F$
  of capacity less than~$\lambda$ in $\Ot(|F|)=\Ot(\lambda)$ additional time, if such a cut-set exists.
\end{remark}
\begin{proof}[Proof sketch]
If the algorithm of Theorem~\ref{t:negcyc} used
inside $\detect(\lambda,H,x^*,y^*)$ detects a negative cycle in $D(H,x^*,y^*)$,
it can also report one such cycle $C$ in $\Ot(n^{3/4}\log{C})$ time.
Recall that $D(H,x^*,y^*)$ is a union of dense distance graphs,
so each edge of $C$ corresponds to some path between boundary vertices
  of a graph (say $G^*_\lambda[H](x^*,y^*)$ or $H_{3-i}$) from which the respective dense distance
graphs is derived from. Since a DDG
is constructed using MSSP, we can reuse MSSP to expand
each edge $e\in C$ into an actual path $P\subseteq G^*$.
We can even extend the dynamic tree used by MSSP (in a standard way,
e.g., by enabling reporting a maximum weight edge on a tree path) so that only the~$p$ positive-weight
edges of $P$ are reported in $\Ot(p)$ time (recall that we added some
zero-capacity reverse edges to $G$ to make use of duality).
This way, an $s,t$ cut-set $F\subseteq G$ of capacity less than $\lambda$ can
be reported in $\Ot(n^{3/4}\log{C}+|F|)$ time.
Observe that since each edge in~$F$ has capacity at least $1$, we have $|F|\leq \lambda$.
\end{proof}

\subsection{An approximate max $s,t$-flow oracle}\label{s:flow-approx}

Let $\eps\in (0,1)$ be given. The feasible flow oracle from Theorem~\ref{t:feasible-flow} can be easily converted into
a $(1-\eps)$-approximate max $s,t$-flow oracle (or a $(1+\eps)$-approximate min $s,t$-cut oracle).

We simply set up feasible flow oracles for each of the values $\lambda_j=\lceil(1+\eps)^j\rceil$, where
\linebreak $j=0,\ldots,\log_{1+\eps}(nC)=O(\log{(nC)}/\eps)$.
To find $(1-\eps)$-approximate max $s,t$-flow value, we simply binary search for (and return) the largest
$\ell$ such that there exists an $s,t$-flow of value $\lambda_\ell$ in~$G$.
To this end, $O(\log(\log{(nC)}/\eps))=O(\log\log{(nC)}+\log{(1/\eps)})$ queries
to feasible flow oracles are enough.
If the max $s,t$-flow value is $\lambda^*$, then we have
$\lambda^*< \lceil (1+\eps)^{\ell+1}\rceil$, which implies
\begin{equation*}
  \lambda^*\leq (1+\eps)^{\ell+1}\leq (1+\eps)\cdot \lceil (1+\eps)^\ell\rceil=(1+\eps)\cdot \lambda_\ell.
\end{equation*}
So $\lambda^*\geq \lambda_\ell\geq \frac{1}{1+\eps}\cdot \lambda^*\geq (1-\eps)\cdot \lambda^*$.

\capproxoracle*

For approximating the min $s,t$-cut value,
note that there exists an $s,t$-cut set $F$ with capacity less than $\lambda_{\ell+1}$,
i.e., at most $\lceil (1+\eps)^{\ell+1}\rceil -1\leq \lfloor (1+\eps)^{\ell+1}\rfloor\leq (1+\eps)\cdot \lambda_{\ell}\leq (1+\eps)\cdot \lambda^*$.
By Remark~\ref{r:report-cut}, we can also find such $F$ in additional $\Ot(n^{3/4}\log{C}+|F|)=\Ot(n^{3/4}\log{C}+\lambda^*)$ time.

\subsection{A dynamic feasible flow oracle}\label{s:dynamic-oracle}

In this section we explain how to convert the static feasible flow oracle into a dynamic one
that supports edge capacity updates. For this, we use a standard trick that have also been used
to convert an $\Ot(\sqrt{n})$-query time static exact distance oracle into an $\Ot(n^{2/3})$
update/query time dynamic distance oracle~\cite{FR06, Klein05}.

Let $r$ be a parameter to be set later.
As proved by~\cite{KleinMS13}, the recursive decomposition $\TG(G^*)$ \emph{admits an $r$-division}. This
means that one can pick a subset $\rdiv$ of pieces from $\TG(G^*)$,
such that (1) each face of $G^*$ is a natural face of precisely one piece $H\in \rdiv$,
(2) $|\rdiv|=O(n/r)$, (3) for every $P\in\TG(G)$, $|P|=O(r)$ and $|\bnd{P}|=O(\sqrt{r})$.
This can be done in linear time~\cite{KleinMS13} during the preprocessing phase.
Since the bounds $|\bnd{H}|=O(\sqrt{n})$ and $\sum_{H\in\TG(G^*)}|H|+|\bnd{H}|^2=\Ot(n)$
still hold even if we start the decomposition process
with $O(\sqrt{n})$ boundary vertices $\bnd{G^*}$ on $O(1)$ holes of $G^*$,
and each $P\in\rdiv$ is decomposed using exactly the same procedure as $G^*$,
for descendants $B$ of $P$ we have
$|\bnd{B}|=O(\sqrt{r})$, and $\sum_{P\supseteq B\in\TG(G^*)}|B|+|\bnd{B}|^2=\Ot(r)$.

Suppose the decomposition of $\pth(G^*,s^*,t^*)$ from Lemma~\ref{l:path-decomp}
is not allowed to have elements of the form $\pth(H,x^*,y^*)$, where $x^*,y^*\in Z(H)$,
unless $H$ is a descendant of some piece $P\in \rdiv$.
If $H$ is not a descendant of some piece in $\rdiv$, by further decomposing $\pth(H,x^*,y^*)$ using the recursive definition~\eqref{eq:path}
until we arrive at descendants $\pth(H',z^*,w^*)$ with $H'\in \rdiv$ (note that $z^*,w^*\in Z(H')$ necessarily holds
in all such descendants),
we end up with a decomposition of size $\Ot(n/r)$ instead of $O(\log{n})$:
there are $O(n/r)$ distinct pieces in $\rdiv$ and to every piece $H'\in \rdiv$ in the decomposition one can apply Lemma~\ref{l:path-decomp}
to see that it decomposes into $O(\log{n})$ elements.

By modifying $\detect$ easily so that it obeys the above $\rdiv$-aware decomposition of $\pth(G^*,s^*,t^*)$,
we obtain that $\Ot(n/r)$ preprocessed graphs
of the form either $\DDG(H',z^*,w^*)$ or $\DDG(H')$, where $H'\in \rdiv$,
contribute to $O(\log{n})$ negative cycle detection algorithm applications on some
graph $X=D(H,x^*,y^*)$.
As a result, the total number of vertices in such
graphs $X$ is  $\Ot(n/r)\cdot O(\sqrt{r})=\Ot(n/\sqrt{r})$,
and by a similar analysis we have 
$T_\cp(X)=\Ot(|V(X)|),T_\nn(X)=\Ot(|V(X)|)$, and $Q_\nn(X)=\Ot(1)$.
Since the bottleneck of $\detect$ lies in using
the algorithm behind Theorem~\ref{t:negcyc} on such graphs $X$, $\detect$ requires
\begin{equation*}
  \sum_{X} \Ot\left(\sqrt{|V(X)|}\cdot |V(X)|\log{C}\right)=\Ot\left(\sqrt{\frac{n}{\sqrt{r}}}\log{C}\right) \cdot \sum_X \Ot(|V(X)|)=\Ot\left(\frac{n^{3/2}}{r^{3/4}}\log{C}\right).
\end{equation*}

Observe that although choosing a small parameter $r$ increases the running time of $\detect$, it causes that we only use the preprocessed data
for pieces that constitute weak descendants of the pieces in $H\in \rdiv$.
When an edge capacity is updated, only $O(1)$ pieces of $\rdiv$ containing that edge are affected, i.e., 
the associated preprocessed data structures of these pieces and their descendants might no longer contain correct data.
At the same time, the preprocessed data structures 
for other pieces of $\rdiv$ remain valid.
Consequently, after a single edge capacity update, the possibly invalid data for an affected piece $H\in\rdiv$ containing
that edge and all its descendants
can be recomputed from scratch in $\Ot\left(\sum_{H\supseteq B\in \TG(G^*)}|B|+|\bnd{B}|^2\right)=\Ot(r)$ time.
By choosing $r=n^{6/7}$, we obtain the following.

\begin{theorem}
  Let $G$ be a planar digraph with integer edge capacities in $[1,C]$ and let $\lambda\in [1,nC]$
  be an integer. 
  There exists a dynamic data structure with $\Ot(n)$ preprocessing time that can test whether
  there exists an $s,t$-flow of value $\lambda$
  for any given query source/sink pair $(s,t)$ in $\Ot(n^{6/7}\log{C})$ time,
  and supports edge capacity updates in $\Ot(n^{6/7})$ time.
\end{theorem}
By proceeding analogously as in Section~\ref{s:flow-approx}, the above data structure can be converted
to a dynamic $(1+\eps)$-approximate max $s,t$-flow oracle with $\Ot(n^{6/7}\log{(C)}/\eps)$ worst-case
update time and $\Ot(n^{6/7}\log{C}\cdot  \log(\log{(C)}/\eps))$ query time, which proves Theorem~\ref{t:dynamic-flow-approx}.

\subsection{An exact max $s,t$-flow oracle}
In this section we show that using subquadratic space and preprocessing time one can
still support exact max $s,t$-flow queries in sublinear time.

Consider the feasible flow oracle of Section~\ref{s:feasible-flow-oracle} along with the $r$-division-aware modification
of the procedure $\detect$ of Section~\ref{s:dynamic-oracle}. Note that the procedure  $\detect$  runs
in $\Ot(n^{3/2}/r^{3/4}\cdot \log{C})$ time as long as we guarantee
that the relevant closest-pair and near neighbor data structures
for the graph $D(H,x^*,y^*)$, that preserves distances between $\bnd{H}$ in $G^*_\lambda[H](x^*,y^*)$,
can be constructed in $\Ot(|\bnd{H}|)$ time
whenever we end up in a leaf call. In the following, we concentrate on case 2 of $\detect$, since case 1 can be handled
in $O(1)$ time even with no preprocessing.
The main goal in this section is to prove the following key lemma.

\begin{lemma}\label{l:exact-piece}
  Let $H\in T(G^*)$ be a piece and let $x^*,y^*\in Z(H)$.
  After preprocessing $H$ in $\Ot(|H|^2)$ time, for any given $\lambda$, one can construct closest pair
  and near neighbor data structures for  $D(H,x^*,y^*):=\DDG(G^*_\lambda[H](x^*,y^*))$, as required by
  the procedure $\detect$, satisfying $T_\cp(D(H,x^*,y^*))=\Ot(|\bnd{H}|)$, $T_\nn(D(H,x^*,y^*))=\Ot(|\bnd{H}|)$,
  and $Q_\nn(D(H,x^*,y^*))=\Ot(1)$.
\end{lemma}
  
Before proving Lemma~\ref{l:exact-piece}, let us explain how it leads to an exact max $s,t$-flow oracle
with subquadratic preprocessing/space and sublinear query time.
We apply the preprocessing of Lemma~\ref{l:exact-piece}
to every piece $H$ that is a weak descendant of some piece in $T(G^*)\cap \rdiv$. Recall that the sum of sizes
of descendant pieces of $H$ is $\Ot(|H|)$ since $T(G^*)$ has depth $O(\log{n})$
and no two siblings share a natural face. 
As a result, the preprocessing time is $\sum_{H\in\rdiv}\Ot(|H|^2)=\Ot(r)\cdot \sum_{H\in \rdiv}\Ot(|H|)=\Ot(nr)$.

Observe that using the preprocessing as given in Lemma~\ref{l:exact-piece}, $\detect(\lambda,G^*,s^*,t^*)$ runs in \linebreak
$\Ot(n^{3/2}/r^{3/4}\log{C})$ time
even if $\lambda$ is a query parameter.
In order to compute the exact max $s,t$-flow value $\lambda^*$, we simply find it
via binary search (over the interval $[0,nC]$) whose decisions are driven by the result of 
$\detect(\lambda,G^*,s^*,t^*)$.
This way, we obtain:

\texactflow*
The remaining part of this section is devoted to proving Lemma~\ref{l:exact-piece}.

  Recall that so far, the only preprocessing required by the closest pair and near neighbor
  data structures was that of Lemma~\ref{l:ddg-ds}. Actually, the proof
  of Lemma~\ref{l:ddg-ds} goes through for $\DDG(G^*_\lambda[H](x^*,y^*))$ if we only assume the min-decomposition
  of $\DDG(G^*_\lambda[H](x^*,y^*))$ into Monge matrices with a total of $\Ot(|\bnd{H}|)$ rows and columns (as in Lemma~\ref{l:decomp})  and $\Ot(1)$-time
  access to the edge weights of $\DDG(G^*_\lambda[H](x^*,y^*))$. In the current setting, we need to access dense distance graphs
  and its decomposition into Monge matrices for an arbitrary (integer) parameter
  $\lambda$ between $0$ and $\lambda^*$, where $\lambda^*_H$ is the maximum value $\lambda$
  such that $G^*_\lambda[H](x^*,y^*)$ contains no negative cycle. Such 
  $\lambda^*_H$ can be computed via binary search or even in near-linear strongly polynomial time
  using the planar maximum flow algorithm of~\cite{BorradaileK09, Erickson10} run
  on the dual of $G^*_\lambda[H](x^*,y^*)$.
  Unfortunately, 
  we cannot afford to apply preprocessing of Lemma~\ref{l:ddg-ds} for all possible values $\lambda$
  so that all the required data is stored explicitly.
  Instead, we will show that after $\Ot(|H|^2)$ preprocessing we can access
  the required data in $\Ot(1)$ time.

  There are two challenges here. First, we need to be able to efficiently access the edges of
  the dense distance graph $\DDG(G^*_\lambda[H](x^*,y^*))$ for an arbitrary parameter $\lambda$.
  Second, the partition of $\DDG(G^*_\lambda[H](x^*,y^*))$ into
  Monge matrices might vary for different values of $\lambda$.
  If we used Lemma~\ref{l:decomp} in a black-box way for the latter task,
  we would need to spend a prohibitive time of $\widetilde{\Omega}(|H|)$ on computing the partition.
  We will show that in our case (assuming the $\Ot(|H|^2)$ preprocessing),
  such a partition can be computed faster, in $\Ot(|\bnd{H}|)$ time, which is sufficient for our needs.
  
Recall that the parameterized edges in $G^*_\lambda[H](x^*,y^*)$
form a path $Q$ in the \emph{dual} of $G^*_\lambda[H](x^*,y^*)$.
From the point of view of $G^*_\lambda[H](x^*,y^*)$, $Q$ is
a curve originating in some hole $h_Q^1$ and ending in some
hole $h_Q^2$ of $G^*_\lambda[H](x^*,y^*)$ (possibly $h_Q^1=h_Q^2$), and not crossing any other holes (in particular $h_Q^1,h_Q^2$)
of that graph.

For each of $O(1)$ pairs $h^1,h^2$ of holes of $H$ (possibly $h^1=h^2$), we will construct separate closest pair and near neighbor data structures
for a subgraph of $D(H,x^*,y^*)=\DDG(G^*_\lambda[H](x^*,y^*))$ containing the edges between $V(h^1)$ and $V(h^2)$
(where we denote by $V(h^i)$ the boundary vertices of $H$ lying on $h^i$)
with total update time $\Ot\left(|V(h^1)|+|V(h^2)|\right)$ and query time $\Ot(1)$.
By Lemmas~\ref{l:nn-sum}~and~\ref{l:ssp-sum}, this will yield the respective data structures
for $D(H,x^*,y^*)$.

  Roughly speaking, for a pair $h^1,h^2$, we will expose $O(1)$ partial\footnote{In a partial Monge matrix, the Monge property is only satisfied if all the four cells in the Monge property are defined, and the defined cells in every row (column) lie in some number of consecutive columns (rows, resp.).} Monge matrices,
  whose element-wise minimum
  encodes the distances from $V(h^1)$ to $V(h^2)$ 
  in $G^*_\lambda[H](x^*,y^*)$, just as the decomposition of Lemma~\ref{l:decomp} does.
  However, we will spend only $\Ot\left(|V(h^1)|+|V(h^2)|\right)$ time on this.

  \subsubsection{The single-hole case}\label{s:single-hole}
  Let us first consider the easier case $h^1=h^2=h$. 
  We will later reduce the more involved case $h^1\neq h^2$ to that case.

  Let us start with the following technical lemmas on the possible
    net number of crossings of simple paths wrt. simple dual paths or cycles.

  \newcommand{\curv}{\mathcal{C}}

  \begin{lemma}\label{l:crossing1}
    Let $G_0$ be a plane digraph.
    Let $P=u\to v$ be a simple path in $G_0$ such that both $u$ and $v$
    lie on a single face $h$ of $G$.
    Let $Q_0\subseteq G_0^*$ be either a simple path $f\to g$ such that
    $h\notin V(Q_0)\setminus \{f,g\}$, 
    or a simple cycle in the dual $G_0^*$.
    Then $\pi_{Q_0}(P)\in \{-1,0,1\}$.
  \end{lemma}
  \begin{proof}
    If $Q_0$ is a cycle, then by Fact~\ref{f:crossing} we have $\pi_{Q_0}(P)=\pi_{P}(Q_0)\in \{-1,0,1\}$.
    Otherwise, $Q_0$ is a simple path in $G_0^*$ with at most
    one edge $z$ incident to $h$. Note that one can embed
    an edge $vu$ inside $h$ so that it does not cross $z$.
    Since $P\cdot vu$ is a simple directed cycle in $G_0+vu$
    and $Q_0$ is a path in $(G_0+vu)^*$, we have
    $\pi_{Q_0}(P\cdot vu)\in \{-1,0,1\}$ by Fact~\ref{f:crossing}.
    But, by construction, $vu,(vu)^R\notin E(Q_0)$, so
    $\pi_{Q_0}(P)=\pi_{Q_0}(P\cdot vu)$.
  \end{proof}

\begin{lemma}\label{l:crossing2}

    For $u,v\in V(G_0)$, let $P=u\to v$ be a path in $G_0$.
    Then:
    \begin{enumerate}[(a)]
      \item If $u$ and $v$ lie on a single face $h$ of $G$, and
        $P$ is simple, then       \item If $f=g$, then $\pi(P)$ is uniquely determined by the respective sides of endpoints $u,v$ wrt. $\curv$.
    \end{enumerate}
  \end{lemma}
 
\begin{proof}
 Now let us consider the latter claim. $\curv$ partitions $G$ into two connected parts $G_1,G_2$ such that $G_1\cap G_2=\curv$,
    and $G_1$ is to the left of $\curv$.
    Note that if both $u,v$ are in $G_i$ for $i\in \{1,2\}$, then the crossing
    number of $P$ is $0$. If $u$ is in $G_1$, and $v$ is in $G\setminus G_1$, then
    the crossing number $\pi(P)$ is necessarily $1$.
    Similarly, if $u$ is in $G_2$, and $v$ is in $G\setminus G_2$, $\pi(P)\equiv -1$.
      \end{proof}

  The next lemma provides a construction of a certain extension $G'$ of a plane digraph $G$, such that
  the shortest paths in $G'$ encode the shortest paths $R\subseteq G$ with a given crossing number wrt. a simple dual path.

  \begin{lemma}\label{l:k-paths}
    Let $G$ be a plane digraph with no negative cycles, and let $n=|V(G)|$. Let $P^*=a\to b$ be some simple path in $G^*$ (in particular, $a\neq b$).
    Then, in $O(n^2)$ time one can construct another plane digraph $G'$ with vertex set
    $V(G')=\{v_j:v\in V(G), j\in \{-n,\ldots,n\}\}$ such
    that:
    \begin{enumerate}
      \item For any $u,v\in V$, $j\in \{-n,\ldots,n\}$, $\dist_{G'}(u_0,v_j)$ equals the length
        of some $u\to v$ path $R\subseteq G$ with $\pi_{P^*}(R)=j$,
      \item For any $u,v\in V$, $j\in \{-n,\ldots,n\}$, $\dist_{G'}(u_0,v_j)$ is at most the
        length of the shortest \emph{simple} $u\to v$ path $R\subseteq G$ with $\pi_{P^*}(R)=j$. 
      \item For any $S\subseteq V(G)$ such that the vertices $S$ lie on a single face $f\notin V(P^*)\setminus\{a,b\}$ of $G$,
        and for any $j\in \{-n,\ldots,n\}$ the vertices $S_j:=\{s_j:s\in S\}$ lie on a single face of $G'$.
        Moreover, the faces containing $S_j$ for distinct $j$ are either equal or vertex-disjoint.
    \end{enumerate}
  \end{lemma}
  \begin{proof}
    The graph $G'$ is defined similarly to the infinite \emph{universal cover graph} from~\cite[Section~2.4]{Erickson10}, 
    albeit $G'$ is finite.
    It is obtained as follows: first obtain a graph $\tilde{G}$ by introducing
    an auxiliary vertex wherever the curve $P^*$ crosses the embedding of $G$,
    and connecting the neighboring auxiliary vertices on $P^*$ by auxiliary edges
    embedded in a way that there is a path in $\tilde{G}$ embedded exactly as $P^*$.
    Next, cut $\tilde{G}$ along the path $P^*$, so that the obtained
    graph $\tilde{G}'$ has two copies $P^*_L$ and $P^*_R$ of $P^*$.
    Make $2n+1$ copies $\tilde{G}'_{-n},\ldots,\tilde{G}'_{n}$
    of the graph $\tilde{G}'$.
    Label the copy of a vertex $v$ in $\tilde{G}'_i$ with $v_i$.
    Connect the copies into a chain by gluing the path
    $P^*_L$ of the copy $\tilde{G}'_{i}$ with the path $P^*_R$
    of the copy $\tilde{G}'_{i+1}$ for all $i=-n,\ldots,n-1$.
    Finally, remove from the obtained graph all the (copies of)
    auxiliary vertices and edges except those appearing on $P^*_L$ in $\tilde{G}'_{-n}$
    and on $P^*_R$ in $\tilde{G}'_n$.
    See Figure~\ref{fig:cover}.

\begin{figure}[ht!]
    \centering
    \includegraphics[scale=0.7]{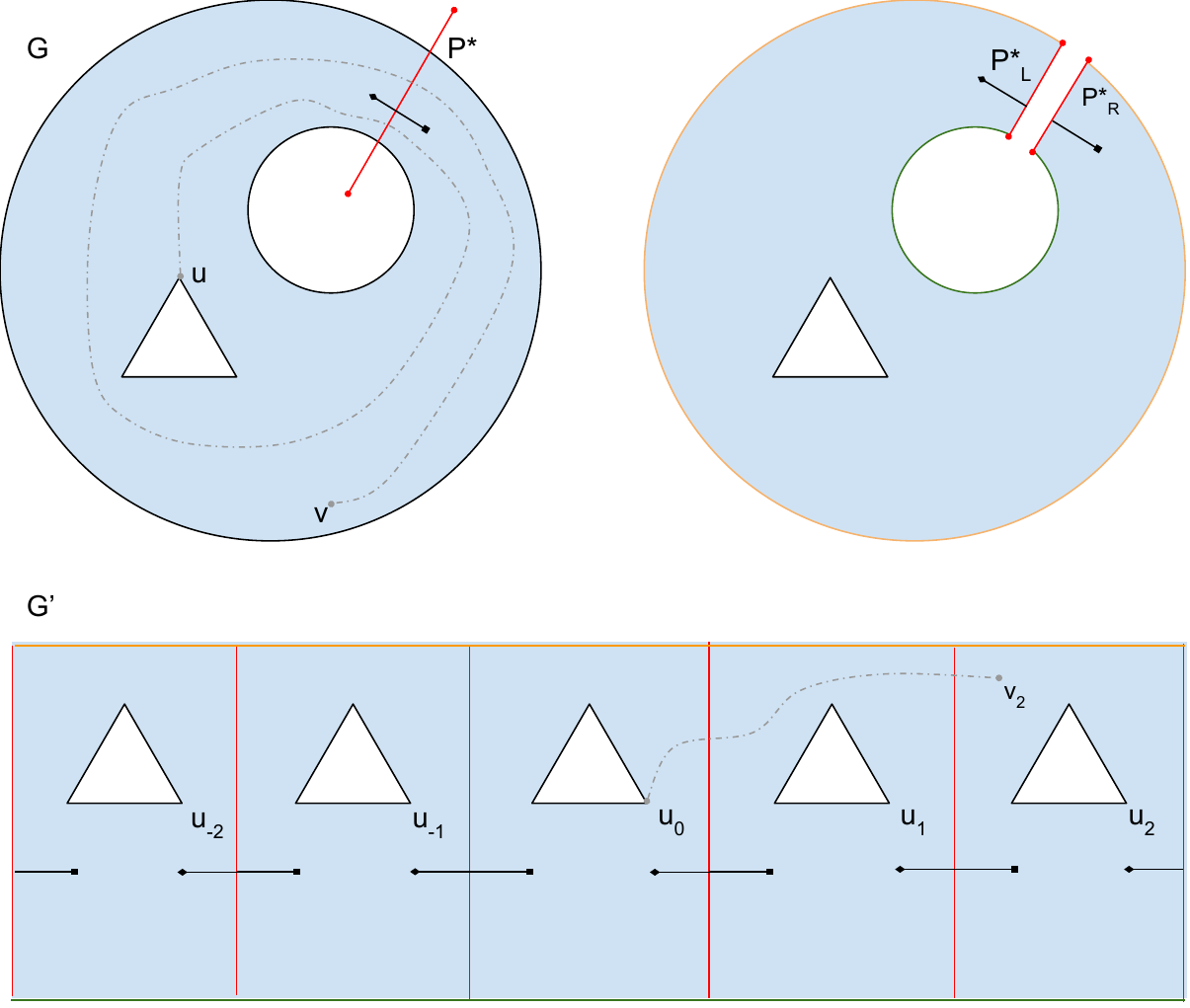}
    \caption{Construction on the graph $G'$ from Lemma~\ref{l:k-paths}. The gray $u\to v$ path in $G$
    is preserved in the graph $G'$ in the bottom}\label{fig:cover}

  \end{figure}

    By construction, every path $R'=u_0\to v_j$ in $G'$ corresponds to some path $R=u\to v$ in $G$ with $\pi_{P^*}(R)=j$. This implies property~$1$.

    For a \emph{simple} path $R=u\to v$ in $G$, $R$ is preserved in $G'$ for the following reason.
    Suppose $R=e_1\ldots e_k$, where $e_i=x_iy_i$, and denote by $c_i$, $i=0,\ldots,k$, the crossing number $\pi_{P^*}(\{e_1,\ldots,e_i\})$.
    That is, the $c_i$ is the number of times $e_1\ldots e_i$ crosses
    $P^*$ left to right, minus the number of times it crosses $P^*$ right to left. Note that we have $|c_i|\leq i$ and $k\leq n$ since $R$ is simple.
    As a result, we have $|c_i|\leq n$ for all $i$.
    It follows that a path $R'=u_0(y_1)_{c_1}(y_2)_{c_2}\ldots (y_k)_{c_k}=u_0\to v_{c_k}$ exists in $G'$
    and has the same length as $R$ in $G$. This proves property 2.

    Note that all the copies of vertices lying on faces $a$ or $b$ in $G$
    lie on a single face of $G'$.
    If some face $f\notin\{a,b\}$ of $G$ is not an intermediate
    vertex of $P^*$, then there are vertex-disjoint copies of $f$
    in all parts $\tilde{G}'_i$ of $G'$. This proves property~3.
  \end{proof}

    If $h^1_Q=h^2_Q$, then $Q$ is a simple directed cycle in the dual. As discussed earlier, the crossing number $\pi_{u,v}$ of \emph{any}
  path $u\to v$ path in $G^*_\lambda[H](x^*,y^*)$ wrt. $Q$ is in $\{-1,0,1\}$ and is uniquely determined by the positions (i.e., exterior or interior) of the vertices $u,v$ wrt. $Q$.
  As a result,
  we have $\dist_{G^*_\lambda[H](x^*,y^*)}(u,v)=\dist_{G^*_0[H](x^*,y^*)}(u,v)+\pi_{u,v}\cdot \lambda$.
  Consequently, the edge weights of $uv$ in $D(H^*,x^*,y^*)$ for $u,v\in V(h)$
  can be accessed in $O(1)$ time by shifting the respective weight
  of $\DDG(G^*_0[H](x^*,y^*))$ by $\lambda\cdot \pi_{u,v}$.
  To decompose $\DDG(G^*_\lambda[H](x^*,y^*))[V(h)]$ into Monge matrices, it is enough
  to use the decomposition of Lemma~\ref{l:decomp} $O(1)$ times, as follows.
  Express $V(h)=A\cup B$ so that the sets $A$ and $B$ lie on different sides of $Q$.
  Then, for each pair $(I,J)\in \{A,B\}^2$, we include in
  the desired decomposition 
  the decomposition of $\DDG(G^*_0[H](x^*,y^*))$ into Monge matrices
  but limited to rows $I$ and columns $J$. For each
  such combination $(I,J)$, the value $\pi_{u,v}$ is equal for all $(u,v)\in I\times J$.
  Therefore, for each Monge matrix included in our decomposition
  all its entries get shifted by the same value, which clearly preserves the Monge property.
  Observe that if $h^1_Q=h^2_Q$, this strategy works even for the more involved case of two distinct holes $h^1\neq h^2$, so
  we will neglect it later on in Section~\ref{s:distinct-holes}.

  Let us now assume $h^1_Q\neq h^2_Q$.
  Since $h^1_Q\neq h^2_Q$, we can build a graph $Z'$ from Lemma~\ref{l:k-paths} for the graph $G:=Z=G^*_0[H](x^*,y^*)$
  and the dual path~$P^*:=Q$.
  Using one of the well-known methods, e.g.~\cite{EricksonFL18}, we enforce uniqueness of shortest paths on the graph $Z'$
  by only manipulating its edge weights and so that the original path
  weight can be recovered in $\Ot(1)$ time from the transformed graph.
  For any face $f$ of $Z$, denote by $f_j$ the face of $Z'$ that
  contains the vertices $\{v_j:v\in V(f)\}$, where $v_j\in V(Z')$ is defined as in Lemma~\ref{l:k-paths}.

  Next, we construct an MSSP data structure on $Z'$ with a distinguished face $h_0$ (that is, the $0$-th copy of
  the face $h$ whose vertices' pairwise distance are of our interest) in a way that allows accessing
the shortest path trees from the vertices on $h_0$ via a dynamic tree interface. 
Note that computing $Z'$ and the associated data structures takes 
$\Ot(|Z'|)=\Ot(|H|^2)$ time.

In order to construct the part of $D(H,x^*,y^*)$ that
is responsible for encoding the distances between the vertices of $V(h)$, we use
an auxiliary graph $D_h=\left(V(h),\bigcup_{j\in\{-1,0,1\}}E_{h,j}\right)$,
satisfying the following.
Each $E_{h,j}$ contains for every $u,v\in V(h)$ an edge $uv$ with
weight $\dist_{Z'}(u_0,v_j)+j\cdot \lambda$.

\begin{lemma}\label{l:dhdist}
  For any $u,v\in V(h)$, $\dist_{G^*_\lambda[H](x^*,y^*)}(u,v)=\dist_{D_h}(u,v)$.
\end{lemma}
\begin{proof}
Since $G^*_\lambda[H](x^*,y^*)$ has no negative cycles, for any $u,v\in V(h)$,
there exists a shortest $u\to v$ path $P$ in $G^*_\lambda[H](x^*,y^*)$ that
is simple.
  As a result, the length of $P$ equals its length in $G^*_0[H](x^*,y^*)$
  plus $\lambda\cdot \pi_Q(P)$.
  Recall that the intermediate faces of $Q$ are not holes of $H$.
  As a result, by Lemma~\ref{l:crossing1}, $P$ has crossing number $\pi_Q(P)$ in $\{-1,0,1\}$.
  Therefore, the distance from $u$ to $v$ in $D_h$ is no more than
  $\dist_{Z'}(u_0,v_{\pi_Q(P)})+\pi_Q(P)\cdot \lambda$.
  By Lemma~\ref{l:k-paths}, $\dist_{Z'}(u_0,v_{\pi_Q(P)})$ is no more
  than the length of any simple $u\to v$ path in $G^*_0[H](x^*,y^*)$
  with crossing number $\pi_Q(P)$, in particular, it is no more
  than the length of $P$ in $G^*_0[H](x^*,y^*)$.
  Hence, we indeed obtain  $\dist_{G^*_\lambda[H](x^*,y^*)}(u,v)\geq \dist_{D_h}(u,v)$.
  
  If we had $\dist_{G^*_\lambda[H](x^*,y^*)}(u,v)>\dist_{D_h}(u,v)$
  for some $u,v\in V(h)$ then there would exist an edge $wz$ in some $E_{h,j}$
  with $\dist_{G^*_\lambda[H](x^*,y^*)}(w,z)>\dist_{Z'}(w_0,z_j)+j\cdot \lambda$.
  But then there would exist a (not necessarily simple) $w\to z$ path $P\subseteq G^*_0[H](x^*,y^*)$ with $\pi_Q(P)=j$
  such that its cost in $G^*_\lambda[H](x^*,y^*)$ is less than 
  $\dist_{G^*_\lambda[H](x^*,y^*)}(w,z)$. This is a contradiction in
  a graph with no negative cycles.
\end{proof}

Of course, we cannot build the graph $D_h$ explicitly since its size
may be $\Theta(|\bnd{H}|^2)$.
Note that for any $u,v\in V(h)$, the (minimum) weight of an edge $uv$ in $D_h$
can be queried in $\Ot(1)$ time using the MSSP data structure built on $Z'$.
However, in order to allow building fast closest pair and near neighbor
data structures on $D_h$, we still need to efficiently express
the edge weights matrix of $D_h$ using Monge matrices.
Recall that these edge weights encode (shifted) distances between the face $h_0$ and
faces $h_{-1},h_0,h_1$ in the graph $Z'$.
The distances between vertices of $h_0$
are easily decomposed into Monge matrices with a total of $\Ot(|V(h_0)|)$
rows and columns (counted with multiplicity) based on the ordering of vertices
on $h_0$ exclusively, using the following fact.
\begin{fact}\label{f:monge-easy}
  Let  $A$ and $B$ be two non-interleaving (wrt. to a clockwise order) subsets of vertices lying on a single
  face of a plane digraph.
  Let $\mmat$ be a matrix encoding all-pairs distances from vertices $a\in A$ to vertices $b\in B$.
  Then, $\mmat$ is a Monge matrix.
\end{fact}
\cite{FR06} apply the above fact a logarithmic number of times to decompose all-pairs
distances between $k$ vertices lying on a single face of a plane digraph into Monge matrices
with a total of $\Ot(k)$ rows/columns. These rows/columns depend only on the ordering
of the vertices on the face.

The distances between $h_0$ and $h_j$, if $h_0\neq h_j$, are more tricky to decompose into
an element-wise minimum of Monge matrices.
  \cite{MozesW10} deal with this problem by observing that from an arbitrary
  shortest paths tree~$T$ rooted at a vertex $s$ of $h_0$, one can choose two paths
  such that for every two vertices $v_1,v_2$ on the holes $h_0,h_j$ respectively,
  some shortest $v_1\to v_2$ path in $Z'$ either does not cross the former path or does not cross the latter.
  Cutting $Z'$ along each of these paths (separately) yields two graphs
  where  $V(h_0)$ and $V(h_j)$ together lie on
  \emph{a single face}  (and do not interleave) of the respective graph. As a result, one can view
  distances between them as a single (full) Monge matrix. The element-wise minimum of distances
  in the two graphs represents distances between $h_0$ and $h_j$ in the original graph.
  Unfortunately, finding the paths used for cutting, actually cutting the graph explicitly, and preprocessing the obtained cut graphs
  would be too costly in our case.

  The following lemma shows a more efficient way to compute the partition
  of distances between $h_0$ and $h_j$ into \emph{partial} Monge matrices,
  and thus also an efficient construction of closest pair and near neighbor
  data structures for a graph induced on the edges $E_{h,j}$.
  Recall that in a partial Monge matrix the elements in each row and each columns are defined
  only in a consecutive segment of columns or rows, respectively.
  As a result, for a partial Monge matrix with rows $R$ and columns $C$, it takes
  only $O(|R|+|C|)$ space to encode which elements are defined: one only
  needs to encode the ordering of rows and columns, and a single segment in each row.
  Below, when talking about decomposing into partial Monge matrices, we
  will mean computing an encoding as defined above.

  \begin{lemma}\label{l:partial-monge}
    Suppose $h_0\neq h_j$. There exist closest pair and near neighbor data structures
    for $D_h[E_{h,j}]$ with $T_\cp(D_h[E_{h,j}]),T_\nn(D_h[E_{h,j}])\in \Ot(|V(h_0)|+|V(h_j)|)$, 
    and $Q_\nn(D_h[E_{h,j}])=\Ot(1)$.
  \end{lemma}
  \begin{proof}
    Recall from Lemma~\ref{l:k-paths} that the faces $h_0$ and $h_j$ in $Z'$ are vertex-disjoint. First, pick
    such $s\in V(h_0)$ and $t\in V(h_j)$ that the unique $s\to t$ shortest
    path $Y$ in $Z'$ does not contain any other vertices of $V(h_0)\cup V(h_j)$.
    This can be done in $\Ot(|V(h_0)|+|V(h_j)|)$ time as follows. First pick any $s^*\in V(h_0)$. Then,
    pick such a $t\in V(h_j)$ that has the lowest depth in the shortest paths tree from $s^*$
    stored in the MSSP data structure -- this can be done in $\Ot(|V(h_j)|)$ time.
    Finally, on the shortest $s^*\to t$ path, pick the deepest vertex $s\in V(h_0)$ in a similar way.
    By construction, the $s\to t$ subpath of the $s^*\to t$ path contains
    no other vertices from $V(h_0)\cup V(h_j)$.

    Fix some vertex $u\in V(h_0)$. The shortest path $P_t=u\to t$ shares some suffix $S$
    with $Y$ (possibly only the vertex $t$) due to uniqueness and cannot cross $Y$.
    Suppose wlog. that $P_t$ merges with $Y$ from the left side of $Y$ (the other case is symmetric), and
    denote by $A$ the set of vertices $u$ with this property (so $u\in A$).
    
    Let us observe that
    for any $v\in V(h_j)$, the shortest $u\to v$ path $P_v$ \emph{does not cross} $Y$
    right to left: indeed, if the crossing vertex $w$ was at the subpath~$S$,
    this would contradict uniqueness of the $u\to w$ subpath of $P_t$.
    On the other hand, if the path $P_v$ crosses $Y$ at vertex $w$ before the subpath~$S$, 
    then it would need to cross either $P_t\setminus S$ or $Y\setminus S$ once more on the
    way to $w$, and thus this would violate the uniqueness of $P_t$ or $Y$ respectively.

    Let us embed an auxiliary vertex $u'$ and an edge $u'u$ inside the face $h_0$ of $Z'$
    and call that edge the $0$-th \emph{sentinel} edge of any shortest
    path $P_x=u\to x$.
    If $x\in V(h_j)\setminus\{t\}$, then let $e_x$ ($e_x'$) be the earliest edge
    of the path $P_x$ ($P_t$, resp.) that does not belong to $P_t$ ($P_x$, resp.).
    Note that if $e_x$ emerges from the path $P_t$ to the left (that is,
    $e_x=ab$ lies between $e_x'$ and and the preceding edge of $e_x'$ on $P_t$ in the clockwise edge ring of $a$),
    then $P_x$ does not cross $P_t$ and thus also does not cross $Y$ (it might still share
    some prefix with $P_t$, and some subpath with $Y$, though).
    On the other hand, if $e_x$ emerges from the path $P_t$ to the right,
    then the shortest $u\to x$ path has to cross $Y$ once.
    Let $\mathcal{L}_u\subseteq V(h_j)$ contain $t$ and the vertices $x\in V(h_j)\setminus\{t\}$ such
    that $e_x$ emerges from $P_t$ to the left.
    Similarly, let $\mathcal{R}_u\subseteq V(h_j)$ contain the vertices $x$ of $h_j$ such
    that $e_x$ emerges from $P_t$ strictly to the right.
    Since the paths $P_x$ for $x\in V(h_j)$ are pairwise non-crossing,
    $\mathcal{L}_u$ and $\mathcal{R}_u$ constitute consecutive fragments of $V(h_j)$ as ordered
    on the face $h_j$.
    Moreover, the border between $\mathcal{L}_u$ and $\mathcal{R}_u$ can be computed via binary search
    on the vertices $x\in V(h_j)$: whether $x\in \mathcal{L}_u$ or $x\in \mathcal{R}_u$
    can be decided using an LCA query and a level ancestor query on the shortest path tree from $u$.
    If we make MSSP use the top tree~\cite{AlstrupHLT05} for dynamic tree operations,
    these queries are supported in polylogarithmic time.

    As $Y$ is a unique $s\to t$ shortest path, it is a simple path.
    Consider a graph $Z'_{s,t}$ obtained from $Z'$ by cutting $Z$ along the path $Y$ (that is,
    the vertices and edges of $Y$ appear in $Z'_{s,t}$ in two copies). 
    Now, let $Z''_{s,t}$ be obtained by gluing two copies of $Z'_{s,t}$ along one of the
    copies of $Y$ (that is, $Z''_{s,t}$ contains three copies of $Y$).
    Denote by $Z''_{s,t,L}$ ($Z''_{s,t,R}$) the copy of $Z'_{s,t}$ to the left (to the right, resp.) of the path $Y$
    used for gluing the two copies into $Z''_{s,t}$.
    Note that the graph $Z''_{s,t}$ contains two copies of each vertex
    in $V(h_0)\setminus \{s\}$ and $V(h_j)\setminus \{t\}$.
    For each such vertex $x$, denote by $x_L,x_R$ the copies
    in the respective parts $Z''_{s,t,L}$, $Z''_{s,t,R}$.
    Moreover, $Z''_{s,t}$ contains three copies of $s$ and $t$.
    For $x\in \{s,t\}$, set both $x_L,x_R$ to be the ``middle''
    copy on the copy of path~$Y$ used for gluing.
    
    Impose a natural clockwise order starting at $s$ of $V(h_0)$ on the face $h_0$ of $Z'$. Similarly, impose
    a counter-clockwise order on $V(h_j)$ starting at $t$ on the face $h_j$ of $Z'$.
    For $W\subseteq V(Z')$, let $W_L=\{w_L:w\in W\}$ and similarly define $W_R$.
    Note that the vertices $V(h_0)_L$, $V(h_0)_R$, $V(h_j)_L$, $V(h_j)_R$ all lie
    on a \emph{single face} $f$ of $Z''_{s,t}$ in the clockwise
    order $V(h_0)_L, V(h_0)_R, V(h_j)_R, V(h_j)_L$, where
    each of this sets is ordered using the order imposed on the respective
    original set $V(h_0)$, $V(h_j)$.

    Observe that every path in $Z''_{s,t}$ naturally corresponds to some path
    in $Z'$ (after projecting the copies of vertices to respective original vertices of $Z'$).
    As a result, if some vertices $u',v'$ of $Z''_{s,t}$ are copies of $u,v$ respectively in $Z'$,
    then $\dist_{Z''_{s,t}}(u',v')\geq \dist_{Z'}(u,v)$.
    Moreover, for a shortest path $P_{u,v}=u\to v$ in $Z'$, if $u\in A$ and $v\in \mathcal{L}_u$, there
    is a path $u_L\to v_L$ in $Z''_{s,t,L}$ corresponding to $P_{u,v}$.
    Consequently, $\dist_{Z'}(u,v)=\dist_{Z''_{s,t}}(u_L,v_L)$.
    Similarly, if $u\in A$ and $v\in \mathcal{R}_u$, there is a path
    $u_L\to v_R$ in $Z''_{s,t}$ corresponding to $P_{u,v}$,
    and we obtain $\dist_{Z'}(u,v)=\dist_{Z''_{s,t}}(u_L,v_R)$.

    Consider a matrix $\mmat$ with rows $A$ and columns $V(h_j)_R\cup V(h_j)_L$
    such that in a row $u\in A$, the elements are defined only for columns $(\mathcal{R}_u)_R\cup (\mathcal{L}_u)_L$.
    If $k'\in (\mathcal{R}_u)_R\cup (\mathcal{L}_u)_L$ is a copy of $k\in V(h_j)$,
    then we put $\mmat_{a,k'}=\dist_{Z'}(a,k)$.
    Note that if the columns of $\mmat$ are ordered as in the clockwise order on the
    face $f$ of $Z''_{s,t}$, then $(\mathcal{R}_u)_R\cup (\mathcal{L}_u)_L$ constitutes
    a contiguous segment of columns of $\mmat$.
    
    We now argue that for each column $k'$ (such that $k'$ is a copy of vertex $k$) of $\mmat$, the defined elements of that
    column lie in a consecutive segment of rows of $\mmat$.
    Suppose wlog. that $k'\in V(h_j)_L$.
    Then, for $a,b\in A$, the elements $\mmat_{a,k'}$ and $\mmat_{b,k'}$ are both defined
    iff the shortest paths $P_{a,k}$ and $P_{b,k}$ both do not cross $Y$ in $Z'$.
    Clearly, for any $c\in A$ that lies between $a$ and $b$ in the clockwise order of $V(h_0)$,
    the path $P_{c,k}$ is weakly contained in a cycle consisting of $P_{a,k},P_{b,k}$,
    and the $a\to b$ part of the enclosing cycle of $V(h_0)$ that does not contain $s$.
    As a result, since $P_{a,k}$ and $P_{b,k}$ do not cross $Y$, indeed $P_{c,k}$ does not cross $Y$ and we conclude
    that $\mmat_{c,k'}$ is defined as well.
    
    Since the rows and columns of $\mmat$ are non-interleaving subsets of vertices on a single
    face $f$ of $Z_{s,t}''$, from Fact~\ref{f:monge-easy} we conclude that $\mmat$ is indeed a partial Monge matrix.
    Note that (values of) the defined elements of $\mmat$ do not depend on the choice
    of $s,t$ -- they constitute distances in $Z'$ and thus can be accessed in $\Ot(1)$ time
    using the MSSP data structure that we have built upon $Z'$.
    Any partial $m_1\times m_2$ Monge matrix can be decomposed into a collection $D$ of
    full Monge matrices with $\Ot(m_1+m_2)$ rows and columns (counted with multiplicities).
    For an intuition and proof, see~\cite[Lemma~5.1, Theorem~5.5, and Figures~3,4,5]{GawrychowskiMW20}.

    The defined entries of $\mmat$ do not quite correspond to the weights of edges $E_{h,j}\cap (A\times V(h_j))$:
    the columns of $\mmat$ include two copies of each vertex of $V(h_j)$.
    However, we can easily obtain a graph $G_{h,j}$ whose $A\times V(h_j)$ distances are equal to
    the weights of $E_{h,j}\cap (A\times V(h_j))$. 
    Let the vertices of $G_{h,j}$ be $A\cup V(h_j)_L\cup V(h_j)_R\cup V(h_j)$.
    $G_{h,j}$ contains a subgraph $G_\mmat$ with edges between $A$ and $V(h_j)_L\cup V(h_j)_R$ corresponding
    to the defined elements of $\mmat$, and $2|V(h_j)|$ auxiliary edges
    $u_Lu,u_Ru$ of weight $0$ for all $u\in V(h_j)$.
    Similarly as in the proof of Lemma~\ref{l:ddg-ds}, we can use the data structure of \cite[Lemma 1]{MozesNW18} and the Monge heap to obtain (via Lemmas~\ref{l:ssp-sum}~and~\ref{l:nn-sum} applied to the decomposition $D$ of $\mmat$ into full Monge matrices) near neighbor and closest pair data structures for $G_\mmat$ with total update time
    $\Ot(|V(h_0)|+|V(h_j)|)$ and query time $\Ot(1)$.
    To obtain the desired closest pair and near neighbor data structures for $G_{h,j}$ (which, as argued above,
    can serve as the respective data structures for $D_h[E_{h,j}]$) we simply apply
    Lemmas~\ref{l:nn-sum}~and~\ref{l:ssp-sum} once again to the data structures for $G_{\mmat}$ and the
    trivial data structures for $G_{h,j}\setminus G_{\mmat}$ which contain
    $\Ot(|V(h_0)|+|V(h_j)|)$ auxiliary non-structured edges.
  
    We proceed symmetrically with the distances from $V\setminus A$ to $V(h_j)$.
    Finally, we merge the obtained data structure for graphs representing distances
    $A\times V(h_j)$ and $(V\setminus A)\times V(h_j)$, resp. into a single one
    with asymptotically same time bounds using Lemmas~\ref{l:nn-sum}~and~\ref{l:ssp-sum}.
\end{proof}

  With the data structures of Lemma~\ref{l:partial-monge} in hand, we can again apply Lemmas~\ref{l:ssp-sum}~and~\ref{l:nn-sum}
  to get the desired data structures for $D_h=D_h[E_{h,-1}]\cup D_h[E_{h,0}]\cup D_h[E_{h,1}]$
  with total update time $\Ot\left(|V(h)|\right)=\Ot(|\bnd{H}|)$ and query time $\Ot(1)$.
  This way, we achieve our goal for the case $h^1=h^2$.

  \subsubsection{The case of two distinct holes}\label{s:distinct-holes}
  The primary difference between the case $h^1\neq h^2$ and the single-hole case
  is that it is no longer the case that the crossing number of a path
  from $h^1$ to $h^2$ has wrt. $Q$ is necessarily in $\{-1,0,1\}$.
  We will nevertheless prove the following. Recall that we have
  already handled the case $h^1_Q=h^2_Q$, so below we assume $h^1_Q\neq h^2_Q$.

  \begin{lemma}\label{l:distinct-holes}
    Let $s\in V(h^1)$, $t\in V(h^2)$, where $h^1\neq h^2$. Suppose
    there exists a \emph{simple} shortest ${s\to t}$ path $P$ in 
    $G^*_\lambda[H](x^*,y^*)$ such that $V(P)\cap (V(h^1)\cup V(h^2))=\{s,t\}$ with crossing number $\pi_Q(P)$.
    Then, for any $(u,v)\in V(h^1)\times V(h^2)$, there exists
    some shortest $u\to v$ path in $G^*_\lambda[H](x^*,y^*)$
    with crossing number wrt. $Q$ in $[\pi_Q(P)-4,\pi_Q(P)+4]$.
  \end{lemma}
  \begin{proof}
    Since $P$ is a shortest path that is simple, there exists some simple shortest path $S=u\to v$ that does not cross~$P$ more than once (either from left to right or right to left).
    Let us focus on the case when $P$ and $S$ have at least one vertex in common, and $S$ indeed crosses
    $P$ once, wlog. left to right. The other
    cases, when $P$ and $S$ are completely disjoint or when they have a common part, but do not cross,
    are easier to handle, and can by considered analogously.

    Since $P$ and $S$ have a vertex in common, one can write $P=P_1RP_2$, $S=S_1RS_2$, where $x\to y=R\subseteq P$ is a possibly empty (i.e., containing zero edges) maximal common subpath of $P$ and $S$,
    and $S_1,S_2$ can share only endpoints with $V(P)$.

    Let $G^+$ be the graph obtained from $G^*_\lambda[H](x^*,y^*)$
    by cutting it along $P$, so that $G^+$ has two copies of the vertices
    and edges of $P$.
    Denote by $P_L,P_R$ the left and right copies of $P$ respectively,
    and by $w_L,w_R$ the respective left and right copies of a vertex $w\in V(P)$.
    For all $w\in V(H)\setminus V(P)$, set $w_L=w_R=v$.
    In this graph, the vertices $(V(h^1)\setminus \{s\})\cup (V(h^2)\setminus\{t\})$ lie on a single face~$f$,
    and each of the two sets comprising this union forms a segment of consecutive vertices on $f$.
    The other vertices on $f$ come from the two copies of~$P$.
    These sets also form contiguous segments on $f$, and one can easily
    see that the ordering of vertices on $f$ can be wlog. assumed to be $V(h^1),V(P_L),V(h^2),V(P_R)$,
    where $V(P_L)$ starts with $s_L$, and $V(P_R)$ starts with $t_R$.
    
    Note that for $S_1=u\to x$, there exists a \emph{unique} corresponding path $S_1^+=u_L\to x_L$ in $G^+$.
    Similarly, for $S_2=y\to v$, there exists a unique corresponding path $S_2^+=y_R\to v_R$ in $G^+$.
    Let $P_1^+$ be the $s_L\to x_L$ subpath of $P_L$, and let $P_2^+$ be the $y_R\to t_R$ subpath of $P_R$ in $G^+$.

    Consider the effect that cutting along $P$ has on the dual path $Q$. Express $Q$ as\linebreak
    $Q_1\cdot e_1 \cdot Q_2 \cdot e_2 \cdot \ldots \cdot e_{k-1}\cdot Q_k$ such that $P$ and $Q$ cross
    at edges $e_1,\ldots,e_{k-1}$.
    Let $i\in \{2,\ldots,k-1\}$.
    In $G^+$, each subpath $e_{i-1}\cdot Q_i \cdot e_i$ of $Q$ corresponds to a simple curve $Q_i^+$ with
    both its endpoints inside the face $f$, and intersecting the bounding cycle 
    of $f$ at copies of edges $e_{i-1},e_i$.
    As a result, each such $Q_i^+$ crosses $f$ at some two points of $f$
    that lie on the segments spanned by either $V(P_L)$ or $V(P_R)$ on $f$.
    See Figure~\ref{fig:split} for better understanding.
    Regardless of which segments $Q_i^+$ crosses, every path from $V(h^1)\setminus \{s\}\cup\{s_L\}$ to some other fixed vertex $w$ on $f$ in $G^+$ has the same crossing number wrt. $Q_i^+$
    because $V(h^1)$ lies entirely on one side of the dual cycle $Q_i^+\subseteq G^+$.
    For an analogous reason, every path from some fixed vertex $w$ on $f$ to a vertex of $V(h^2)\setminus\{t\}\cup\{t_R\}$ has the same crossing number wrt. $Q_i^+$.
    So the crossing numbers of $P_j^+$ and $S_j^+$ (for $j=1,2$) wrt. $Q_i^+$ are equal.

  \begin{figure}[ht!]
    \centering
    \includegraphics[scale=0.7]{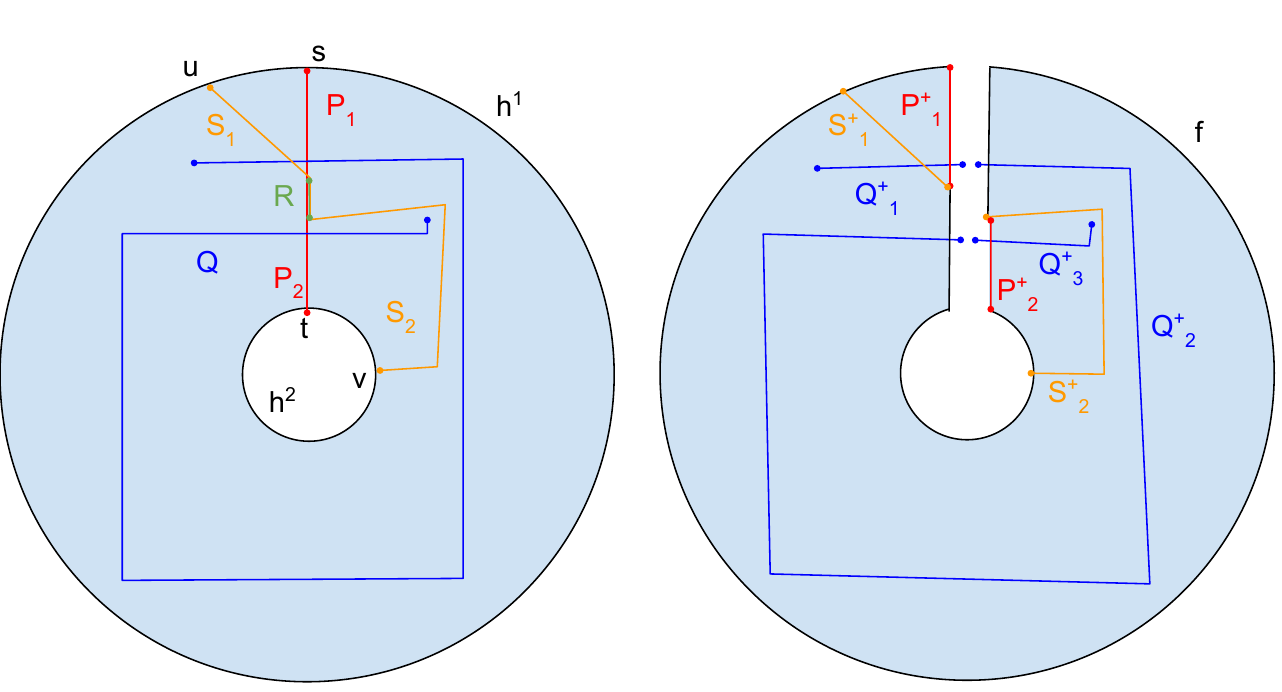}
    \caption{Illustration for the proof of Lemma~\ref{l:distinct-holes}.}\label{fig:split}
  \end{figure}

    On the other hand, observe that 
    the crossing number of $P_j$ wrt. $Q$, for $j=1,2$, equals the sum of crossing numbers of $P_j^+$
    wrt. $Q_2^+,\ldots,Q_{k-1}^+$. This is because the crossing number of $P_j$ wrt. each $Q_l$,
    where $l=1,\ldots,k$ is $0$ (by the definition of $Q_l$), and the crossing number
    of $P_j$ wrt. some $e_l$, where $l=1,\ldots,k-1$, contributes to precisely
    one of the crossing numbers of $P_j^+$ wrt.~$Q^+_i$.
    Moreover, the crossing number of $S_j$ wrt. $e_1\cdot Q_2\cdot e_2\cdot \ldots\cdot Q_{k-1}\cdot e_{k-1}$, equals the sum of crossing numbers of $S_j^+$
    wrt. $Q_2^+,\ldots,Q_{k-1}^+$. To see this,
    recall that $S_j$, disjoint with $P$, has crossing number~$0$ with all $e_l$.
    But we have proved that the crossing numbers of $P_j^+$ and $S_j^+$
    wrt. each $Q_i^+$ are equal,
    and thus we conclude that the crossing
    number of $P_j$ wrt. $Q$ equals the crossing number of $S_j$ wrt. $Q\setminus Q_1\setminus Q_k$.

    We now argue that the crossing number of $S_j$ wrt. $Q_1$ or $Q_{k}$ is in $\{-1,0,1\}$.
    It will follow that the crossing number of $P_j$ wrt. $Q$ can deviate by at most $2$ from
    the crossing number of $S_j$ wrt. $Q$.
    Let us focus on the crossing number wrt. $Q_1$, as the other case is analogous.
    Since $Q_1$ is preserved in $G^+$, and $S_j$ is preserved in $G^+$ as $S_j^+$,
    the crossing number of $S_j$ wrt. $Q_1$ in $G^*_\lambda[H](x^*,y^*)$ equals the crossing number of $S_j^+$ wrt. $Q_1$ in $G^+$.
    But the endpoints of the simple path $S_j^+$ lie on a single
    face $f$ of $G^+$, and no intermediate vertex of $Q_j$ (seen as a dual path in $G^+$) equals $f$, so the crossing number of $Q_j$ wrt. $S_j^+$ is in $\{-1,0,1\}$ by Lemma~\ref{l:crossing1}.

    To summarize, the crossing numbers of $P_j$ and $S_j$ wrt. $Q$ may differ by at most $2$.
    We conclude that the crossing numbers of $P=P_1RP_2$ and $S=S_1RS_2$ wrt. $Q$ can differ by at most $4$.
  \end{proof}
\begin{lemma}\label{l:distinct-prep}
  Let the holes $h^1,h^2$, where $h^1\neq h^2$, be fixed.
  In $\Ot(|H|^2)$ time one
  can build a data structure answering queries of the following form in $\Ot(1)$ time.
  Given $\lambda\in [0,\lambda^*]$,
  find some two vertices $s\in V(h^1)$, $t\in V(h^2)$ and a number $k$ such that there
  exists a \emph{simple} shortest $s\to t$ path $P$ in 
  $G^*_\lambda[H](x^*,y^*)$ with $V(P)\cap (V(h^1)\cup V(h^2))=\{s,t\}$ and the crossing
  number of $P$ wrt. $Q$ is $k$.
\end{lemma}
\begin{proof}
  Consider the graph $Z'$ (built using Lemma~\ref{l:k-paths}) from Section~\ref{s:single-hole}.
  Recall that for each vertex $v$ of $G^*_0[H](x^*,y^*)$, $Z'$ has $2n+1$ copies
  $v_{-n},\ldots,v_n$ of $v$ so that a path from any $u_0$ to $v_j$
  corresponds to an $u\to v$ path in $G^*_0[H](x^*,y^*)$ with crossing number
  $j$ wrt. $Q$.
  
  Pick any $s'\in V(h^1)$ and $t'\in V(h^2)$.
  Compute a shortest paths tree $T$ from $s'_0$ in $Z'$.
  This takes $O(|Z'|)=O(|H|^2)$ time.
  Denote by $T[v_j]$ the 
  $s'_0\to v_j$ path in that tree, and by $\len(T[v_j])$ its length.
  
  Similarly as in Lemma~\ref{l:k-paths}, one can prove that $T[t'_j]$ corresponds to some
  $s'\to t'$ path in $G^*_0[H](x^*,y^*)$ with crossing number wrt. $Q$
  exactly $j$ and not longer than
  the shortest \emph{simple} $s'\to t'$ path in that graph with crossing number
  $j$. Consequently, by proceeding as in the proof
  of Lemma~\ref{l:dhdist}, we obtain that the length of the shortest
  $s'\to t'$ path in $G^*_\lambda[H](x^*,y^*)$ equals
  \begin{equation*}
    \min_{j\in \{-n,\ldots,n\}}\{\len(T[t'_j])+\lambda\cdot j\},
  \end{equation*}
  and if some $j$ minimizes the above for a given $\lambda$, then
  the path $T[t_j']$ can be projected to an $s'\to t'$ path in $G^*_\lambda[H](x^*,y^*)$
  with cost precisely $\len(T[t'_j])+\lambda\cdot j$, that is, a shortest
  $s'\to t'$ path.

  In order to be able to compute a minimizer $j$ given some $\lambda$, we apply
  the following preprocessing.
  We want to compute a minimum of $2n+1=O(|H|)$ linear functions at coordinate~$\lambda$.
  It is well-known that such a minimum is the lower envelope of these functions
  and can be described using $k\leq 2n$ breakpoints $\lambda_1,\ldots,\lambda_{k}$
  and numbers $j_1,\ldots,j_k$, $j_1\geq \ldots \geq j_k$,
  such that for every $i=0,\ldots,k$, the minimum at $\lambda\in [\lambda_i,\lambda_{i+1})$
  equals $\len(T[t'_{j_i}])+\lambda\cdot j_i$ (where $\lambda_0=-\infty$ and $\lambda_{k+1}=\infty$).
  The lower envelope can be easily computed in linear time.
  With the breakpoints, locating the path $T[t'_{j^*}]$ that achieves the minimum
  for a given $\lambda$ can be done using binary search in $\Ot(1)$ time.

  The last problem is that the projection of $T[t'_{j^*}]$ might not correspond to a simple path
  in $G^*_\lambda[H](x^*,y^*)$, and it may contain more than two vertices of $V(h^1)\cup V(h^2)$.
  To deal with this problem, we make another preprocessing step on the tree $T$
  so that to any $T[v_j]$ with $v\in V(h^2)$ we assign a simple path $P^*_{v_j}\subseteq G^*_0[H](x^*,y^*)$ from $V(h^1)$ to $V(h^2)$
  such that $P^*_{v_j}$ is a projection of a subpath of $T[v_j]\subseteq Z'$
  onto  $G^*_0[H](x^*,y^*)$.
  Note that if (the projection of) $T[t'_{j^*}]$ is a shortest $s'\to t'$ path in 
  $G^*_\lambda[H](x^*,y^*)$, then $P^*_{t'_{j^*}}$, being a projection of a subpath
  of $T[t'_{j^*}]$, is also a shortest path in $G^*_\lambda[H](x^*,y^*)$
  that meets the desired
  requirements. Recall that we are only asked to return the endpoints
  of such a path, and its crossing number wrt.~$Q$.

  \newcommand{\dstr}{\mathcal{D}}

  We process the vertices $T$ in pre-order. On the way, we maintain a data structure $\dstr$.
  We maintain the following invariant: before the processing of $v_j$ (whose parent is $p_l$)
  starts, and also after it finishes, $\dstr$ stores some simple $s'\to p$ path $P_p$ in 
  $G^*_0[H](x^*,y^*)$ that is a projection of a subpath
  of $T[p_l]$.
  We will first describe how the data structure is used, and
  explain how its operations are implemented later.

  Let $e_{v_j}$ be the parent edge of $v_j$ in $T$, and let $e_v$ be the corresponding edge in $G^*_0[H](x^*,y^*)$.
  To process the vertex $v_j$, we first check whether the path $P_p$ contains $v$.
  Since $P_p$ is simple, it might contain at most one occurrence of $v$.
  Suppose $P_p=R_1R_2$, where $R_1=s'\to v$.
  In this case, the path $R_1$ forms a simple $s'\to v$ path $P_v$ that
  is a projection of a subpath of $T[p_l]$, so we may set $P_v:=R_1$.
  On the other hand, if $P_p$ does not contain $v$, 
  we may simply set $P_v:=P_p\cdot e_v$ and achieve the same.
  We update the data structure $\dstr$ so that it stores $P_v$.
  If $v\in V(h^2)$, 
  we use $\dstr$ to find the first vertex $t_v\in V(h^2)$ on $P_v$ (if exists),
  and then the latest vertex $s_v\in V(h^1)$ on $P_v$ that precedes the occurrence $t_v$.
  Note that the $s_v\to t_v$ subpath of $P_v$ is simple, does not contain
  any other vertices of $V(h^1)\cup V(h^2)$ apart from the endpoints $s_v,t_v$,
  and is a projection of a subpath of $T[v_j]$.
  We thus let $P^*_{v_j}$ equal that $s_v\to t_v$ subpath of $P_v$.
  We then record the endpoints $s_v,t_v$ of $P^*_{v_j}$ and its crossing number wrt. $Q$.
  Afterwards, we process the children of $v$ recursively. Note that the invariant posed
  on $\dstr$ is met. When all the children of $v$ are processed,
  we revert changes to $\dstr$ so that it again stores $P_p$ instead of $P_v$.
  
  We now explain how to implement the data structure $\dstr$ storing
  a simple path $P$ in $G^*_0[H](x^*,y^*)$.
  First of all, $\dstr$ stores two arrays $A,I$, and the number of edges $q$ of the stored path $P$.
  The array $A$ stores the subsequent vertices
  of $P$ at positions $1,\ldots,q$.
  $I$ is indexed with vertices of $H$:
  the vertex $v$ might appear on $P$ only at position $I[v]$ (if such a position exists).
  It might still hold that $A[I[v]]\neq v$; then $P$ does not contain the vertex $v$.

  Moreover, we store the positions of vertices from $V(h^1)$ on the path $P$ in an array $L$.
  More precisely, if the vertices of $V(h^1)$ appear on positions $l_1,\ldots,l_g$
  on $P$, then $L[i]=l_i$ for all $i=1,\ldots,g$.
  We also maintain the first occurrence of a vertex from $V(h^2)$ on the path $P$.

  Finally, we also store the crossing numbers of all prefixes of $P$ in
  an array $S$. That is, if $i\leq q$, then $S[i]$ is the crossing number
  of the subpath $s'\to A[i]$ of $P$.
  Storing prefix crossing numbers allows us to compute the crossing
  number of any subpath of $P$ in $O(1)$ time by subtracting
  two crossing numbers of path prefixes.

  Suppose once again we start processing a vertex $v_j$ with parent $p_l$.
  The array $I$ can locate the occurrence of $v$ on $P_p$, if it exists.
  To make $\dstr$ store $P_v$, we either simply decrease $q$ (to $I[v]$) if $P_p$ contains~$v$
  or increase $q$ by one and update $A[q]$, $S[q],I[v]$ in $O(1)$ time.
  If $v\in V(h^1)$, we also update $L$ if no occurrence of $v$ has been found in $P_p$,
  as follows.
  We find the first $L[i]$ with $L[i]\geq \ell$ (via binary search) and set $L[i]:=v$.
  Updating the first occurrence of a vertex from $V(h^2)$ is trivial.
  It is easy to check that this update process is
  correct and all the invariants are satisfied.
  Moreover, it takes $O(\log{n})$ time, and involves updating $O(1)$ values in the worst case.
  As a result, after descending down the tree recursively, we can efficiently revert all the updates to $\dstr$ in $O(1)$ time
  so that after the processing of a vertex $v_j$ ends, $\dstr$ stores $P_p$ again.

  To find the subpath $s_v\to t_v$ of $P_v$, it is enough to binary search
  for the latest position in $L$ that is smaller than $I[t_v]$. Note
  that this only happens if $v\in V(h^2)$, so $t_v$ is well-defined here,
  and the requested position can be found since $L[1]=1$.
  The crossing number of the $P^*_{v_j}=s_v\to t_v$ can be read from the array $S$.

  We conclude that finding the required endpoints and crossing numbers of $P_{v_j}^*$
  for all $v_j$ can be done using additional $\Ot(|T|)=\Ot(|Z'|)=\Ot(|H|^2)$ preprocessing time, as desired.
\end{proof}

Equipped with Lemmas~\ref{l:distinct-holes}~and~\ref{l:distinct-prep},
obtaining the closest pair and near neighbor data structures
for the subgraph of $D(H,x^*,y^*)$ encoding the distance between $V(h^1)$
and $V(h^2)$ is relatively easy.
First, during preprocessing, we compute the data structure of Lemma~\ref{l:distinct-prep}.
Given some $\lambda$, in $\Ot(1)$ time we can find such $s\in V(h^1)$,
$t\in V(h^2)$ that some simple shortest $s\to t$ path $P$ in $G^*_\lambda[H](x^*,y^*)$
has crossing number $k$, and $V(P)\cap (V(h^1)\cup V(h^2))=\{s,t\}$.
By Lemma~\ref{l:distinct-holes}, for every $(u',v')\in V(h^1)\times V(h^2)$,
some shortest $u'\to v'$ path in $G^*_\lambda[H](x^*,y^*)$
has crossing number between $k-4$ and $k+4$.

Similarly as in Section~\ref{s:single-hole}, we consider an auxiliary graph
\begin{equation*}
  D_{h^1,h^2}=\left(V(h^1)\cup V(h^2),\bigcup_{j\in [k-4,k+4]} E_{h^1,h^2,j}\right),
\end{equation*}
where each $E_{h^1,h^2,j}$ contains, for each $(u,v)\in V(h^1)\times V(h^2)$, an edge $uv$
with weight $\dist_{Z'}(u_0,v_j)+j\cdot \lambda$.

Analogously as in Lemma~\ref{l:dhdist}, from the definition of $k$, one can prove
that  for all \linebreak
${(u,v)\in V(h^1)\times V(h^2)}$, $\dist_{G^*_\lambda[H](x^*,y^*)}(u,v)=\dist_{D_{h^1,h^2}}(u,v)$.
Finally, observe that the proof of \linebreak
Lemma~\ref{l:partial-monge}
is valid not only for the graphs $D_h[E_{h,j}]$ from Section~\ref{s:single-hole}, but also for each of the graphs $D_{h^1,h^2}[E_{h^1,h^2,j}]$.
Indeed, there, the goal was to decompose distances between the distinct faces $h_0$ and $h_j$
of $Z'$. The proof goes through without any changes if we replace
$h_0$ with $h^1_0$, and $h_j$ with $h^2_j$.
Consequently, we obtain data structures with
 $T_\cp(D_{h^1,h^2}[E_{h^1,h^2,j}]),T_\nn(D_{h^1,h^2}[E_{h^1,h^2,j}])\in \Ot(|V(h^1)|+|V(h^2)|)=\Ot(|\bnd{H}|)$, 
    and $Q_\nn(D_{h^1,h^2}[E_{h^1,h^2,j}])=\Ot(1)$.
By combining those for $O(1)$ required values $j\in [k-4,k+4]$ as previously,
we obtain that 
 $T_\cp(D_{h^1,h^2}),T_\nn(D_{h^1,h^2})\in \Ot(|V(h^1)|+|V(h^2)|)=\Ot(|\bnd{H}|)$, 
  and $Q_\nn(D_{h^1,h^2})=\Ot(1)$.
  This concludes the proof of Lemma~\ref{l:exact-piece}.

\section{Flows, circulations, and negative cycle detection}\label{s:circ0}
Let $G_0=(V,E_0)$ be the input directed graph.
  Let $n=|V|$ and $m=|E_0|$.
  Define $G=(V,E)$ to be a \emph{flow network}, i.e., a multigraph such that $E=E_0\cup\rev{E_0}$, $E_0\cap \rev{E_0}=\emptyset$, where $\rev{E_0}$ is the set of \emph{reverse edges}.
  For any $uv=e\in E$, there is an edge $\rev{e}\in E$ such that $\rev{e}=vu$ and $\rev{(\rev{e})}=e$.
  We have $e\in E_0$ iff $\rev{e}\in \rev{E_0}$.

  Let $u:E_0\to \mathbb{Z}_+\cup\{\infty\}$ be an integral \emph{capacity function}.
  A \emph{flow} is a function $f:E\to \mathbb{R}$ such that for any $e\in E$
  $f(e)=-f(\rev{e})$ and for each $e\in E_0$, $0\leq f(e)\leq u(e)$.
  These conditions imply that for $e\in E_0$,
  $-u(e)\leq f(\rev{e})\leq 0$.
  We extend the function $u$ to $E$ by setting $u(\rev{e})=0$ for all $e\in E_0$.
  Then, for all edges $e\in E$ we have $-u(\rev{e})\leq f(e)\leq u(e)$.

  The \emph{excess} $\exc_f(v)$ of a vertex $v\in V$ is defined as $\sum_{uv=e\in E} f(e)$. Due
  to anti-symmetry of $f$, $\exc_f(v)$ is equal to the amount of flow going into $v$ by the
  edges of $E_0$ minus the amount of flow going out of $v$ by the edges of $E_0$.
  The vertex $v\in V$ is called an \emph{excess vertex} if $\exc_f(v)>0$ and a \emph{deficit vertex} if
  $\exc_f(v)<0$.   Let $X$ be the set of excess vertices of $G$ and
  let $D$ be the set of deficit vertices.
  Define the \emph{total excess} $\totexc_f$ as the sum of excesses of the excess vertices, i.e.,
  $\totexc_f=\sum_{v\in X} \exc_f(v)=\sum_{v\in D}-\exc_f(v).$

  A flow $f$ is called a \emph{circulation} if there are no excess vertices,
  or equivalently, $\totexc_f=0$.

  Let $c:E_0\to \mathbb{R}$ be the input \emph{cost} function.
  We extend $c$ to $E$ by setting $c(\rev{e})=-c(e)$ for all $e\in E_0$.
  The \emph{cost} $c(f)$ of a flow $f$ is defined as $\frac{1}{2}\sum_{e\in E} f(e)c(e)=\sum_{e\in E_0} f(e)c(e)$.

  \emph{To send a unit of flow} through $e\in E$ means to increase $f(e)$ by $1$ and simultaneously
  decrease $f(\rev{e})$ by $1$.
  By sending a unit of flow through $e$ we increase the cost of flow by $c(e)$.
  \emph{To send a unit of flow through a path $P$} means to send a unit of flow
  through each edge of $P$. In this case
  we also say that we \emph{augment flow $f$ along path $P$}.

  The \emph{residual network} $G_f$ of $f$ is defined as $(V,E_f)$, where $E_f=\{e\in E: f(e)<u(e)\}$.
  The capacity of an edge $e\in E_f$ in the residual network $G_f$ is defined as $u(e)-f(e)$.

For any $v\in V$, let us define:
\begin{equation*}
  \lambda_v=\min\left(\sum_{uv=e\in E} u(e),\sum_{vw=e\in E} u(e)\right).
\end{equation*}
For example, if $G$ is a unit-capacitated network, then $\lambda_v$ equals the minimum of the indegree and
the outdegree of $v$ in the underlying network $G_0$. Let us also put
\begin{equation*}
  \Lambda:=\sum_{v\in V}\lambda_v.
\end{equation*}
In the following we will assume that $\Lambda\geq n$, as vertices with $\lambda_v=0$
are of no use in the minimum-cost circulation problem.

Moreover, we will assume that $\Lambda$ is finite.
Observe that finite $\Lambda$ implies that any circulation has finite flow values.
As a result, the cost of a minimum-cost circulation is finite as well.

Let us partition $E_0$ into $E_\fin\cup E_\infty$
so that for all $e\in E_\infty$, $u(e)=\infty$, and $E_\fin=E_0\setminus E_\infty$.
Let us put $G_\infty:=G\setminus E_\fin$ and $m_\fin=|E_\fin|$.
Let $C=\max(2,\max\{-c(e):e\in E_0\})$. 

In Section~\ref{s:circ} we prove:

\begin{theorem}\label{t:circ}
  Suppose closest pair and near neighbor data structures have been preprocessed for $G_\infty$.
  Let $f^*$ be a min-cost
  circulation in $G$. Then, an \emph{integral} circulation $f_\delta$ satisfying
  \begin{equation*}
    c(f^*)\leq c(f_\delta)\leq c(f^*)+\delta
  \end{equation*}
  can be computed in time:
  \begin{equation*}
    \Ot\left(\sqrt{\Lambda}\log{(\Lambda C/\delta)}\cdot \left(T_\cp(G_\infty)+T_\nn(G_\infty)+m_\fin+\Lambda\cdot Q_\nn(G_\infty)\right) \right).
  \end{equation*}
  Within the same bounds, a price function $\pi$ satisfying
  $c_\pi(e)\geq -\delta/\Lambda$ for all $e\in E(G_{f_\delta})$ is produced.
\end{theorem}
Before proving Theorem~\ref{t:circ}, let us deduce Theorem~\ref{t:negcyc} from it.
\tnegcyc*
\begin{proof}
We use a reduction similar to that of~\cite{Gabow85, GabowT89}.
Consider a flow network $G'$ obtained from $G$ by converting weights to costs, setting all capacities to $\infty$, vertex splitting,
and subsequently adding, for each $u\in V$,
an edge $u_\tin u_\tout$ of capacity $1$ and cost $0$. Note that for $G'$, $\Lambda=n$ holds.

  Observe that the subgraph of infinite-capacity edges $G_\infty$ of $G'$ satisfies
  $T_\cp(G_\infty)=\Ot(T_\cp(G))$ by Lemma~\ref{l:ssp-split}.
  Moreover, we have $T_\nn(G_\infty)=O(T_\nn(G))$ and
  $Q_\nn(G_\infty)=O(Q_\nn(G))$ by Observation~\ref{l:nn-split}.

  Clearly, $G$ has a negative cycle if and only if $G'$ has a negative cycle.
Let us now argue that $G'$ has a negative cycle if and only if the minimum-cost
circulation $f^*$ in $G'$ has negative cost.
Observe that a zero circulation (i.e., where
the flow on each edge is $0$) is a valid circulation of cost $0$.
  As a result, we have $c(f^*)\leq 0$.

  Suppose $G'$ has a negative cycle. Then, by sending a unit of flow through that
  cycle, we obtain a valid circulation of negative cost, i.e., we conclude $c(f^*)<0$.

  Now suppose that $c(f^*)<0$. It is well-known (see e.g.~\cite{networkflows}) that any circulation
  -- in particular $f^*$ --
  can be decomposed into cycles in $G'$ of positive flow.
  Since the sum of costs of these cycles is negative, at least one of them
  has to have a negative cost.
  
  Since the costs and capacities are integral or infinite, and $\Lambda$ is finite for $G'$, $c(f^*)$ is integral.
  As a result, $c(f^*)<0$ implies $c(f^*)\leq -1$.
  Hence, if we invoke Theorem~\ref{t:circ} to 
  compute an integer circulation $f_\delta$ with $\delta\leq\frac{1}{2}$, to be set later,
  then $c(f_\delta)\leq -1+\delta\leq -\frac{1}{2}<0$.
  We conclude that $c(f^*)<0$ implies $c(f_\delta)<0$.
  Since $f^*$ is an optimal circulation, $c(f_\delta)<0$ trivially implies $c(f^*)<0$.
  As a result, $G$ has a negative cycle iff $c(f_\delta)<0$.
  For the graph $G'$ we have $\Lambda=n$, and $G'$ has $O(n)$ finite-capacity edges, so computing $f_\delta$
  takes
  \begin{equation*}
    \Ot\left(\sqrt{n}\cdot (T_\cp(G_\infty)+T_\nn(G_\infty)+n\cdot Q_\nn(G_\infty))\cdot \log(C/\delta)\right)
  \end{equation*}
  time by Theorem~\ref{t:circ}.
  Note that if one decomposes the circulation $f_\delta$ into cycles of positive flow,
  at least one of them will be negative. Such a cycle in~$G'$ can be easily converted
  to a simple negative cycle in $G$.

  However, if $c(f^*)=0$, i.e., $G$ has no negative cycle, we still need to produce a feasible price function of $G$.
  To this end we need the feasible price function $\pi$
  produced by Theorem~\ref{t:circ} when computing $f_\delta$.
  Note that $c(f_\delta)\leq \delta\leq \frac{1}{2}$ implies $c(f_\delta)=0$ by integrality of $f_\delta$.

  Since $f_\delta$ is a circulation, it can be decomposed into a collection of $j\geq 0$
  simple cycles $C_1,\ldots,C_j$,
  of positive flow (in fact, a unit flow).
  In fact, each of these cycles has to have $0$ cost, as otherwise at least
  one of them would have negative cost.
  Consider one of the cycles $C_i$. Since $C_i$ is a simple cycle,
  it has at most $2n$ edges (recall that vertex splitting doubles the number of vertices).
  Consider any edge $uv\in C_i$. Clearly, since $f_\delta(uv)=1$, 
  $vu\in G'_{f_\delta}$ and thus
  \begin{equation*}
    -c(uv)-\pi(v)+\pi(u)=c(vu)-\pi(v)+\pi(u)=c_\pi(vu)\geq -\frac{\delta}{n}
  \end{equation*}
  by Theorem~\ref{t:circ}.
  Equivalently, $c(uv)-\pi(u)+\pi(v)\leq \frac{\delta}{n}$.
  Suppose we have $c(xy)-\pi(x)+\pi(y)<-2\delta$ for some edge $xy\in C_i$.
  Then the sum
  \begin{equation*}
    \sum_{uv\in C_i}(c(uv)-\pi(u)+\pi(v))
  \end{equation*}
  has terms no more than $\frac{\delta}{n}$, and at least one term less than $-2\delta$,
  so it is less than ${\frac{(2n-1)\delta}{n}-2\delta<0}$.
  But since $C_i$ is a cycle, the prices $\pi(\cdot)$
  in the above sum cancel out, so we actually obtain
  $\sum_{uv\in C_i}c(uv)<0$ which contradicts that $C_i$ has zero cost.
  Therefore, we conclude that for any edge $uv\in C_i$ we have:
  \begin{equation}\label{eq:tight}
    c(uv)-\pi(u)+\pi(v)\geq -2\delta.
  \end{equation}
  Suppose that for some $v$ we have $f_\delta(v_\tin v_\tout)=1$. Then $v_\tin v_\tout$
  lies on some cycle $C_i$ of the decomposition and thus by~\eqref{eq:tight} we have:
  \begin{equation*}
  -\pi(v_\tin)+\pi(v_\tout)=c(v_\tin v_\tout)-\pi(v_\tin)+\pi(v_\tout)\geq -2\delta.
  \end{equation*}
  As a result $\pi(v_\tin)\leq \pi(v_\tout)+2\delta$ in this case.
  
  If $f_\delta(v_\tin v_\tout)=0$, then $v_\tin v_\tout\in G_{f_\delta}'$, so we can also conclude $\pi(v_\tin)\leq \pi(v_\tout)+2\delta$
  from $c(v_\tin v_\tout)-\pi(v_\tin)+\pi(v_\tout)\geq -\delta/n$.

  Now consider some original edge $uv$ of $G$. There are two cases.
  If $f_\delta(u_\tout v_\tin)=1$, then
  $u_\tout v_\tin$ lies on some cycle $C_i$ and therefore by~\eqref{eq:tight} we know that
  \begin{equation*}
    c(u_\tout v_\tin)-\pi(u_\tout)+\pi(v_\tin)\geq -2\delta
  \end{equation*}
  
  Observe that in this case the edge $u_\tin u_\tout$ also necessarily lies on $C_i$, so again by~\eqref{eq:tight}
  we have $-\pi(u_\tout)-2\delta\leq -\pi(u_\tin)$. From that, we conclude:
  \begin{equation*}
    \wei_G(uv)-\pi(u_\tin)+\pi(v_\tin)=c(u_\tout v_\tin)-\pi(u_\tin)+\pi(v_\tin)\geq c(u_\tout v_\tin)-\pi(u_\tout)+\pi(v_\tin)-2\delta\geq -4\delta.
  \end{equation*}

  Now assume $f_\delta(u_\tout v_\tin)=0$. Then $u_\tout v_\tin\in G_{f_\delta}'$ and by $\pi(u_\tout)\geq \pi(u_\tin)-2\delta$
  we have:
  \begin{equation*}
    \wei_G(uv)-\pi(u_\tin)+\pi(v_\tin)\geq c(u_\tout v_\tin)-\pi(u_\tout)+\pi(v_\tin)-2\delta\geq -\frac{\delta}{n}-2\delta\geq -4\delta.
  \end{equation*}
  We conclude that a price function $p(v):=\pi(v_\tin)$ satisfies $w_G(uv)-p(u)+p(v)\geq -4\delta$ for all
  original edges $uv\in E(G)$.

  Finally, consider the graph $G''$ obtained from $G$ by  adding,
for an arbitrarily chosen vertex $s\in V(G)$,
  the $n-1$ edges $st$, for all $t\in V\setminus\{s\}$, of weight $\lceil\max(nC,p(s)-p(t))\rceil$. 
  We have $T_\cp(G'')=\Ot(T_\cp(G))$ by Fact~\ref{l:ssp-dummy-clo} and Lemma~\ref{l:ssp-sum}.
  Now, let $G^+$ be obtained from $G''$ by increasing all the edge weights by $4\delta$.
  Note that $p$ is a feasible price function of $G^+$.
  As a result, by an easy modification of Dijkstra's algorithm implementation behind Lemma~\ref{l:dijkstra-add}
  (it is enough to add $4\delta$ to the right-hand side of~\eqref{eq:dijkstra-min}), one can compute distances $\dist_{G^+}(s,\cdot)$ in
  $\Ot(T_\cp(G''))=\Ot(T_\cp(G))$ time.
  Since $G^+$ was obtained by increasing edge weights of $G''$ by $4\delta$, for any $t\in V(G)$ we have:
  \begin{equation*}
    \dist_{G''}(s,t)\leq \dist_{G^+}(s,t)< \dist_{G''}(s,t)+4n\delta.
  \end{equation*}
  In particular, if $\delta=\frac{1}{4n}$, then we have $\dist_{G''}(s,t)\leq \dist_{G^+}(s,t)<\dist_{G''}(s,t)+1$.
  But since $G$ has integral weights, all $\dist_{G''}(s,t)$ are integral.
  As a result, for $\delta=\frac{1}{4n}$ we can obtain $\dist_{G''}(s,\cdot)$ by simply
  rounding the respective distances in $G^+$ down to the nearest integers.
  To finish the proof, recall that $-\dist_{G''}(s,\cdot)$ is a feasible price function
  of $G''$ (and thus of $G\subseteq G''$) if $G''$ has no negative cycles.

  We conclude that the running time of the algorithm is indeed
  \begin{equation*}
    \Ot\left(\sqrt{n}\cdot (T_\cp(G)+T_\nn(G)+n\cdot Q_\nn(G))\cdot \log{C}\right).\qedhere
  \end{equation*}
\end{proof}

\section{Min-cost circulation algorithm}\label{s:circ}
  
  In this section we prove Theorem~\ref{t:circ}. We first need the following technical lemma.
  
\begin{lemma}\label{l:nonzeroedges}
  Let $f$ be any flow. For any $v\in V$, the sum of $|f(e)|$ over incoming edges
  $e=zv$ (or outgoing edges $e=vz$) is at most $|\exc_{f}(v)|+2\lambda_v$.
\end{lemma}
\begin{proof}
  For each $v$, let $\lambda^\tin_v=\sum_{xv=e\in E}u(e)$, and $\lambda^\tout_v=\sum_{vy=e\in E}u(e)$.
  Recall that $\lambda_v=\min(\lambda^\tin_v,\lambda^\tout_v)$.

  Let $p_v=\sum_{zv=e\in E}|f(e)|$.
  We have
  \begin{equation*}
    \exc_{f}(v)=\sum_{\substack{e=zv\\ f(e)\neq 0}} f(e)=\
    \sum_{\substack{e=zv\\ f(e)>0}}|f(e)|-\sum_{\substack{e=zv\\ f(e)<0}}|f(e)|=p_v-2\cdot \sum_{\substack{e=zv\\ f(e)<0}}|f(e)|=-p_v+2\cdot \sum_{\substack{e=zv\\ f(e)>0}}|f(e)|,
  \end{equation*}
  and thus
  \begin{align*}
    p_v&=\exc_{f}(v)+2\cdot \sum_{\substack{e=zv\\ f(e)<0}}|f(e)|=\exc_{f}(v)+2\cdot \sum_{\substack{e=vz\\ f(e)>0}}f(e)\leq \exc_{f}(v)+2\lambda^\tout_v.\\
    p_v&=-\exc_{f}(v)+2\cdot \sum_{\substack{e=zv\\ f(e)>0}}|f(e)|\leq -\exc_{f}(v)+2\lambda^\tin_v.
  \end{align*}
  Hence for any $v\in V$ we have:
  \begin{align*}
    p_v&\leq\min(-\exc_{f}(v)+2\lambda^\tin_v,\exc_{f}(v)+2\lambda^\tout_v)\\
    &\leq \min(|\exc_{f}(v)|+2\lambda^\tin_v,|\exc_{f}(v)|+2\lambda^\tout_v)\\
    &=|\exc_{f}(v)|+2\lambda_v.
  \end{align*}
  The proof that the sum of $|f(e)|$ over edges $e=vz$
  is at most $|\exc_{f}(v)|+2\lambda_v$ is completely analogous.
\end{proof}
Since in circulations, all vertex excesses are zero, we obtain the following.
\begin{corollary}\label{c:sumabscirc}
  If $f$ is a circulation, then $\sum_{e\in E}|f(e)|\leq 2\Lambda$.
\end{corollary}

  \subsection{Scaling framework for minimum-cost circulation.}

  \begin{fact}[\cite{circcycle}]\label{f:negcycle}
    Let $f$ be a circulation. Then $c(f)$ is minimum iff $G_f$ has no negative cycles.
  \end{fact}
  It follows that a circulation $f$ is minimum if there exists
  a feasible price function of the residual network $G_f$.

  \begin{definition}[\cite{BertsekasT88, GoldbergT90, Tardos85}]
    A flow $f$ is $\eps$-optimal wrt. price function $p$ if
    for any $uv=e\in E_f$, $c(e)-p(u)+p(v)\geq -\eps.$
  \end{definition}
  The above notion of $\eps$-optimality allows us, in a sense,
  to measure the optimality of a circulation: the smaller $\eps$,
  the closer to the optimum a circulation $f$ is.
  In particular, we have the following:

\begin{lemma}\label{l:circdiff}
  Let $f^*$ be a min-cost circulation. Let $f$ be an $\eps$-optimal circulation wrt $p$.
  Then:
  \begin{equation*}
    c(f)-c(f^*)\leq 2\eps\cdot \Lambda.
  \end{equation*}
\end{lemma}
\begin{proof}
    We have
  \begin{equation*}
    c_p(f)-c_p(f^*)=\frac{1}{2}\sum_{e\in E} (f(e)-f^*(e))c_p(e).
  \end{equation*}
  If $f(e)>f^*(e)$, then $\rev{e}\in E_{f}$ and hence $c_p(\rev{e})\geq -\eps$, and
  thus $c_p(e)\leq \eps$.
  Otherwise, if $f(e)<f^*(e)$ then $e\in E_{f}$ and $c_p(e)\geq -\eps$.
  In both cases $(f(e)-f^*(e))c_p(e)\leq |f(e)-f^*(e)|\cdot \eps\leq (|f(e)|+|f^*(e)|)\cdot \eps$.
  This ultimate inequality also holds if $f(e)=f^*(e)$.
  As a result, by Lemma~\ref{l:nonzeroedges}, we have
  \begin{align*}
    c_p(f)-c_p(f^*)&\leq \frac{1}{2}\eps \cdot \left(\sum_{e\in E}|f(e)|+\sum_{e\in E}|f^*(e)|\right)\\
    &\leq\frac{1}{2}\eps \cdot \left(\sum_{v\in V}\sum_{zv=e\in E}|f(e)|+\sum_{v\in V}\sum_{zv=e\in E}|f^*(e)|\right)\\
    &\leq\frac{1}{2}\eps \cdot \left(\sum_{v\in V}|\exc_f(v)|+|\exc_{f^*}(v)|+4\lambda_v\right).
  \end{align*}
  But for circulation all the excesses are zero, and thus we obtain $c_p(f)-c_p(f^*)\leq 2\eps\sum_{v\in V}\lambda_v=2\eps\Lambda.$

Finally, observe that the price functions do not influence the costs of circulations. 
  As a result, $c(f)=c_p(f)$ and $c(f^*)=c_p(f^*)$ and thus $c(f)-c(f^*)\leq 2\eps\Lambda$.
\end{proof}

  Recall that $C$ equals minus the most negative edge cost. Observe that a zero circulation is $C$-optimal wrt. a zero price function.
  Suppose we have a procedure $\textsc{Refine}(G,f_0,p_0,\eps)$ that, given an \emph{integral} circulation $f_0$ in $G$ that is $2\eps$-optimal wrt. $p_0$,
  computes a pair $(f',p')$ such that $f'$ is an \emph{integral} circulation in $G$, and it is $\eps$-optimal wrt.~$p'$.
  We use the general \emph{successive approximation} framework, due to Goldberg and Tarjan~\cite{GoldbergT90}, as given in Algorithm~\ref{alg:scaling}.
  Therefore, if we implement $\textsc{Refine}$ to run in $T(n,m)$ time,
  by Lemma~\ref{l:circdiff} we can compute a $\delta$-additive approximation
  of the minimum cost circulation in $G$ in $O(T(n,m)\log{(C\Lambda/\delta)})$ time.

\begin{algorithm}[]
  \begin{algorithmic}[1]
    \Procedure{MinimumCostCirculation}{$G$}
  \State $f(e):=0$ for all $e\in G$
  \State $p(v):=0$ for all $v\in V$
  \State $\eps:=C/2$
  \While{$\eps>\frac{\delta}{2\Lambda}$}\Comment{$f$ is $2\eps$-optimal wrt. $p$}
    \State $(f,p):=\Call{Refine}{G,f,p,\eps}$
    \State $\eps:=\eps/2$
  \EndWhile
    \State \Return $f$ \Comment{$f$ is circulation $\frac{\delta}{2\Lambda}$-optimal wrt. $p$, i.e., a $\delta$-additive approximation}
  \EndProcedure
\end{algorithmic}
  \caption{Scaling framework for min-cost circulation.\label{alg:scaling}}
\end{algorithm}
  \subsection{Refinement}\label{s:general}
  We now describe our implementation of $\textsc{Refine}(G,f_0,p_0,\eps)$. A significant part of the analysis
  is based on the successive shortest paths-based unit-capacity min-cost flow algorithm described in~\cite{KarczmarzS19}, and
  is also inspired by the analysis in~\cite{GoldbergHKT17};
  however, the details of the refinement procedure are different.

  The first step is to adjust $p_0$ as follows:
  \begin{itemize}
    \item for each $v\in V$ with $\lambda^\tin_v\leq \lambda^\tout_v$, decrease $p_0(v)$ by $\eps$,
    \item for each $v\in V$ with $\lambda^\tin_v>\lambda^\tout_v$, increase $p_0(v)$ by $\eps$.
  \end{itemize}
  \begin{lemma}\label{l:relabel}
    After the adjustment, $p_0$ satisfies the following:
    \begin{enumerate}
      \item $f_0$ is $4\eps$-optimal wrt. $p_0$.
      \item The total capacity of edges $e\in G_{f_0}$ with $c_{p_0}(e)<0$ is no more than $3\Lambda$.
    \end{enumerate}
  \end{lemma}
  \begin{proof}
    First, for each $uv=e\in G_{f_0}$, $c_{p_0}(e)$ can decrease by at most $2\eps$,
    so item (1) follows by $2\eps$-optimality of the initial price function $p_0$.
  
    Now consider item (2). If $e\in E_0$, then the capacity of $\rev{e}$ in $G_{f_0}$ is $u(\rev{e})-f_0(\rev{e})=-f_0(\rev{e})=f_0(e)=|f_0(e)|$.
    As a result, the total capacity of the edges from $\rev{E_0}\cap E_{f_0}$ is
    no more than $\sum_{e\in E_0}|f_0(e)|\leq \sum_{e\in E}|f_0(e)|\leq 2\Lambda$, by Corollary~\ref{c:sumabscirc}.
    
    Now suppose $e\in E_0$. An edge $uv=e\in G_{f_0}$ can only satisfy $c_{p_0}(e)<0$
    after the adjustment if either $p_0(u)$ was increased or $p_0(v)$ was decreased,
    since otherwise $c_{p_0}(e)$ grows by $2\eps$.

    The capacity of $e$ in $G_{f_0}$ is at most $u(e)$.
    If $p_0(u)$ was increased, then, $\lambda_u=\lambda^\tout_u$, so we can
    charge the capacity of $e$ to $\lambda_u$.
    If $p_0(v)$ was decreased, then $\lambda_v=\lambda^\tin_v$,
    so we can charge the capacity of $e$ to $\lambda_v$.
    Observe that this way, the total charge is $\sum_{v\in V}\lambda_v=\Lambda$,
    so the total capacity of edges $E_0\cap E_{f_0}$ in $G_{f_0}$ is $\leq \Lambda$.
    Finally, the total capacity of $(E_0\cup \rev{E_0})\cap E_{f_0}=E_{f_0}$ is at most $3\Lambda$.
  \end{proof}
  The second step is setting $c(e):=c(e)-p_0(u)+p_0(v)$ for all $e\in E$ for further convenience.
  After the refinement completes, i.e., when we have a circulation $f'$ that is $\eps$-optimal wrt. some $p'$,
  (assuming costs reduced with $p_0$, as above), we will return $(f',p'+p_0)$ instead.
  Therefore, from now on we assume that $c(e)\geq -4\eps$ for all $e\in E_{f_0}$.

  Let $f_1$ be the flow initially obtained from $f_0$ by sending a unit of
  flow through each edge $e\in E_{f_0}$ such that $c(e)<0$.
  Note that $f_1$ is $\eps$-optimal (in fact $0$-optimal), but
  it need not be a circulation.
  Note that by Lemma~\ref{l:relabel}, the total excess $\totexc_{f_1}$ after this step is integral and
  bounded by $3\Lambda$.

  We denote by $f$ the \emph{current flow}, initialized with $f_1$, which we will gradually
  transform into a circulation.
  Recall that $X$ is the set of excess vertices of $G$ and
  $D$ is the set of deficit vertices (wrt. to the current flow $f$).
  A well-known method of finding the min-cost circulation exactly~\cite{succsp1, succsp2, succsp3}
  is to repeatedly send flow through shortest $X\to D$ paths in $G_f$.
  The sets $X$ and $D$ only shrink in time.
  However, doing this on $G_f$ exactly would be too costly.
  Instead,
  we will gradually convert $f$ into a circulation,
  by sending flow from vertices of $X$ to vertices of $D$
  but only using approximate (in a sense) shortest paths.

  For any $e\in E$, let us set
  \begin{equation*}
    c'(e)=\begin{cases}c(e) & \text{ if }e\in E_0,\\
      c(e)+\eps & \text{ if }e\in \rev{E_0}.
    \end{cases}
  \end{equation*}
  Note that for any $e\in E$ we have $c'(e)+c'(\rev{e})=c(e)+c(\rev{e})+\eps=\eps$.

  We define $G_f'$ to be the ``approximate'' graph $G_f$ with the costs given by $c'$ instead of $c$.
  Observe that $c(e)\geq c'(e)-\eps$, so if $p$ is a feasible price function of $G_f'$, then $f$ is $\eps$-optimal
  wrt. $p$ (in $G$).

  For convenience, let us also define an extended version $G_f''$ of $G_f'$ to be
  $G_f'$ with two additional vertices~$\source$ (a super-excess-vertex) and $\sink$ (a super-deficit-vertex) added.
  
  \newcommand{\aux}{{\mathrm{aux}}}
  We also add to  $G_f''$ the following \emph{auxiliary} edges $E_\aux$:
  \begin{enumerate}
    \item an edge $\source x$ with $c'(\source x)=0$ and capacity $\exc_{f_1}(x)$ for all $x\in X$.
    \item an edge $d \sink$ with $c'(d \sink)=$ and capacity $-\exc_{f_1}(d)$ for all $d\in D$.
  \end{enumerate}
  Clearly, $\dist_{G_f''}(\source,\sink)=\dist_{G_f'}(X,D)$.
  
  Set $M$ to be larger than the cost of any simple path in $G_f'$, e.g., $M=n\cdot (C+2\max_{v\in V}|p_0(v)|+\eps)$.
  Let $A$ be a graph on $V$ with an edge $\source v$ with $c'(\source v)=M$ for all $v\in V$.

  Now we describe the \emph{main loop} of the refinement procedure.
  Start with $f=f_1$ and zero price function $p$ on $G_f''$. Maintain an invariant the $p$ is a feasible price function
  on $G_f''$.
  While $X\neq \emptyset$:
  \begin{enumerate}[label={(\arabic*)}]
    \item Compute distances $d(v)=\dist_{G_f''\cup A}(s,v)$ from $s$ in $G_f''\cup A$ with the help of $p$.
    \item Set $p(v):=-d(v)$ for all $v\in V(G_f'')$.
    \item Repeatedly send unit flow through $X\to D$ paths $P\subseteq G_f$
      consisting solely of \emph{nearly tight} edges $uv=e$ satisfying $c'(e)-p(u)+p(v)<\eps/2$,
      until $\totexc_f=0$ or there are no such paths left.
  \end{enumerate}

  Whenever $X\neq\emptyset$, let us denote by $\Delta=\dist_{G_f''}(s,t)=\dist_{G_f}(X,D)$.

  We now state some crucial properties of the refinement procedure that hold regardless
  of how the steps (1) and (3) are implemented.
  The proofs of Lemmas~\ref{l:finite-delta} and~\ref{l:mainbound} are deferred to Section~\ref{s:mainbound}.
  \begin{restatable}{lemma}{lfinitedelta}\label{l:finite-delta}
    If $X\neq\emptyset$, then there exists an $X\to D$ path in $G_f$ (or equiv., in $G_f'$), i.e., $\Delta<\infty$.
  \end{restatable}
  Equivalently, if $X\neq\emptyset$, an $s\to t$ path exists in $G_f''$.

  \begin{lemma}\label{l:p-feasible}
    At all times, $p$ remains a feasible price function of $G_f''$.
  \end{lemma}
  \begin{proof}
    This is clearly satisfied at the very beginning. If $p$ is changed to $-\dist_{G_f''\cup A}(s,\cdot)$ in step (2),
    then it is a feasible price function of $G_f''$ since $G_f''$ has no negative cycles
    (by the feasibility of $p$ before step~(2)), and minus distance from $s$ function is a feasible price function.

    Now let us argue that $p$ remains a feasible price function after sending flow through
    some path $P$ in $G_f$ consisting of nearly tight edges. Let $f_2$ be the flow $f$ after augmenting along $P$.
    We clearly have $c'_p(e)\geq 0$ for all edges in $E(G_f'')\cap E(G_{f_2}'')$.
    Each edge $uv=e\in E(G_{f_2}'')\setminus E(G_f')$ is such that $e\notin E_\aux$ and
    $\rev{e}=vu$ belongs to the path~$P$.
    As a result, we have
    \begin{equation*}
      c'(e)-p(u)+p(v)=(-c'(\rev{e})+\eps)-p(u)+p(v)=\eps-(c'(\rev{e})-p(v)+p(u))>\eps-\eps/2=\eps/2,
    \end{equation*}
    so indeed $p$ is a feasible price function after augmenting the flow.
  \end{proof}

  \begin{lemma}\label{l:delta-growth}
    $\Delta$ increases by at least $\eps/2$ as a result of completing step (3) of the main loop.
  \end{lemma}
  \begin{proof}
    Note that immediately before some fixed step (3), $G_f''$ has no negative cycles by~Lemma~\ref{l:p-feasible},
    so $p(s)$ is reset to $0$ and $p(t)$ is reset to $-\Delta$.
    Suppose $X\neq\emptyset$ immediately after that step (3) and let $f_2$ be the flow at that time.
    Let $\Delta_2=\dist_{G_{f_2}'}(X,D)=\dist_{G_{f_2}''}(s,t)$ and consider 
    some path $s\to t=P=e_1\ldots e_k$ in $G_{f_2}''$ with minimum cost $\Delta_2$.
    We have:
    \begin{equation*}
      \Delta_2=\dist_{G_{f_2}''}(s,t)=\sum_{i=1}^k c'(e_i)=p(s)-p(t)+\sum_{i=1}^k c'_p(e_i)=\Delta+\sum_{i=1}^k c'_p(e_i).
    \end{equation*}
    Since $p$ is a feasible price function of $G_{f_2}''$, every term $c_p'(e_i)$ in the sum above is non-negative.
    However, the step (3) has completed, so for at least one edge $e_j$ we have $c_p'(e_j)\geq \eps/2$,
    as otherwise more flow could be sent through $P$ in that step (3).
    As a result, $\sum_{i=1}^k c'_p(e_i)\geq \eps/2$ and thus $\Delta_2\geq \Delta+\eps/2$, as desired.
  \end{proof}

  \begin{restatable}{lemma}{lmainbound}\label{l:mainbound}
    At all times, we have $\totexc_f\cdot \Delta\leq 48\eps \cdot \Lambda$.
  \end{restatable}

\begin{lemma}\label{l:main-loop-count}
  The main loop is performed $O(\sqrt{\Lambda})$ times.
\end{lemma}
\begin{proof}
  Since $G_{f_1}$ has only non-negative-cost edges, $\Delta$ is initially non-negative.
  By Lemma~\ref{l:finite-delta}, $\Delta$ remains finite as long as $\totexc_f>0$. Note that $\Delta$ can only change
  when $f$ evolves, i.e., in step~(3). By Lemma~\ref{l:delta-growth}, it grows by at least $\eps/2$ after each
  iteration of the main loop.
  Thus, there can be no more than $O(\sqrt{\Lambda})$ augmentations until $\Delta\geq \eps\cdot\sqrt{\Lambda}$.
  But then, by Lemma~\ref{l:mainbound}, we have $\totexc_f\leq 48\eps\Lambda/\Delta\leq 48\sqrt{\Lambda}$.
  Each iteration of step~(3) decreases the total excess by at least one, so there can be at most $O(\sqrt{\Lambda})$
  additional iterations of that step, and thus also of the main loop.
\end{proof}

\subsection{Computing distances}
In this section we discuss how step~(1) is implemented.
Note that $G_\infty\subseteq G_f''$ at all times, since for all $e\in E_0\cap E_\infty$
we have $c'(e)=c(e)$ and we always have $f(e)<u(e)=\infty$.

\begin{lemma}\label{l:non-infty}
  The number of edges in $E(G_f'')\setminus E_\infty$ is $O(m_\fin+\Lambda)$.
\end{lemma}
\begin{proof}
Consider any edge $e\in E(G_f'')\setminus E_\infty$.
  Then we either have (1) $e\in E_\fin\cup \rev{E_\fin}$ or (2) $e\in \rev{E_\infty}$,
or (3) $e\in E_\aux$.
There are clearly only $O(m_\fin)$ edges of the first kind.
On the other hand, if $e\in \rev{E_\infty}$, then $f(\rev{e})\geq 1$.
But by Lemma~\ref{l:nonzeroedges}, the number of
such edges is at most $\sum_{e\in \rev{E_\infty}}f(\rev{e})\leq \sum_{e\in E}|f(e)|\leq 2(\totexc_f +\Lambda)$.
However, as argued before, $\totexc_f\leq \totexc_{f_1}\leq 3\Lambda$.
As a result, the number of edges of type (2) is no more than $8\Lambda$.
Finally, there are clearly $O(n)=O(\Lambda)$ edges of the third type.
\end{proof}
From the above we conclude that $G_f''$ can be seen as $G_\infty$ augmented with $O(m_\fin+\Lambda)$ edges.

Note that by combining the closest pair data structure for $G_\infty$ that we already
have preprocessed, with $O(m_\fin+n+\Lambda)$ single-edge data structures
of Fact~\ref{l:ssp-dummy-clo}, by Lemma~\ref{l:ssp-sum} we obtain that
$T_\cp(G_f'')=\Ot(T_\cp(G_\infty)+m_\fin+\Lambda)$ without any additional preprocessing.

Since $p$ is a feasible price function of $G_f''$, by using Dijkstra's algorithm implementation of Lemma~\ref{l:dijkstra-add}, we obtain:
\begin{lemma}\label{l:distance-time}
  Step~(1) of the main loop takes $\Ot(T_\cp(G_\infty)+\Lambda+m_\fin)$ time.
\end{lemma}

\subsection{Sending flow efficiently}
  In this section we discuss how step~(3) is implemented.

\begin{observation}
  During step~(3), nearly tight edges are only deleted from $G_f''$, and never added.
  Moreover, only the edges $E(G_f'')\setminus E_\infty$ can be deleted.
\end{observation}
\begin{proof}
  An edge $e=uv$ can only be added to $G_f'$ if flow is sent through $\rev{e}$.
  But then $\rev{e}$ has to be nearly tight, so $c'(\rev{e})-p(v)+p(u)<\eps/2$.
  We get:
  \begin{equation*}
    c'(e)-p(u)+p(v)=c'(e)+c'(\rev{e})-(c'(\rev{e})-p(v)+p(u))>c'(e)+c'(\rev{e})-\eps/2=\eps/2,
  \end{equation*}
  so $e$
  is not nearly tight.

  An edge $e$ can only be deleted from $E_f$ if $f(e)=u(e)$, but this never happens for the edges
  $e\in E_\infty$ since $\sum_{e\in E}|f(e)|$ is $O(\Lambda)$.
\end{proof}
Denote by $H$ the subgraph of $G_f''$ consisting of nearly tight edges.
By the above observation, during step~(3) $H$ is decremental in time, and the edges of $E(H)\cap E_\infty$
are never deleted from $H$.

We will actually not only send flow through paths, but also through cycles. 
It is only important that we do not introduce more excess, which sending through cycles
satisfies. Note that
our analysis up to this point did not really depend on the fact
that the flow is pushed through a path. We only used the fact that
flow is pushed through nearly-tight edges.

In order to efficiently find paths/cycles, we will reuse some
of the information from previous searches.
We will maintain a set $S$ of vertices $v\in V(G_f'')$ for which we are still unsure
whether no more $v\to t$ paths consisting solely of nearly-tight edges exist in $G_f''$.
Initially, $S=V(G_f'')$.
The set $S$ will only shrink.
Moreover, for each vertex $v\in V(G_f'')$ we will maintain a set
\linebreak $F_v=\{vz=e\in E(H)\setminus E_\infty:z\in S\}$.

By Lemma~\ref{l:non-infty}, the initial size of sets $F_v$
is $O(m_\fin+\Lambda)$ and these sets can be easily
initialized within this
time bound. Moreover, since both sets $S$ and $E(H)\setminus E_\infty$ only shrink, so do the sets~$F_v$.
Observe that the sets $F_v$ can be also easily maintained subject to
deletions issued to the sets $S$ and $E(H)\setminus E_\infty$ in $O(m_\fin+\Lambda)$ total time.

Finally, whenever we enter step~(3), we initialize a near neighbor data structure
for $G_\infty$ with price function $p$ and thresholds $\tau(u)=p(u)+\eps/2$,
so that any returned edge $e=uv$ satisfies $\wei(e)+p(v)<\tau(u)=p(u)+\eps/2$, that
is $\wei(e)-p(u)+p(v)<\eps/2$.
The set $T$ of that data structure is identified with $S\cap V$ and all deletions to $S\cap V$
are passed to that data structure immediately. The total update time
of this data structure is $T_\nn(G_\infty)$ per single step~(3) run.

On the way to finding new paths/cycles, we store a (possibly empty) simple \emph{incomplete path}
$Q=s\to q$ in $H[S]$ for some $q\in V(G_f'')\setminus\{t\}$. Initially $Q=\emptyset$.

Given all the required data structures, the implementation of step~(3) is very simple.
While $s\in S$, we do the following.
While $q\notin t$, we take any edge $e=qz\in E(H)$ with $z\in S$.
If $F_q$ is non-empty, $e$ can be extracted from $F_q$ in $O(1)$ time.
Otherwise, such an edge has to come from $G_\infty$.
But then we can find such edge using a single query to the near neighbor data structure
in $Q_\nn(G_\infty)$ time.
If this fails as well, by the definition of $S$, the nearly tight edges originating in $q$
only lead to vertices $V\setminus S$ for which we already know that
cannot reach $t$ in $H$. So if finding a desired $e$ fails, we remove $q$
from $S$ and also remove the last edge of $Q$.

However, if $e$ has been found, we append $e$ to $Q$.
Since $Q$ was a simple path, this might cause some suffix $Q'$ of $Q$ now form a cycle of nearly tight edges.
In such a case we send a unit of flow through the cycle $Q'$ and
cut the suffix $Q'$ off $Q$.
On the other hand, if $z\notin V(Q)$, we extend $Q$ with $e$.
If it happens that $z=t$, we send flow through $Q[V]$ and reset $Q$ to be an empty path.

In the above, whenever flow is sent through a path or a cycle, $H$ (and thus the
sets $F_v$) are updated accordingly due to some edges whose flow has reached their capacity
being dropped from $G_f$.
This update time can be charged
to the number of edges of that path/cycle.

Note that when the entire procedure ends we have $Q=\emptyset$ since
an invariant $Q\subseteq H[S]$ is maintained.
To analyze the total time, we first observe that every insertion of an edge to the end of~$Q$
can be charged to its later deletion.
Moreover, every deletion can be charged to either removing some vertex
from $S$ or increasing a flow on an edge in $H$.
As a result, we can bound the total number of insertions or deletions to $Q$ by $O(n+\Phi)$, where
$\Phi$ is the total flow pushed \emph{through all edges} of~$H$ in the process. 
Since every step of the algorithm can be charged to either
a deletion or insertion of an edge to $Q$, the total time needed to execute
step~(3) can
be bounded by
\begin{equation*}
  O(T_\nn(G_\infty)+m_\fin+\Lambda+Q_\nn(G_\infty)\cdot (n+\Phi)).
\end{equation*}

\begin{lemma}
  We have $\Phi=O(\Lambda)$.
\end{lemma}
\begin{proof}
  Let $f'$ be the flow $f$ at the beginning of the run of step~(3) of our interest.
  Similarly, let $f''$ be the flow $f$ at the end of that run.

  Consider some edge $e$ of that undergoes a unit flow increase
  during the run of step~(3).
  Recall that $f(e)$ can only increase during that run since
  $\rev{e}$ is not nearly tight if $e$ is.
  As a result, the number of times $f(e)$ increases is exactly $f''(e)-f'(e)$.
  The total flow $\Phi$ sent through all edges in that run of step~(3) 
  can be thus bounded using Lemmas~\ref{l:nonzeroedges}~and~\ref{l:relabel} as follows:
  \begin{equation*}
    \Phi\leq \sum_{e\in E}|f''(e)-f'(e)|\leq \sum_{e\in E}|f''(e)|+\sum_{e\in E}|f'(e)|\leq (\totexc_{f'}+2\Lambda)+(\totexc_{f''}+2\Lambda)\leq 10\Lambda.\qedhere
  \end{equation*}
\end{proof}

Since $\Lambda\geq n$, we can thus conclude the following.
\begin{lemma}\label{l:send-time}
  Step~(3) of the main loop takes  $O(T_\nn(G_\infty)+m_\fin+Q_\nn(G_\infty)\cdot \Lambda)$ time.
\end{lemma}

By combining Lemmas~\ref{l:main-loop-count},~\ref{l:distance-time}~and~\ref{l:send-time},
we obtain that the refinement procedure can be implemented to run
in  $\Ot\left(\sqrt{\Lambda}\cdot \left(T_\cp(G_\infty)+T_\nn(G_\infty)+m_\fin+\Lambda\cdot Q_\nn(G_\infty)\right)\right)$
time, which immediately yields Theorem~\ref{t:circ}.

  \subsection{Proof of Lemma~\ref{l:mainbound}}\label{s:mainbound}
  In this section we prove the following key bound. On the way, we also prove Lemma~\ref{l:finite-delta}.
  \lmainbound*

  First suppose that $\Delta\leq 12\eps$. Since we also have $\totexc_f\leq \totexc_{f_1}\leq 3\Lambda$,
  the desired inequality follows immediately.
  So in the following let us assume $\Delta>12\eps$.

  Inspired by the unit-capacity min-cost flow analysis of~\cite{GoldbergHKT17}, we will consider the following auxiliary graph.
  Define $G^+_f=(V,E^+_f)$, where $E^+_f=\{e\in E: f(e)<f_0(e)\}.$
  For $e\in E^+_f$, set $u^+(e)=f_0(e)-f(e)$ and let us treat $u^+(e)$ as the capacity of $e$ in $G_f^+$.
  \begin{lemma}\label{l:cutbound}
    For any $S\subseteq V$ we have:
  \begin{equation*}
    \sum_{v\in S} \exc_f(v)\leq \sum\{u^+(e):ab=e\in E^+_f,a\in S, b\notin S\}.
  \end{equation*}
  \end{lemma}
  \begin{proof}
  Indeed, by antisymmetry we have
  \begin{equation*}
    \sum_{v\in S}\sum_{\substack{wv=e\in E\\w\in S}}f(e)=0,
  \end{equation*}
  and thus
  \begin{align*}
    \sum_{v\in S} \exc_f(v)&=\sum_{v\in S}\sum_{wv=e\in E}f(e)\\
    &=\sum_{v\in S}\sum_{\substack{wv=e\in E\\w\in S}}f(e)+\sum_{v\in S}\sum_{\substack{wv=e\in E\\w\in V\setminus S}}f(e)\\
    &=-\sum_{v\in S}\sum_{\substack{vw=e\in E\\w\in V\setminus S}}f(e).
  \end{align*}
  Applying the same transformation to $f_0$, and by the fact that circulations have zero vertex excesses:
  \begin{equation*}
    0=\sum_{v\in V}\exc_{f_0}(v)=-\sum_{v\in S}\sum_{\substack{vw=e\in E\\w\in V\setminus S}}f_0(e)
  \end{equation*}
  By subtracting the obtained equations, and skipping non-positive terms, we get:
 \begin{equation*}
   \sum_{v\in S}\exc_{f}(v)=\sum_{v\in S}\sum_{\substack{vw=e\in E\\w\in V\setminus S}}(f_0(e)-f(e))\leq \sum_{v\in S}\sum_{\substack{vw=e\in E^+_f\\w\in V\setminus S}}(f_0(e)-f(e))=\sum_{v\in S}\sum_{\substack{vw=e\in E^+_f\\w\in V\setminus S}}u^+(e).
  \end{equation*}
  \end{proof}
  Using Lemma~\ref{l:cutbound}, we can prove:
  \lfinitedelta*
  \begin{proof}
    Note that $E_f^+\subseteq E_f$, so $G_f^+\subseteq G_f$.
    As a result, it is enough to prove that a path $X\to D$ exists in $G_f^+$.
    Consider the set $Z$ of vertices reachable from $X$ in $G_f^+$ (via edges with $u^+(e)>0$),
    and suppose, for contradiction, that $Z\cap D=\emptyset$.
    By the definition of $Z$ and Lemma~\ref{l:cutbound}, the total capacity of edges $ab=e\in E^+_f\cap Z\times (V\setminus Z)$
    is at least $\sum_{z\in Z}\exc_f(z)\geq \totexc_f>0$.
    So there exists an edge $ab=e\in E^+_f\cap Z\times (V\setminus Z)$ with $u^+(e)>0$.
    But $b\notin Z$ contradicts the definition of $Z$.
  \end{proof}
  Now define, for any $v\in V$, $\lambda^{+\tout}_v=\sum_{vw=e\in E^+}u^+(e)$ and, symmetrically, \linebreak
  $\lambda^{+\tin}_v=\sum_{wv=e\in E^+}u^+(e)$.
  Set $\lambda^+_v=\min(\lambda^{+\tout}_v,\lambda^{+\tin}_v)$.
  
  \begin{lemma}\label{l:lambdaplus}
    For any $v\in V$, $\lambda^+_v\leq 2\lambda_v$.
  \end{lemma}
\begin{proof}
  We show that $\lambda^{+\tout}_v\leq 2\lambda^\tin_v$ and $\lambda^{+\tin}_v\leq 2\lambda^\tout_v$.
  This will give $\lambda^+_v=\min(\lambda^{+\tout}_v,\lambda^{+\tin}_v)\leq \min(2\lambda^{\tin}_v,2\lambda^{\tout}_v)=2\lambda_v$
  as desired.

  We have
  \begin{align*}
    \lambda^{+\tout}_v&=\sum_{vw=e\in E_f^+}u^+(e)\\
    &=\sum_{vw=e\in E_f^+\cap E_0}u^+(e)+\sum_{vw=e\in E_f^+\cap \rev{E_0}}u^+(e)\\
    &\leq\sum_{vw=e\in E_f^+\cap E_0}f_0(e)+\sum_{vw=e\in E_f^+\cap \rev{E_0}}(-f(e))\\
    &\leq\sum_{vw=e\in E_0}f_0(e)+\sum_{vw=e\in \rev{E_0}}(-f(e)).
  \end{align*}
  By the fact that $f_0$ is a circulation (and so $\exc_{f_0}(v)=0$) and antisymmetry:
 \begin{align*}
   \lambda^{+\tout}_v&\leq\sum_{vw=e\in \rev{E_0}}-f_0(e)+\sum_{wv=e\in E_0}f(e)\\
   &\leq\sum_{wv=e\in E_0}f_0(e)+\sum_{wv=e\in E_0}f(e)\\
   &\leq2\sum_{wv=e\in E_0}u(e)\\
   &=2\Lambda^\tin_v
  \end{align*}
  The proof that $\lambda^{+\tin}_v\leq 2\lambda^\tout_v$ is symmetric.
\end{proof}

  Consider some moment of the algorithm's runtime. For convenience, put $\delta(v):=\dist_{G_f''}(s,v)=-p(v)$.
  Recall that $\delta(x)\leq 0$ for all $x\in X$ and 
  and $\delta(d)\geq \Delta$ for all $d\in D$.

  For any $uv=e\in E_f^+\subseteq E_f$ we have
  \begin{equation*}
    \delta(v)-\delta(u)=p(u)-p(v)\leq c(e)+\eps=-c(\rev{e})+\eps,
  \end{equation*}
  but since $\rev{e}\in E_{f_0}$, $-c(\rev{e})\leq 4\eps$,
  and thus we have
  \begin{equation*}
  \delta(v)-\delta(u)\leq 5\eps.
  \end{equation*}
  For a number $z$, define $L(z)=\{w\in V: \delta(w)<z\}$.
  We have $X\subseteq L(z)$ for any $z>0$ and $D\cap L(z)=\emptyset$ for any $z\leq \Delta$.
  Consequently, for any $0<z\leq \Delta $
  we have 
  \begin{equation*}
    \sum_{v\in L(z)} \exc_f(v)=\sum_{v\in X} \exc_f(v)=\totexc_f.
  \end{equation*}

  Moreover, let $S(z)=\{w\in V: \delta(w)\in [z,z+5\eps]\}$ and $q=\left\lfloor\frac{\Delta-6\eps}{\eps}\right\rfloor\geq \lfloor \Delta/(2\eps)\rfloor\geq \Delta/(4\eps)$ (here, we have used the assumption $\Delta\geq 12\eps$).
  Consider the sets
  \begin{equation*}
    S(\eps),S(2\cdot \eps),\ldots,S(q\cdot \eps).
  \end{equation*}
  Note that each vertex $w\in V$ belongs to at most $6$ of the above sets.

  As a result, using Lemma~\ref{l:lambdaplus}, we get:
  \begin{equation*}
    \sum_{i=1}^q \sum_{w\in S(i\cdot \eps)} \lambda_w^+ \leq \sum_{w\in V} 6\lambda_w^+\leq \sum_{w\in V} 12\lambda_w\leq 12\Lambda.
  \end{equation*}
  Observe that there exists such an index $j\in\{1,\ldots,q\}$ that
  \begin{equation*}
    \sum_{w\in S(j\cdot \eps)} \lambda_w^+ \leq \frac{12\Lambda}{q}\leq \frac{48\eps\Lambda}{\Delta}.
  \end{equation*}
  Now consider the set
  \begin{equation*}
    A=L(j\cdot \eps)\cup \{w\in S(j\cdot \eps):\lambda_w^+=\lambda_w^{+\tout}\}.
  \end{equation*}
  Observe that $X\subseteq L(\eps)\subseteq L(j\cdot\eps)\subseteq A$,
  and so $\sum_{w\in A}\exc_f(w)=\totexc_f$.
  We now show that $\sum_{w\in A}\exc_f(w)$ can be bounded 
  by $\sum_{w\in S(j\cdot \eps)} \lambda_w^+\leq \frac{48\eps\Lambda}{\Delta}$.
  By Lemma~\ref{l:cutbound}, it is enough to show that
  \begin{equation}\label{eq:bb}
    \sum\{u^+(e):e=ab\in E^+_f, a\in A, b\in V\setminus A\}\leq \sum_{w\in S(j\cdot \eps)} \lambda_w^+.
  \end{equation}
  Take any edge $e=ab\in E_f^+$, $a\in A, b\in V\setminus A$.
  Recall that $\delta(b)-\delta(a)\leq 5\eps$, so if $a\in L(j\cdot \eps)$, then $b\in (V\setminus A)\cap S(j\cdot\eps)=\{y\in S(j\cdot \eps):\lambda_y^+<\lambda^{+\tout}_y\}$.
  In particular, we have $\lambda_b^+=\lambda_b^{+\tin}$.
  As a result, the total sum of $u^+(e)$ over all considered edges $e$ with $a\in L(j\cdot\eps)$ is
  no more than $\sum_{b\in (V\setminus A)\cap S(j\cdot \eps)}\lambda_b^+$.

  Note that the total sum of $u^+(e)$ over all edges $e=ab$ with $a\in A\cap S(j\cdot\eps)$
  is clearly no more than $\sum_{a\in A\cap S(j\cdot\eps)}\lambda_a^{+\tout}=\sum_{a\in A\cap S(j\cdot\eps)}\lambda_a^+$.
  Combining the two cases yields the desired bound~\eqref{eq:bb}.

\section{Omitted proofs for graph near-neighbor and closest-pair data structures}\label{s:omitted}

\sspdummyclo*
\begin{proof}
  For each~$v$ we simply maintain
  a pointer to the first edge $e=vu$ in the list of incident edges of $v$ that we have
  not yet detected that either $u\notin T$ or $\wei(e)+p(u)\geq \tau(v)$.
  When answering a query about $v$, we possibly advance the pointer at $v$
  in the list until we reach an edge that can be returned. 
  Since all the pointers only move forward, the total update time is $O(m+n)$.
  The amortized query time is clearly $O(1)$.
\end{proof}

\nnsum*
\begin{proof}
  To initialize the data structure, initialize individual near neighbor data structures $\ds_1,\ldots,\ds_k$.
  Denote by $T_1,\ldots,T_k$ the respective sets $T$ of the data structures.
  We maintain the invariant that $T_i=T\cap V_i$:
  upon deactivation of $v\in T$, it is passed to all $\ds_i$, where $v\in V_i$ (equiv. $v\in T_i$).

  Note that if one issues a query about vertex $x$ to $\ds_i$ and no edge within the threshold $\tau(v)$ is found,
  every subsequent query to $\ds_i$ about $x$ will fail as well.

  Let $X_v:=\{i:v\in V_i\}$. To handle queries, for each vertex $u\in V$ we maintain an index $z_u\in X_u\cup\{\infty\}$,
  satisfying the following invariant:
  for all $i\in X_u$ such that $i<z_u$ we have $w_{G_i}(e)+p(v)\geq \tau(v)$
  for all edges $uv=e\in E_i\cap (V_i\times T_i)$. Initially, we set $z_u=\min X_u$.

  Upon query, we first check whether $z_u=\infty$. If so, by the invariant, 
  no such ``near neighbor'' edge exists and we return $\perp$.
  Otherwise, we repeat the following until $z_u=\infty$ or we find some desired
  edge $uv$. Issue a query to $\ds_{z_u}$. If that query fails (i.e., returns $\perp$),
  advance $z_u$ to the succeeding element of $X_u\cup \{\infty\}$.
  If the query produces an edge, return that edge and stop.

  Note that immediately before $z_u$ advances, a query to $\ds_{z_u}$ is
  made. The total cost of such queries is $\sum_{i\in X_u}Q_\nn(G_i)$.
  Through all $u$, we have
  \begin{equation*}
    \sum_{u\in V}\sum_{i\in X_u}Q_\nn(G_i)=\sum_{i=1}^k |V_i|\cdot Q_\nn(G_i).
  \end{equation*}
  We charge the above queries issued to the ``local'' data structures to the total update time
  of the ``global'' data structure.
  Each query $u$ to the data structure for $\bigcup_{i=1}^k G_i$ results in
  at most one query to some $\ds_i$ that does not advance $z_u$.
  We charge this query to the query time of the global data structure which leads to
  the bound
  $Q_\nn(\bigcup_{i=1}^k G_i)=\max_{i=1}^k\{Q_\nn(G_i)\}+O(1)$.
  All the work that has not been charged can be easily seen to be $\sum_{i=1}^k O(|V_i|)$.
\end{proof}

\sspsum*
\begin{proof}
  Let $S,T\subseteq \bigcup_{i=1}^k V_i$.
  The combined data structure $\ds$ for $G_1\cup\ldots \cup G_k$ does no additional preprocessing,
  and, upon initialization, sets up
  separate closest pair data structures $\ds_i$ for each of the graphs $G_i$.
  Denote by $e_i^*,S_i,T_i,\alpha_i,\beta_i,d_i$ the respective parameters of $\ds_i$.
  In the $i$-th data structure we maintain that $S_i=S\cap V$, $T_i=T\cap V$,
  and $\alpha_i,\beta_i$ naturally match $\alpha,\beta$ respectively.

We also store the indices $i=1,\ldots,k$ in a priority queue $H$ so that
  the key of $i$ equals $d_i(e_i^*)$ and $H$ prioritizes smallest keys.
  Whenever some $\ds_i$ changes $e_i^*$, which happens at most $O(|V_i|)$ times, 
  the key of $i$ is updated in $H$ accordingly. 
  As a result, maintaining $H$ requires $O\left(\sum_{i=1}^k |V_i|\log{k}\right)$ time.
  Moreover, since each edge of $G$ is contained in some $G_i$, if $j$ is the top key
  of $H$, then $e_j^*=st$ minimizes $\wei_G(e_j^*)+\alpha(s)+\beta(t)$
  over all edges of $G$.
  
  The activations and extractions of vertices $v$ are simply passed to the relevant
  data structures $\ds$ such that $v\in V_i$. Consequently, the total update time
  is $\sum_{i=1}^k T_\cp(G_i)+\sum_{i=1}^k O(|V_i|\log{k})$.
\end{proof}

Below, for any $X\subseteq V$, we use the notation $X_\tin:=\{x_\tin:x\in X\}$, $X_\tout:=\{x_\tout:x\in X\}$.
\nnsplit*
\begin{proof}
  Suppose we need a near neighbor data structure $\ds'$ for $G'$ initialized with $T'$,
  price function~$p$, and thresholds $\tau'(u)$.
  Since the vertices $V_\tout$ have no incoming edges, we can put $T':=T'\cap V_\tin$.
  Moreover, since the vertices $V_\tin$ have no outgoing edges,
  the near neighbor queries for vertices in $V_\tin$ always fail.
  We thus focus on near neighbor queries for vertices in $V_\tout$.

  We simply use a near neighbor data structure $\ds$ for $G$ initialized with
  $T=\{v:v_\tout\in T\}$, price function $p(u):=p'(u_\tin)$, and thresholds
  $\tau(v):=\tau'(v_\tout)$.
  A query about vertex $x_\tout$ on $\ds'$ is translated to a query
  about $x$ in $\ds$.
\end{proof}

\sspsplit*
\begin{proof}
  Let us identify $V$ with the numbers $\{1,\ldots,n\}$.
  For each $b=1,\ldots,{\lfloor\log_2{n}\rfloor+1}$ and ${i\in\{0,1\}}$, denote by
  $A_{b,i}$ be the subset of numbers of $\{1,\ldots,n\}$ whose $b$-th bit in the binary representation equals $i$.
  
  Let $S,T,\alpha,\beta$ be the parameters of the desired data structure
  for $G'$. Let $T_0$ be the initial set~$T$.
  For each $b=1,\ldots,{\lfloor\log_2{n}\rfloor+1}$ and $i\in \{0,1\}$, we set up a closest pair data structure
  $\ds_{b,i}$ for $G$
  with the respective parameters $S_{b,i},T_{b,i},\alpha_{b,i},\beta_{b,i}$.  
  We maintain the invariants: $S_{b,i}=S\cap A_{b,i}$, $T_{b,i}=T\setminus A_{b,i}$, 
  $\alpha_{b,i}(u)=\alpha(u_\tout)$ for $u\in A_{b,i}$, and $\beta_{b,i}(v)=\beta(v_\tin)$ for
  $v\in T_0\setminus A_{b,i}$. 

  By the definition of vertex splitting, we can focus on activations
  of vertices $V_\tout$ and extractions of vertices $V_\tin$
  and assume $S\subseteq V_\tout$ and $T\subseteq V_\tin$.

  For each activation of some $u_\tout$ with $\alpha(u_\tout)$, we pass
  the activation of $u$ only to data structures $\ds_{b,i}$ 
  satisfying $u\in A_{b,i}$.
  Similarly, if some $v_\tin$ is extracted, we pass that extraction
  only to data structures $\ds_{b,i}$ satisfying $v\in T_{b,i}$.
  As a result, the data structure $\ds_{b,i}$ can be effectively seen
  as simulating a closest pair data structure of the graph
  \begin{equation*}
    G'[E'\cap ((A_{b,i})_\tout\times (V\setminus A_{b,i})_\tin)]=G'[E'\cap ((A_{b,i})_\tout\times (A_{b,1-i})_\tin)].
  \end{equation*}

  Note that for each edge $uv$ in $G$, $u\neq v$, there exists
  such $(b,i)$ that $u\in A_{b,i}$ and ${v\in V\setminus A_{b,i}}=A_{b,1-i}$:
  indeed, $u$ and $v$ differ at some bit in their binary representations.
  Therefore, the $O(\log{n})$ sets $E'\cap ((A_{b,i})_\tout\times (V\setminus A_{b,i})_\tin)]$
  actually cover the set $E'\cap (V_\tout\times V_\tin)=E'$,
  so by Lemma~\ref{l:ssp-sum} applied with $k=O(\log{n})$, there exists 
  a closest pair data structure for $G'$ with total update time $T_\cp(G')=O(T_\cp(G)\log{n}+|V|\log{n}\log\log{n})$
  that requires no additional preprocessing. The lemma follows by the assumption $T_\cp(G)=\Omega(|V|)$.
\end{proof}

\section{Dealing with non-simple or non-disjoint holes}\label{s:non-simple}

\subsection{Removing the assumption from Section~\ref{s:negcyc}}

If a piece has $P$ has non-simple holes, by repeatedly splitting its boundary vertices
that are repeated on some hole, one can convert it into a piece $P'$ 
with simple holes exclusively of size $|P|+O(|\bnd{P}|)=|P|+O(\sqrt{r})$,
plus a graph $P''$ consisting of $O(\sqrt{r})$ $0$-weight auxiliary edges with endpoints
in $\bnd{P'}$ such that (1) $\bnd{P}\subseteq \bnd{P'}$ and (2) for any
$u,v\in \bnd{P}$, $\dist_{P'\cup P''}(u,v)=\dist_{P}(u,v)$.
This is described in detail in~\cite[Section~5.1]{KaplanMNS17}.
Any feasible price function on $P$ can be trivially converted into a feasible
price function on $P'\cup P''$ by setting the price of every copy
to that of the original vertex.

By Lemmas~\ref{l:nn-sum}, \ref{l:ssp-sum} and~\ref{l:ddg-ds}, there exist
closest pair and near neighbor data structures with total update time $\Ot(\sqrt{r})$
for the graph $P'\cup P''$. This is asymptotically no worse than we could
have if $P$ had no non-simple holes.
The only property that we required from $P'\cup P''$
is that it preserves distances in $P$ between the vertices $\bnd{P}$.

\subsection{Removing the assumption from Section~\ref{s:flow-oracle}}

Recall that in Section~\ref{s:flow-oracle}, we require a recursive decomposition $\TG(G^*)$ of $G^*$ such that all pieces
have simple and pairwise vertex-disjoint holes.
Let us call such a decomposition \emph{simple}.

First of all, we can assume that $G^*$ is constant-degree: this can be achieved
easily by simply
triangulating the faces of $G$ using $0$-capacity edges at the very beginning.

In~\cite{ItalianoKLS17, ItalianoKLS17a} it is shown, for a constant-degree plane graph $G_0$,
how a fixed recursive decomposition using cycle separators $\TG(G_0)$,
as described n Section~\ref{s:flow-oracle},
can be converted into a simple recursive decomposition $\TG(G_1)$ of a certain augmented graph $G_1$
in $O(n\log{n})$ time.
The decomposition is of the same shape, i.e., there is a 1-1 correspondence between the pieces
of $\TG(G_0)$ and $\TG(G_1)$.
In the graph $G_1$, each vertex $v\in V(G_0)$ corresponds to a strongly connected induced
subgraph $G_1[S_v]$ of $G_1$, and each edge $e=uv$ of $G_0$ is replaced with a bunch of edges connecting
the subsets $S_u$ and $S_v$ in $G_1$.

For each piece $H_1\in \TG(G_1)$ corresponding to $H_0\in \TG(G_0)$, the faces
of $H_1$ can be split into ``natural'' faces, holes, and auxiliary faces.
The auxiliary faces are created when splitting original vertices of $G_0$ into subgraphs,
or edges into bunches of edges.
Importantly, there is a 1-1 correspondence between the holes of $H_0$ and the holes of $H_1$,
and a 1-1 correspondence between the natural faces of $H_0$ and ``natural'' faces of $H_1$.
Every two ``natural'' faces or holes of $H_1$ are fully separated (i.e., do not share vertices in $H_1$) with auxiliary faces and
that's why they are all simple and vertex-disjoint.
The piece $H_1$ is only constant-factor larger than $H_0$,
and similarly the boundary size $|\bnd{H_1}|$ is only constant-factor larger than $|\bnd{H_0}|$.

The transformation that $G_0$ needs to undergo on the way to becoming $G_1$ can be described
using a sequence of either (a) vertex-splits (splitting a vertex $v$ into two adjacent vertices or into a triangle, with embedding-respecting partition
of edges $v$ among the resulting vertices)
or (b) parallel edge introductions. 

We apply the conversion procedure of~\cite{ItalianoKLS17, ItalianoKLS17a} to the decomposition $\TG(G^*)$
of our dual graph $G^*$. Vertex splits in $G^*$ corresponds to introducing auxiliary vertices
inside some of the faces of $G$ -- we set the capacities of newly created edges to $0$ so that the cut/flow values
are preserved. Introducing parallel edges in $G^*$, in turn, corresponds to splitting an edge
of $G$ into two sequentially connected edges -- to preserve the cut/flow values,
the parallel edge needs to inherit the capacity/weight of the original edge.

Since the transformations that $G^*$ needs to undergo to have a simple recursive decomposition
do not change the flow/cut values, and the original vertices of $G$ (original faces of $G^*$)
are preserved in a 1-1 correspondence after the transformations, we conclude
that indeed the assumptions about the holes of Section~\ref{s:flow-oracle} can be dropped.

\bibliographystyle{alpha}
\bibliography{references}

\newcommand{\etalchar}[1]{$^{#1}$}
\newcommand{\sortkey}[1]{}
\begin{thebibliography}{vdBCK{\etalchar{+}}23}

\bibitem[ABC{\etalchar{+}}23]{abs-2303-00811}
Vikrant Ashvinkumar, Aaron Bernstein, Nairen Cao, Christoph Grunau, Bernhard
  Haeupler, Yonggang Jiang, Danupon Nanongkai, and Hsin{-}Hao Su.
\newblock Parallel and distributed exact single-source shortest paths with
  negative edge weights.
\newblock {\em CoRR}, abs/2303.00811, 2023.

\bibitem[AES99]{AgarwalES99}
Pankaj~K. Agarwal, Alon Efrat, and Micha Sharir.
\newblock Vertical decomposition of shallow levels in 3-dimensional
  arrangements and its applications.
\newblock {\em {SIAM} J. Comput.}, 29(3):912--953, 1999.

\bibitem[AHdLT05]{AlstrupHLT05}
Stephen Alstrup, Jacob Holm, Kristian de~Lichtenberg, and Mikkel Thorup.
\newblock Maintaining information in fully dynamic trees with top trees.
\newblock {\em {ACM} Trans. Algorithms}, 1(2):243--264, 2005.

\bibitem[AKL{\etalchar{+}}22]{AbboudK0PST22}
Amir Abboud, Robert Krauthgamer, Jason Li, Debmalya Panigrahi, Thatchaphol
  Saranurak, and Ohad Trabelsi.
\newblock Breaking the cubic barrier for all-pairs max-flow: Gomory-hu tree in
  nearly quadratic time.
\newblock In {\em 63rd {IEEE} Annual Symposium on Foundations of Computer
  Science, {FOCS} 2022}, pages 884--895. {IEEE}, 2022.

\bibitem[AKLR20]{AsathullaKLR20}
Mudabir~Kabir Asathulla, Sanjeev Khanna, Nathaniel Lahn, and Sharath
  Raghvendra.
\newblock A faster algorithm for minimum-cost bipartite perfect matching in
  planar graphs.
\newblock {\em {ACM} Trans. Algorithms}, 16(1):2:1--2:30, 2020.

\bibitem[AKM{\etalchar{+}}87]{AggarwalKMSW87}
Alok Aggarwal, Maria~M. Klawe, Shlomo Moran, Peter~W. Shor, and Robert~E.
  Wilber.
\newblock Geometric applications of a matrix-searching algorithm.
\newblock {\em Algorithmica}, 2:195--208, 1987.

\bibitem[AKT21]{AbboudKT21}
Amir Abboud, Robert Krauthgamer, and Ohad Trabelsi.
\newblock Subcubic algorithms for gomory-hu tree in unweighted graphs.
\newblock In {\em {STOC} '21: 53rd Annual {ACM} {SIGACT} Symposium on Theory of
  Computing}, pages 1725--1737. {ACM}, 2021.

\bibitem[AKT22]{AbboudKT22}
Amir Abboud, Robert Krauthgamer, and Ohad Trabelsi.
\newblock Friendly cut sparsifiers and faster gomory-hu trees.
\newblock In {\em Proceedings of the 2022 {ACM-SIAM} Symposium on Discrete
  Algorithms, {SODA} 2022}, pages 3630--3649. {SIAM}, 2022.

\bibitem[AMO93]{networkflows}
Ravindra~K. Ahuja, Thomas~L. Magnanti, and James~B. Orlin.
\newblock {\em Network flows - theory, algorithms and applications}.
\newblock Prentice Hall, 1993.

\bibitem[AWY18]{AbboudWY18}
Amir Abboud, Virginia~Vassilevska Williams, and Huacheng Yu.
\newblock Matching triangles and basing hardness on an extremely popular
  conjecture.
\newblock {\em {SIAM} J. Comput.}, 47(3):1098--1122, 2018.

\bibitem[BCF23]{sssp-logs}
Karl Bringmann, Alejandro Cassis, and Nick Fischer.
\newblock Negative-weight single-source shortest paths in near-linear time: Now
  faster!
\newblock {\em CoRR}, abs/2304.05279, 2023.

\bibitem[Bel58]{bellman1958routing}
Richard Bellman.
\newblock On a routing problem.
\newblock {\em Quarterly of applied mathematics}, 16(1):87--90, 1958.

\bibitem[BG60]{succsp1}
Robert~G Busacker and Paul~J Gowen.
\newblock A procedure for determining a family of minimum-cost network flow
  patterns.
\newblock Technical report, RESEARCH ANALYSIS CORP MCLEAN VA, 1960.

\bibitem[BK09]{BorradaileK09}
Glencora Borradaile and Philip~N. Klein.
\newblock An \emph{O}(\emph{n} log \emph{n}) algorithm for maximum
  \emph{st}-flow in a directed planar graph.
\newblock {\em J. {ACM}}, 56(2):9:1--9:30, 2009.

\bibitem[BKM{\etalchar{+}}17]{BorradaileKMNW17}
Glencora Borradaile, Philip~N. Klein, Shay Mozes, Yahav Nussbaum, and Christian
  Wulff{-}Nilsen.
\newblock Multiple-source multiple-sink maximum flow in directed planar graphs
  in near-linear time.
\newblock {\em {SIAM} J. Comput.}, 46(4):1280--1303, 2017.

\bibitem[BNW22]{BernsteinNW22}
Aaron Bernstein, Danupon Nanongkai, and Christian Wulff{-}Nilsen.
\newblock Negative-weight single-source shortest paths in near-linear time.
\newblock In {\em 63rd {IEEE} Annual Symposium on Foundations of Computer
  Science, {FOCS} 2022}, pages 600--611. {IEEE}, 2022.

\bibitem[BSW15]{BorradaileSW15}
Glencora Borradaile, Piotr Sankowski, and Christian Wulff{-}Nilsen.
\newblock Min \emph{st}-cut oracle for planar graphs with near-linear
  preprocessing time.
\newblock {\em {ACM} Trans. Algorithms}, 11(3):16:1--16:29, 2015.

\bibitem[BT88]{BertsekasT88}
Dimitri~P. Bertsekas and Paul Tseng.
\newblock Relaxation methods for minimum cost ordinary and generalized network
  flow problems.
\newblock {\em Operations Research}, 36(1):93--114, 1988.

\bibitem[CCE13]{CabelloCE13}
Sergio Cabello, Erin~W. Chambers, and Jeff Erickson.
\newblock Multiple-source shortest paths in embedded graphs.
\newblock {\em {SIAM} J. Comput.}, 42(4):1542--1571, 2013.

\bibitem[CE01]{ChanE01}
Timothy~M. Chan and Alon Efrat.
\newblock Fly cheaply: On the minimum fuel consumption problem.
\newblock {\em J. Algorithms}, 41(2):330--337, 2001.

\bibitem[Cha20]{Chan20}
Timothy~M. Chan.
\newblock Dynamic generalized closest pair: Revisiting eppstein's technique.
\newblock In {\em 3rd Symposium on Simplicity in Algorithms, {SOSA} 2020},
  pages 33--37. {SIAM}, 2020.

\bibitem[CKL{\etalchar{+}}22]{ChenKLPGS22}
Li~Chen, Rasmus Kyng, Yang~P. Liu, Richard Peng, Maximilian~Probst Gutenberg,
  and Sushant Sachdeva.
\newblock Maximum flow and minimum-cost flow in almost-linear time.
\newblock In {\em 63rd {IEEE} Annual Symposium on Foundations of Computer
  Science, {FOCS} 2022}, pages 612--623. {IEEE}, 2022.

\bibitem[DGG{\etalchar{+}}22]{DongGGLPSY22}
Sally Dong, Yu~Gao, Gramoz Goranci, Yin~Tat Lee, Richard Peng, Sushant
  Sachdeva, and Guanghao Ye.
\newblock Nested dissection meets ipms: Planar min-cost flow in nearly-linear
  time.
\newblock In {\em Proceedings of the 2022 {ACM-SIAM} Symposium on Discrete
  Algorithms, {SODA} 2022}, pages 124--153. {SIAM}, 2022.

\bibitem[EFL18]{EricksonFL18}
Jeff Erickson, Kyle Fox, and Luvsandondov Lkhamsuren.
\newblock Holiest minimum-cost paths and flows in surface graphs.
\newblock In {\em Proceedings of the 50th Annual {ACM} {SIGACT} Symposium on
  Theory of Computing, {STOC} 2018}, pages 1319--1332. {ACM}, 2018.

\bibitem[EK72]{EdmondsK72}
Jack Edmonds and Richard~M. Karp.
\newblock Theoretical improvements in algorithmic efficiency for network flow
  problems.
\newblock {\em J. {ACM}}, 19(2):248--264, 1972.

\bibitem[Eri10]{Erickson10}
Jeff Erickson.
\newblock Maximum flows and parametric shortest paths in planar graphs.
\newblock In {\em Proceedings of the Twenty-First Annual {ACM-SIAM} Symposium
  on Discrete Algorithms, {SODA} 2010}, pages 794--804, 2010.

\bibitem[ET75]{EvenT75}
Shimon Even and Robert~Endre Tarjan.
\newblock Network flow and testing graph connectivity.
\newblock {\em {SIAM} J. Comput.}, 4(4):507--518, 1975.

\bibitem[FJ56]{ford1956network}
Lester~R Ford~Jr.
\newblock Network flow theory.
\newblock Technical report, Rand Corp Santa Monica Ca, 1956.

\bibitem[FR06]{FR06}
Jittat Fakcharoenphol and Satish Rao.
\newblock Planar graphs, negative weight edges, shortest paths, and near linear
  time.
\newblock {\em J. Comput. Syst. Sci.}, 72(5):868--889, 2006.

\bibitem[Fre87]{DBLP:journals/siamcomp/Frederickson87}
Greg~N. Frederickson.
\newblock Fast algorithms for shortest paths in planar graphs, with
  applications.
\newblock {\em {SIAM} J. Comput.}, 16(6):1004--1022, 1987.

\bibitem[Gab85]{Gabow85}
Harold~N. Gabow.
\newblock Scaling algorithms for network problems.
\newblock {\em J. Comput. Syst. Sci.}, 31(2):148--168, 1985.

\bibitem[GH61]{gomory1961multi}
Ralph~E Gomory and Tien~Chung Hu.
\newblock Multi-terminal network flows.
\newblock {\em Journal of the Society for Industrial and Applied Mathematics},
  9(4):551--570, 1961.

\bibitem[GHKT17]{GoldbergHKT17}
Andrew~V. Goldberg, Sagi Hed, Haim Kaplan, and Robert~E. Tarjan.
\newblock Minimum-cost flows in unit-capacity networks.
\newblock {\em Theory Comput. Syst.}, 61(4):987--1010, 2017.

\bibitem[GK18]{GawrychowskiK18}
Pawel Gawrychowski and Adam Karczmarz.
\newblock Improved bounds for shortest paths in dense distance graphs.
\newblock In {\em 45th International Colloquium on Automata, Languages, and
  Programming, {ICALP} 2018}, pages 61:1--61:15, 2018.

\bibitem[GMW20]{GawrychowskiMW20}
Pawel Gawrychowski, Shay Mozes, and Oren Weimann.
\newblock Submatrix maximum queries in monge and partial monge matrices are
  equivalent to predecessor search.
\newblock {\em {ACM} Trans. Algorithms}, 16(2):16:1--16:24, 2020.

\bibitem[GMWW18]{vorexact}
Pawel Gawrychowski, Shay Mozes, Oren Weimann, and Christian Wulff{-}Nilsen.
\newblock Better tradeoffs for exact distance oracles in planar graphs.
\newblock In {\em Proceedings of the Twenty-Ninth Annual {ACM-SIAM} Symposium
  on Discrete Algorithms, {SODA} 2018}, pages 515--529, 2018.

\bibitem[GN80]{GalilN80}
Zvi Galil and Amnon Naamad.
\newblock An o(evlog{\({^2}\)}v) algorithm for the maximal flow problem.
\newblock {\em J. Comput. Syst. Sci.}, 21(2):203--217, 1980.

\bibitem[Gol95]{Goldberg95}
Andrew~V. Goldberg.
\newblock Scaling algorithms for the shortest paths problem.
\newblock {\em {SIAM} J. Comput.}, 24(3):494--504, 1995.

\bibitem[GR98]{GoldbergR98}
Andrew~V. Goldberg and Satish Rao.
\newblock Beyond the flow decomposition barrier.
\newblock {\em J. {ACM}}, 45(5):783--797, 1998.

\bibitem[GT89]{GabowT89}
Harold~N. Gabow and Robert~Endre Tarjan.
\newblock Faster scaling algorithms for network problems.
\newblock {\em {SIAM} J. Comput.}, 18(5):1013--1036, 1989.

\bibitem[GT90]{GoldbergT90}
Andrew~V. Goldberg and Robert~E. Tarjan.
\newblock Finding minimum-cost circulations by successive approximation.
\newblock {\em Math. Oper. Res.}, 15(3):430--466, 1990.

\bibitem[HKRS97]{HenzingerKRS97}
Monika~Rauch Henzinger, Philip~N. Klein, Satish Rao, and Sairam Subramanian.
\newblock Faster shortest-path algorithms for planar graphs.
\newblock {\em J. Comput. Syst. Sci.}, 55(1):3--23, 1997.

\bibitem[HRT15]{HolmRT15}
Jacob Holm, Eva Rotenberg, and Mikkel Thorup.
\newblock Planar reachability in linear space and constant time.
\newblock In {\em {IEEE} 56th Annual Symposium on Foundations of Computer
  Science, {FOCS} 2015}, pages 370--389. {IEEE} Computer Society, 2015.

\bibitem[IKLS17a]{ItalianoKLS17a}
Giuseppe~F. Italiano, Adam Karczmarz, Jakub Lacki, and Piotr Sankowski.
\newblock Decremental single-source reachability in planar digraphs.
\newblock {\em CoRR}, abs/1705.11163, 2017.

\bibitem[IKLS17b]{ItalianoKLS17}
Giuseppe~F. Italiano, Adam Karczmarz, Jakub Lącki, and Piotr Sankowski.
\newblock Decremental single-source reachability in planar digraphs.
\newblock In {\em Proceedings of the 49th Annual {ACM} {SIGACT} Symposium on
  Theory of Computing, {STOC} 2017}, pages 1108--1121, 2017.

\bibitem[IKP21]{ItalianoKP21}
Giuseppe~F. Italiano, Adam Karczmarz, and Nikos Parotsidis.
\newblock Planar reachability under single vertex or edge failures.
\newblock In {\em Proceedings of the 2021 {ACM-SIAM} Symposium on Discrete
  Algorithms, {SODA} 2021}, pages 2739--2758. {SIAM}, 2021.

\bibitem[INSW11]{ItalianoNSW11}
Giuseppe~F. Italiano, Yahav Nussbaum, Piotr Sankowski, and Christian
  Wulff{-}Nilsen.
\newblock Improved algorithms for min cut and max flow in undirected planar
  graphs.
\newblock In {\em Proceedings of the 43rd {ACM} Symposium on Theory of
  Computing, {STOC} 2011}, pages 313--322, 2011.

\bibitem[Iri60]{succsp2}
Masao Iri.
\newblock A new method of solving transportation-network problems.
\newblock 1960.

\bibitem[Jew62]{succsp3}
William~S. Jewell.
\newblock Optimal flow through networks with gains.
\newblock {\em Operations Research}, 10(4):476--499, 1962.

\bibitem[JV83]{JohnsonV83}
Donald~B. Johnson and Shankar~M. Venkatesan.
\newblock Partition of planar flow networks (preliminary version).
\newblock In {\em 24th Annual Symposium on Foundations of Computer Science},
  pages 259--263, 1983.

\bibitem[Kar84]{Karmarkar84}
Narendra Karmarkar.
\newblock A new polynomial-time algorithm for linear programming.
\newblock {\em Comb.}, 4(4):373--396, 1984.

\bibitem[Kar21]{Karczmarz21}
Adam Karczmarz.
\newblock Fully dynamic algorithms for minimum weight cycle and related
  problems.
\newblock In {\em 48th International Colloquium on Automata, Languages, and
  Programming, {ICALP} 2021}, volume 198 of {\em LIPIcs}, pages 83:1--83:20.
  Schloss Dagstuhl - Leibniz-Zentrum f{\"{u}}r Informatik, 2021.

\bibitem[Kle05]{Klein05}
Philip~N. Klein.
\newblock Multiple-source shortest paths in planar graphs.
\newblock In {\em Proceedings of the Sixteenth Annual {ACM-SIAM} Symposium on
  Discrete Algorithms, {SODA} 2005}, pages 146--155, 2005.

\bibitem[KMNS17]{KaplanMNS17}
Haim Kaplan, Shay Mozes, Yahav Nussbaum, and Micha Sharir.
\newblock Submatrix maximum queries in monge matrices and partial monge
  matrices, and their applications.
\newblock {\em {ACM} Trans. Algorithms}, 13(2):26:1--26:42, 2017.

\bibitem[KMS13]{KleinMS13}
Philip~N. Klein, Shay Mozes, and Christian Sommer.
\newblock Structured recursive separator decompositions for planar graphs in
  linear time.
\newblock In {\em Symposium on Theory of Computing Conference, STOC'13}, pages
  505--514, 2013.

\bibitem[KMW10]{KleinMW10}
Philip~N. Klein, Shay Mozes, and Oren Weimann.
\newblock Shortest paths in directed planar graphs with negative lengths: {A}
  linear-space \emph{O}(\emph{n} log\({}^{\mbox{2}}\) \emph{n})-time algorithm.
\newblock {\em {ACM} Trans. Algorithms}, 6(2):30:1--30:18, 2010.

\bibitem[KS19]{KarczmarzS19}
Adam Karczmarz and Piotr Sankowski.
\newblock Min-cost flow in unit-capacity planar graphs.
\newblock In {\em 27th Annual European Symposium on Algorithms, {ESA} 2019},
  volume 144 of {\em LIPIcs}, pages 66:1--66:17. Schloss Dagstuhl -
  Leibniz-Zentrum f{\"{u}}r Informatik, 2019.

\bibitem[KS21]{KarczmarzS21}
Adam Karczmarz and Piotr Sankowski.
\newblock A deterministic parallel {APSP} algorithm and its applications.
\newblock In {\em Proceedings of the 2021 {ACM-SIAM} Symposium on Discrete
  Algorithms, {SODA} 2021}, pages 255--272. {SIAM}, 2021.

\bibitem[KT18]{KrauthgamerT18}
Robert Krauthgamer and Ohad Trabelsi.
\newblock Conditional lower bounds for all-pairs max-flow.
\newblock {\em {ACM} Trans. Algorithms}, 14(4):42:1--42:15, 2018.

\bibitem[{\L}NSW12]{LackiNSW12}
Jakub {\L}ącki, Yahav Nussbaum, Piotr Sankowski, and Christian Wulff{-}Nilsen.
\newblock Single source - all sinks max flows in planar digraphs.
\newblock In {\em 53rd Annual {IEEE} Symposium on Foundations of Computer
  Science, {FOCS} 2012}, pages 599--608, 2012.

\bibitem[LPS21]{LiPS21}
Jason Li, Debmalya Panigrahi, and Thatchaphol Saranurak.
\newblock A nearly optimal all-pairs min-cuts algorithm in simple graphs.
\newblock In {\em 62nd {IEEE} Annual Symposium on Foundations of Computer
  Science, {FOCS} 2021}, pages 1124--1134. {IEEE}, 2021.

\bibitem[LR19]{LahnR19}
Nathaniel Lahn and Sharath Raghvendra.
\newblock A faster algorithm for minimum-cost bipartite matching in minor-free
  graphs.
\newblock In {\em Proceedings of the Thirtieth Annual {ACM-SIAM} Symposium on
  Discrete Algorithms, {SODA} 2019}, pages 569--588, 2019.

\bibitem[Mil84]{Miller84}
Gary~L. Miller.
\newblock Finding small simple cycle separators for 2-connected planar graphs.
\newblock In {\em Proceedings of the 16th Annual {ACM} Symposium on Theory of
  Computing}, pages 376--382. {ACM}, 1984.

\bibitem[MN95]{MillerN95}
Gary~L. Miller and Joseph Naor.
\newblock Flow in planar graphs with multiple sources and sinks.
\newblock {\em {SIAM} J. Comput.}, 24(5):1002--1017, 1995.

\bibitem[MNW18]{MozesNW18}
Shay Mozes, Yahav Nussbaum, and Oren Weimann.
\newblock Faster shortest paths in dense distance graphs, with applications.
\newblock {\em Theor. Comput. Sci.}, 711:11--35, 2018.

\bibitem[MW10]{MozesW10}
Shay Mozes and Christian Wulff{-}Nilsen.
\newblock Shortest paths in planar graphs with real lengths in
  \emph{O}(\emph{n}log\({}^{\mbox{2}}\)\emph{n}/loglog\emph{n}) time.
\newblock In {\em Algorithms - {ESA} 2010, 18th Annual European Symposium},
  pages 206--217, 2010.

\bibitem[Nus14]{nussbaum2014network}
Yahav Nussbaum.
\newblock {\em Network flow problems in planar graphs}.
\newblock PhD thesis, 2014.

\bibitem[SA12]{SharathkumarA12}
R.~Sharathkumar and Pankaj~K. Agarwal.
\newblock Algorithms for the transportation problem in geometric settings.
\newblock In {\em Proceedings of the Twenty-Third Annual {ACM-SIAM} Symposium
  on Discrete Algorithms, {SODA} 2012}, pages 306--317. {SIAM}, 2012.

\bibitem[SB65]{circcycle}
Thomas~L Saaty and Robert~G Busacker.
\newblock {\em Finite graphs and networks: An introduction with applications}.
\newblock McGraw-Hill Book Company, 1965.

\bibitem[ST83]{SleatorT83}
Daniel~Dominic Sleator and Robert~Endre Tarjan.
\newblock A data structure for dynamic trees.
\newblock {\em J. Comput. Syst. Sci.}, 26(3):362--391, 1983.

\bibitem[Tar85]{Tardos85}
{\'{E}}va Tardos.
\newblock A strongly polynomial minimum cost circulation algorithm.
\newblock {\em Combinatorica}, 5(3):247--256, 1985.

\bibitem[vdBCK{\etalchar{+}}23]{det-flow}
Jan van~den Brand, Li~Chen, Rasmus Kyng, Yang~P. Liu, Richard Peng,
  Maximilian~Probst Gutenberg, Sushant Sachdeva, and Aaron Sidford.
\newblock A deterministic almost-linear time algorithm for minimum-cost flow.
\newblock {\em CoRR}, abs/2309.16629, 2023.

\bibitem[vdBLL{\etalchar{+}}21]{BrandLLSS0W21}
Jan van~den Brand, Yin~Tat Lee, Yang~P. Liu, Thatchaphol Saranurak, Aaron
  Sidford, Zhao Song, and Di~Wang.
\newblock Minimum cost flows, mdps, and {\(l\)}\({}_{\mbox{1}}\)-regression in
  nearly linear time for dense instances.
\newblock In {\em {STOC} '21: 53rd Annual {ACM} {SIGACT} Symposium on Theory of
  Computing}, pages 859--869. {ACM}, 2021.

\bibitem[Zha22]{Zhang22}
Tianyi Zhang.
\newblock {Faster Cut-Equivalent Trees in Simple Graphs}.
\newblock In {\em 49th International Colloquium on Automata, Languages, and
  Programming (ICALP 2022)}, volume 229 of {\em Leibniz International
  Proceedings in Informatics (LIPIcs)}, pages 109:1--109:18. Schloss Dagstuhl
  -- Leibniz-Zentrum f{\"u}r Informatik, 2022.

\end{thebibliography}

\end{document}